\documentclass[11pt, a4paper]{amsart}

\usepackage{color}
\usepackage{amsmath}
\usepackage{amssymb}\usepackage{amsthm}
\usepackage{amsfonts, amscd}

\usepackage{hyperref}

\usepackage{url}
\usepackage{fullpage}
\usepackage{mathtools}
\usepackage[all]{xy}
\usepackage{slashed}
\usepackage{multirow}
\theoremstyle{plain}
\newtheorem{thm}{Theorem}[section]
\newtheorem{lem}[thm]{Lemma}
\newtheorem{prop}[thm]{Proposition}

\theoremstyle{definition}
\newtheorem{assum}[thm]{Assumption}

\newtheorem{defn}[thm]{Definition}

\newtheorem{ex}[thm]{Example}
\theoremstyle{remark}
\newtheorem{rem}[thm]{Remark}
\newtheorem{rems}[thm]{Remarks}

\newcommand{\unit}{\mathbb I}

\newcommand{\Ad}{\mathop{\mathrm{Ad}}\nolimits}
\newcommand{\Ind}{\mathop{\mathrm{Ind}}\nolimits}
\newcommand{\ind}{\mathop{\mathrm{ind}}\nolimits}

\newcommand{\cct}{\bbC^2}

\def\proofof[#1]{\trivlist \item[\hskip \labelsep{\bf Proof of #1.\ }]}

\newcommand{\caA}{{\mathcal A}}
\newcommand{\caB}{{\mathcal B}}
\newcommand{\caC}{{\mathcal C}}
\newcommand{\caD}{{\mathcal D}}

\newcommand{\caF}{{\mathcal F}}
\newcommand{\caG}{{\mathcal G}}
\newcommand{\caH}{{\mathcal H}}

\newcommand{\caK}{{\mathcal K}}
\newcommand{\caL}{{\mathcal L}}
\newcommand{\caM}{{\mathcal M}}
\newcommand{\caN}{{\mathcal N}}
\newcommand{\caO}{{\mathcal O}}
\newcommand{\caP}{{\mathcal P}}

\newcommand{\caR}{{\mathcal R}}
\newcommand{\caS}{{\mathcal S}}

\newcommand{\caV}{{\mathcal V}}

\newcommand{\bbC}{{\mathbb C}}

\newcommand{\bbE}{{\mathbb E}}

\newcommand{\bbN}{{\mathbb N}}

\newcommand{\bbP}{{\mathbb P}}

\newcommand{\bbR}{{\mathbb R}}

\newcommand{\bbT}{{\mathbb T}}

\newcommand{\bbZ}{{\mathbb Z}}

\newcommand{\lv}{\left \vert}
\newcommand{\rv}{\right \vert}
\newcommand{\lV}{\left \Vert}
\newcommand{\rV}{\right \Vert}

\newcommand{\lmk}{\left (}
\newcommand{\rmk}{\right )}
\newcommand{\al}{{\mathcal A}}

\newcommand{\id}{\mathop{\mathrm{id}}\nolimits}
\newcommand{\Tr}{\mathop{\mathrm{Tr}}\nolimits}

\newcommand{\Mat}{\mathop{\mathrm{M}}\nolimits}
\newcommand{\Clf}{\mathop{\mathfrak{C}}\nolimits}

\newcommand{\spa}{\mathop{\mathrm{span}}\nolimits}

\newcommand{\rank}{\mathop{\mathrm{rank}}\nolimits}

\newcommand{\nan}{\mathbb{N}}

\newcommand{\braket}[2]{\langle#1,#2\rangle}

\def\proofof[#1]{\trivlist \item[\hskip \labelsep{\bf Proof of #1.\ }]}

\newcommand{\bk}{\caB(\caK)}
\newcommand{\bh}{\caB(\caH)}

\newcommand{\op}{\mathfrak h}
\newcommand{\Oo}{{\Omega_\omega}}
\newcommand{\Ho}{{\caH_\omega}}
\newcommand{\po}{{\pi_{\omega}}}

\newcommand{\diam}[1]{\text{diam}(#1)}

\newcommand{\ut}{\unit_{\bbC^2}}

\newcommand{\C}{\ensuremath{\mathbb{C}}}

\newcommand{\hox}{\,\hat{\otimes}\,} 
\newcommand{\Z}{\ensuremath{\mathbb{Z}}}

\newcommand{\Hom}{\ensuremath{\mathrm{Hom}}}

\newcommand{\frakp}{{\mathfrak p}}

\newcommand{\mpp}{\mathfrak{p}}
\newcommand{\mqq}{\mathfrak{q}}

\def\calK{\mathcal{K}}
\def\calB{\mathcal{B}}
\def\calH{\mathcal{H}}

\def\calF{\mathcal{F}}

\def\calN{\mathcal{N}}
\def\calP{\mathcal{P}}
\def\calM{\mathcal{M}}

\newcommand{\one}{{\unit}}
\newcommand{\ol}{\overline}
\DeclareMathOperator{\Aut}{Aut}

\newcommand{\A}{\mathcal{A}}

\newcommand{\rst}[1]{\ensuremath{{\mathbin\upharpoonright}%
\raise-.5ex\hbox{$#1$}}}

\numberwithin{equation}{section}

\begin{document}
\title{The classification of symmetry protected topological phases of one-dimensional fermion systems}
\author{C. Bourne}
\address{WPI-AIMR, Tohoku University,
2-1-1 Katahira, Aoba-ku, Sendai, 980-8577 \emph{and}
RIKEN iTHEMS, Wako, Saitama 351-0198, Japan}
\email{chris.bourne@tohoku.ac.jp}
\author{Y. Ogata}
\address{Graduate School of Mathematical Sciences, The University of Tokyo, Komaba, Tokyo, 153-8914, Japan}
\email{yoshiko@ms.u-tokyo.ac.jp}

\begin{abstract}
We introduce an index for symmetry protected topological (SPT) 
phases of infinite fermionic chains
with an on-site symmetry
given by a finite group $G$.
This index takes values in
$\mathbb{Z}_2 \times H^1(G,\mathbb{Z}_2) \times H^2(G, U(1)_{\mathfrak{p}})$ with a 
generalized Wall group law under stacking.
We show that this index 
is an invariant of the classification of SPT phases.
When the ground state is translation invariant and
has reduced density matrices with uniformly bounded 
rank on finite intervals, 
we derive a fermionic matrix product representative of this state 
with on-site symmetry.
\end{abstract}
\maketitle
\section{Introduction}\label{intro}

The notion of symmetry protected topological (SPT) phases was introduced by Gu and Wen \cite{GuWen2009}.
We consider the set of all Hamiltonians with a prescribed symmetry 
and which have a unique gapped ground state in the bulk. 
Two Hamiltonians in this set are equivalent if
there is a smooth path within the set connecting them.
By this equivalence relation, we may classify the Hamiltonians 
in this family.
The equivalence class of a Hamiltonian with only on-site interactions is regarded as 
a trivial phase. If a phase is non-trivial, it is called a SPT phase 
 (see also Remark \ref{rk:terminology}).


A basic question is how to show that a given Hamiltonian belongs 
to a SPT phase.
A mathematically natural approach for this problem is to define an invariant of the classification. 
This approach has been studied in the physics literature using 
matrix product states (MPS)~\cite{po, po2, GuWen2009, ChenGuWEn2011, Perez-Garcia2008}.
MPS is a powerful framework introduced in \cite{Fannes:1992vq}, after the discovery of the famous 
AKLT model~\cite{Affleck:1988vr}.
Hastings showed that MPS approximates unique gapped ground states of 
quantum spin chains well \cite{area}. 
However, we can not comprehensively study invariants of the path-connected components of the space of unique 
gapped ground states via MPS only. 
Firstly, MPS are {\it translationally invariant} systems and we would like to define an invariant that 
does not require this assumption. Furthermore, an approximation of a gapped ground state by MPS may 
not be compatible with the path connected components and so is insufficient to define an index in general.
If the index is not defined for {\it all} unique gapped ground states, there is no way to discuss if it is actually 
an invariant or not.

In \cite{OgataTRI,OgataRF,OgataSPT}, an index for SPT phases 
with on-site finite group symmetry and
global reflection symmetry was defined for infinite
quantum spin chains in a fully general setting. 
In these papers, it was proven that the index is actually 
an invariant of the classification of SPT phases.
An important observation for stability of the index is the 
factorization property of the automorphic equivalence.
The key ingredient for the definition of the index is the split property of unique 
gapped ground states, proven by Matsui~\cite{Matsui13}.
The index introduced in~\cite{OgataTRI,OgataRF} generalizes 
the indices introduced for MPS in \cite{po,po2,GuWen2009,ChenGuWEn2011,Perez-Garcia2008},
where a MPS emerges naturally from a translation invariant 
split state whose  reduced density matrix has 
uniformly bounded rank on finite intervals~\cite{BJKW, Matsui13}.

In this paper, we are interested in the analogous problem for fermionic chains with  
on-site finite group symmetries. 
Fermionic SPT phases for finite systems in one dimension have already been extensively studied in the 
physics literature~\cite{fk, FK2, BWHV, KT, KapustinfMPS, TurzilloYou}. 
In contrast to quantum spin chains, for parity-symmetric gapped ground 
states without additional symmetries, there are two distinct phases.
A $\bbZ_2$-index to distinguish these phases in infinite systems 
was introduced in \cite{BSB} and independently in~\cite{Matsui20}.
It was outlined in \cite{BSB} that this $\mathbb{Z}_2$-index is an invariant 
of the classification of unique parity-invariant gapped ground state phases using 
techniques from~\cite{OgataTRI} and~\cite{NSY18}.
The aim of this paper is to extend the analysis of fermionic gapped ground states 
to the case with an on-site symmetry. 
Namely, a classification of one-dimensional fermionic SPT phases.

\subsection{Setting and outline} \label{subsec:SettingOutline}
We assume the reader has some familiarity with the basics of operator algebras 
and their application to quantum statistical mechanical systems, 
see~\cite{BR1, BR2}.
Throughout this paper, we fix $d\in\bbN$. 
Let $\op:=l^2(\bbZ)\otimes \bbC^d$ and $\al$ be the CAR-algebra over $\op$, i.e. the 
universal $C^*$-algebra generated by the identity and $\{a(f)\}_{f\in\op}$ such that 
$f\mapsto a(f)$ is anti-linear and 
\begin{align}
  &\{a(f_1),a(f_2)\} = 0,   &&\{a(f_1),a(f_2)^*\} = \langle f_1,f_2\rangle.
\end{align}
For each subset $X$ of $\bbZ$, we set $\op_X:=l^2(X)\otimes\bbC^d$, and denote by
$\al_X$ the CAR-algebra over $\op_X$.
We naturally regard $\al_X$ as a subalgebra of $\al$.
We also use the notation $\al_R:=\al_{\bbZ_{\ge 0}}$ and $\al_L:=\al_{\bbZ_{<0}}$.
We denote the set of all finite subsets in ${\bbZ}$ by ${\mathfrak S}_{\bbZ}$ and 
set $\al_{\rm loc}:=\bigcup_{X\in {\mathfrak S}_{\bbZ}}\al_X$.
Given a Hilbert space $\mathfrak{K}$, the fermionic Fock space of anti-symmetric tensors 
is denoted by $\calF(\mathfrak K)$.
For a unitary/anti-unitary operator $U$ on $\bbC^d$,
we denote  the second quantization of $U$ on the Fock space $\calF(\C^d)$ by $\mathfrak{\Gamma}(U)$.

By the universality of the CAR-algebra,
for any unitary/anti-unitary $w$ on $\op$, we may define
a linear/anti-linear automorphism $\beta_w$ on
$\al$ such  that $\beta_w(a(f))=a(wf)$,
$f\in\op$.
In particular, for $w=-\unit$, we obtain the parity operator
 $\Theta:=\beta_{-\unit}$.
 For each $X\in {\mathfrak S}_{\bbZ}$,
 $\al_X$ is $\Theta$-invariant.
 We denote the restriction $\Theta\vert_{\al_X}$ by $\Theta_X$.
For $\sigma=0,1$, set of elements $A$ in $\al$ with $\Theta(A)=(-1)^\sigma A$
is denoted by $\al^{(\sigma)}$. Elements in $\al^{(0)}$ are said to be even
and elements in $\al^{(1)}$ are said to be odd.

In this paper, we consider an on-site symmetry given by a finite group $G$. 
We let $\Mat_d$ denote the algebra of $d\times d$ matrices with complex entries 
and consider a projective unitary/anti-unitary representation of $G$ on $\bbC^d$ 
relative to a group homomorphism $\mpp:G\to\bbZ_2$.\footnote{Throughout this 
paper, we use the presentation of $\bbZ_2$ as the additive group $\{0,1\}$.}
That is, there is a projective representation 
$U$ of $G$ on $\bbC^d$ such that 
$U_g$ is unitary if $\mpp(g)=0$ and anti-unitary if ${\mpp(g)}=1$. Because $U$ is 
projective, there is a $2$-cocycle $\upsilon:G\times G\to U(1)$ such that 
$U_g U_h = \upsilon(g,h)U_{gh}$ and 
for all $f,g,h\in G$
\begin{align}
  &\upsilon(e,g)=1=\upsilon(g,e), 
  &&\frac{\ol{\upsilon(g,h)}^{\frakp(f)} \upsilon(f,gh)}{\upsilon(f,g)\upsilon(fg,h)} = 1,
\end{align}
where $\ol{z}^{\frakp(f)} = z$ if $\frakp(f)=0$ and $\ol{z}^{\frakp(f)} = \ol{z}$ if $\frakp(f)=1$. 
For a fixed homomorphism $\mpp$, equivalence classes of such $2$-cocycles give rise to the 
cohomology group $H^2(G, U(1)_\mpp)$.

For a fixed projective unitary/anti-unitary representation $U$ of $G$ on $\bbC^d$ relative 
to $\mpp:G\to \bbZ_2$, we 
can extend this representation to an on-site representation $\bigoplus_\mathbb{Z}U$ on 
$l^2(\mathbb{Z}) \otimes \bbC^d$. We therefore  can define the 
linear/anti-linear automorphism $\alpha$ on $\al$, where
\begin{align}
\alpha_g:=\beta_{\lmk \bigoplus_{\bbZ} U_g \rmk},\quad g\in G.
\end{align}
If ${\mpp(g)}=0$, then
$\alpha_g$ is an automorphism on $\al$ and if ${\mpp(g)}=1$, then
$\alpha_g$ is an anti-linear automorphism on $\al$.
Note that $\alpha$ satisfies
\begin{align}
\alpha_g\circ\Theta=\Theta\circ\alpha_g, \qquad
\alpha_g(\al_X)=\al_X, \quad g\in G, \quad X\in  {\mathfrak S}_{\bbZ}.
\end{align}
For each $g\in G$ and a state $\varphi$ on $\al$, we define a state $\varphi_g$
by $\varphi_g(A)=\varphi\circ\alpha_g(A)$, $A\in\al$ if ${\mpp(g)}=0$,
and by $\varphi_g(A)=\varphi\circ\alpha_g(A^*)$, $A\in\al$ if ${\mpp(g)}=1$.
We say $\varphi$ is $\alpha$-invariant if $\varphi_g=\varphi$ for any $g\in G$.

In the latter half of the paper we also consider space translations
$\beta_{S_x}$, $x\in\bbZ$. Here 
the unitary $S_x$ is given by
$S_x=s_x\otimes \unit_{\bbC^d}$ with $s_x$ the shift by $x$ on $l^2(\bbZ)$.

Throughout this paper,
 for a state $\varphi$ on $\al_{X}$ (with $X$ a subset of $\bbZ$),
  $(\caH_\varphi,\pi_\varphi,\Omega_\varphi)$
 denotes a GNS triple of $\varphi$. When $\varphi$ is $\Theta_{X}$-invariant,
 then $\hat \Gamma_\varphi$ denotes the self-adjoint unitary on $\caH_\varphi$
 defined by 
 $\hat\Gamma_\varphi\pi_\varphi(A)\Omega_\varphi=\pi_\varphi\circ\Theta_X(A)\Omega_\varphi$ for $A\in\al_X$.
If $\varphi$ is $\alpha$-invariant, then
we denote by $\hat\alpha_\varphi$ 
the extension of $\alpha\vert_{\al_X}$ to $\pi_{\varphi}(\al_X)''$.

The mathematical model of a one-dimensional fermionic system is fully specified by 
the interaction $\Phi$.
An interaction is a map $\Phi$ from 
${\mathfrak S}_{\bbZ}$ into ${\caA}_{\rm loc}$ such
that $\Phi(X) \in {\caA}_{X}$ 
and $\Phi(X) = \Phi(X)^*$
for $X \in {\mathfrak S}_{\bbZ}$. 
When we have  $\Theta(\Phi(X))=\Phi(X)$
for all $X\in  {\mathfrak S}_{\bbZ}$,
$\Phi$ is said to be even.
We say $\Phi$ is $\alpha$-invariant if
we have $\alpha_g(\Phi(X))=\Phi(X)$
for all $X\in  {\mathfrak S}_{\bbZ}$ and $g\in G$.
An interaction $\Phi$ is translation invariant if
$\Phi(X+x)=\beta_{S_x}(\Phi(X))$, 
for all $x\in{\mathbb Z}$ and $X\in  {\mathfrak S}_{\bbZ}$.
Furthermore, an interaction $\Phi$
is finite range if there exists an $m\in {\mathbb N}$ such that
$\Phi(X)=0$ for any $X$ with diameter larger than $m$.
We denote by $\caB_{f}^{e}$, 
the set of all finite range even interactions $\Phi$ which satisfy
\begin{align}\label{fi}
 \sup_{X\in {\mathfrak S}_{\bbZ}}
\lV
\Phi\lmk X\rmk
\rV<\infty.
\end{align}


For an interaction $\Phi$ and a finite set $\Lambda\in{\mathfrak S}_{\bbZ}$, we define the local Hamiltonian on $\Lambda$ by
\begin{equation}\label{GenHamiltonian}
H_{\Lambda,\Phi}:=\sum_{X\subset{\Lambda}}\Phi(X).
\end{equation}
The dynamics given by this local Hamiltonian is denoted by
\begin{align}\label{taulamdef}
\tau_t^{\Phi,\Lambda}\lmk A\rmk:= e^{itH_{\Lambda,\Phi}} Ae^{-itH_{\Lambda,\Phi}},\quad
t\in\bbR,\quad A\in\al.
\end{align}
If $\Phi$ belongs to $\caB_{f}^e$,
the limit
\begin{align}\label{taudef}
\tau_t^{\Phi}\lmk A\rmk=\lim_{\Lambda\to\bbZ}
\tau_t^{\Phi,\Lambda}\lmk A\rmk
\end{align}
exists for each $A\in \caA$ and $t\in{\mathbb R}$, 
and defines a strongly continuous one parameter group of automorphisms $\tau^\Phi$ on $\caA$ 
(see Appendix \ref{flr}). 
We denote the generator of $\tau^{\Phi}$ by $\delta_{\Phi}$.

For $\Phi\in\caB_f^{e}$, a state $\varphi$ on $\caA$ is called a \mbox{$\tau^{\Phi}$-ground} state
if the inequality
$
-i\,\varphi(A^*{\delta_{\Phi}}(A))\ge 0
$
holds
for any element $A$ in the domain $\caD({\delta_{\Phi}})$ of ${\delta_\Phi}$.
If $\varphi$ is a $\tau^\Phi$-ground state with GNS triple $(\caH_\varphi,\pi_\varphi,\Omega_\varphi)$, 
then there exists a unique positive operator $H_{\varphi,\Phi}$ on $\caH_\varphi$ such that
$e^{itH_{\varphi,\Phi}}\pi_\varphi(A)\Omega_\varphi=\pi_\varphi(\tau_t^\Phi(A))\Omega_\varphi$,
for all $A\in\caA$ and $t\in\mathbb R$.
We call this $H_{\varphi,\Phi}$ the bulk Hamiltonian associated with $\varphi$.
Note that $\Omega_\varphi$ is an eigenvector of $H_{\varphi,\Phi}$ with eigenvalue $0$.
The following definition clarifies what we mean by a model with a unique gapped ground state.
\begin{defn}
We say that a model with an interaction $\Phi\in\caB_f^{e}$ has a unique gapped ground state if 
(i)~the $\tau^\Phi$-ground state, which we denote as $\varphi$, is unique, and 
(ii)~there exists a $\gamma>0$ such that
$\sigma(H_{\varphi,\Phi})\setminus\{0\}\subset [\gamma,\infty)$, where  $\sigma(H_{\varphi,\Phi})$ is the spectrum of $H_{\varphi,\Phi}$.
\end{defn}
Note that the uniqueness of $\varphi$ implies that 0 is a non-degenerate eigenvalue of $H_{\varphi,\Phi}$.

If $\varphi$ is a \mbox{$\tau^{\Phi}$-ground} state of $\alpha$-invariant 
and $\Theta$-invariant 
interaction $\Phi\in {\caB}_f^e$, 
then $\varphi \circ \Theta$ and $\varphi_g$ is also a
\mbox{$\tau^{\Phi}$-ground} state for each $g\in G$.
In particular, if $\varphi$ is a unique \mbox{$\tau^{\Phi}$-ground} state, it is pure, 
$\Theta$-invariant and  $\alpha$-invariant.
We denote by $\caG_{f}^{e,\alpha}$ 
the set of all $\alpha$-invariant interactions $\Phi\in\caB_f^{e}$
with a unique gapped ground state.

Now the classification problem of SPT phases is
the classification of $\caG_{f}^{e,\alpha}$ with respect to the following equivalence relation:
$\Phi_0, \Phi_1\in \caG_{f}^{e,\alpha}$ are equivalent if
there is a smooth path in $\caG_{f}^{e,\alpha}$ connecting them.
(See Section \ref{stability} for a more precise definition.)

We now outline the main results of the paper. 
In Section \ref{indsec}, we introduce an index for $\Theta$-invariant 
and $\alpha$-invariant gapped ground states in a fully general setting. 
This index  takes value in
$\bbZ_2\times H^1(G,\bbZ_2)\times H^2(G, U(1)_{\mpp})$,
which is analogous to the indices introduced in~\cite{KT} in the context of spin-TQFT 
and~\cite{BWHV,KapustinfMPS,TurzilloYou} for the fermionic MPS setting. When $G$ is trivial, 
the index is $\bbZ_2$-valued and recovers the index studied in~\cite{BSB, Matsui20}.
The key ingredient for the definition is again 
the split property of unique gapped ground states for fermionic systems proven 
recently in \cite{Matsui20}.
 In Section \ref{stability}, we show that our defined index is 
 an invariant of the classification, i.e., it is stable along the smooth path in
 $\caG_{f}^{e,\alpha}$.

Because our index takes values in a group, it suggests that one may compose fermionic 
SPT phases to obtain a new phase with index determined from the original systems. 
In the physics literature, this is achieved by stacking of systems, see~\cite{FK2, TurzilloYou} 
for example. Mathematically this operation corresponds to a (graded) tensor product of 
ground states. In Section \ref{stacksec}, we show that our index is indeed closed under this 
tensor product operation. However, despite the notation, the group operation on 
$\bbZ_2\times H^1(G,\bbZ_2)\times H^2(G, U(1)_{\mpp})$ is \emph{not} the direct sum, but 
rather a twisted product that follows a generalized Wall group law, cf.~\cite{Wall}. 
As an example, we consider the case of an anti-linear $\bbZ_2$-action (e.g. an on-site time-reversal 
symmetry) and show that our index takes values in $\bbZ_8$. This recovers the $\bbZ_8$-classification 
of time-reversal symmetric one-dimensional fermionic SPT phases noted in~\cite{fk,FK2} and 
extends this classification to infinite systems.

In Sections \ref{transsec} and \ref{fmpssec}, we consider the 
unique ground state of a translation invariant 
$\Phi\in \caG_{f}^{e,\alpha}$.
For quantum spin systems, it is known that a representation of 
Cuntz algebra emerges from translation invariant pure split states \cite{BJKW,Matsui3}. 
The generators of this Cuntz algebra representation 
give an operator product representation of the state and also implement the space translation.
We find an analogous object for fermionic systems in Section \ref{transsec}. 
Because odd elements with disjoint support anti-commute in the CAR-algebra, 
the operator product representation and space translation 
is more complicated than the quantum spin chain setting.
The results of Section \ref{transsec} are then applied to the study of 
fermionic MPS in Section \ref{fmpssec}. When the rank of the reduced
 density matrices of the infinite volume ground state is uniformly bounded, 
 we show that the ground state has a presentation as a fermionic MPS with 
 on-site symmetry. 
 We then show that our index agrees with and therefore extends the indices 
 defined for fermionic MPS with an on-site symmetry in~\cite{BJKW, KapustinfMPS, TurzilloYou}.
 
Basic properties of graded von Neumann algebras are reviewed 
in Appendix \ref{sec:GradedvN}. In Appendix \ref{flr}, we adapt the Lieb--Robinson bound 
to the setting of lattice fermions (see also~\cite{BSP, NSY17}).


\begin{rem}[A note on terminology] \label{rk:terminology}
For the sake of clarity, let us more carefully specify the characterization of a SPT phase used in this paper. 
Given a $G$-symmetric unique gapped ground state, a SPT phase is often defined to be 
an equivalence class of ground states which can be connected to a ground state from an on-site 
interaction, but which can {not} be connected $G$-equivariantly. 
In this paper, we define a $\Z_2 \times H^1(G, \Z_2)\times H^2(G, U(1)_\mpp)$-valued invariant for {\it any} 
unique gapped ground state of a one-dimensional fermionic interaction and do {\it not} assume 
the ground state can be connected to a ground state from an on-site interaction without symmetry. 
\end{rem}


\section{The index of fermionic SPT phases}\label{indsec}

\subsection{Graded von Neumann algebras and dynamical systems} \label{Subsec:graded_prelims}
In order to introduce the index, we first need to introduce 
type I central balanced graded $W^*$-$(G,\mpp)$-dynamical systems.  
Further details on graded von Neumann algebras can be found in 
Appendix \ref{sec:GradedvN}.

\begin{defn}
A graded von Neumann algebra is a pair $(\caM,\theta)$ with 
$\caM$ a von Neumann algebra $\theta$ an involutive automorphism on $\caM$, 
$\theta^2 = \mathrm{Id}$.
 If $\caM \subset \caB(\caH)$ and there is a self-adjoint unitary
$\Gamma$ on $\caH$ such that $\mathrm{Ad}_{\Gamma}|_\caM = \theta$, then
we call $(\caM,\theta)$ a spatially graded von Neumann algebra acting with
grading operator $\Gamma$. 
If $\theta$ is the identity automorphism, then 
we say that $(\caM, \theta)$ is trivially graded.
\end{defn}
We say that a graded von Neumann algebra $(\caM,\theta)$ is 
 of type $\lambda$, $\lambda \in \{\mathrm{I, II, III}\}$, if $\caM$ is type $\lambda$.

Given a graded von Neumann algebra $(\caM,\theta)$,
 $\caM$ is a direct sum of two self-adjoint
 $\sigma\mathrm{\hbox{-}weak}$-closed linear subspaces as $\caM=\caM^{(0)}\oplus \caM^{(1)}$,
 where
\begin{align}
\caM^{(\sigma)}:=
\left\{
x\in\caM\mid 
\theta (x)=(-1)^\sigma x
\right\},\quad x\in\caM, \,\, \sigma \in \{0,1\}.
\end{align}
An element of $\caM^{(\sigma)}$ is said to be homogeneous of degree $\sigma$,
or even/odd for $\sigma=0$/$\sigma=1$, respectively.
For a homogeneous $x\in\caM$, its degree is denoted by $\partial x$.
For graded von Neumann algebras $(\caM_1,\theta_1)$, $(\caM_2,\theta_2)$,
a homomorphism $\phi: \caM_1\to\caM_2$ is a graded homomorphism if
$\phi\big(\caM_1^{(\sigma)}\big) \subset \caM_2^{(\sigma)}$ for $\sigma=0,1$. 

\begin{defn}
Let $(\caM, \theta)$ be a graded von Neumann algebra.
We say $(\caM,\theta)$
 is balanced if $\caM$ contains an odd self-adjoint unitary. 
If $Z(\caM)\cap \caM^{(0)}=\bbC\unit$ for
 the center $Z(\caM)$ of $\caM$, we say $(\caM, \theta)$ is central.
\end{defn}

We now consider dynamics on graded von Neumann algebras via a 
linear/anti-linear group action.
\begin{defn}
Let $G$ be a finite group and $\mpp: G\to\bbZ_2$ be a group homomorphism.
A graded $W^*$-$(G,\mpp)$-dynamical system $(\caM,\theta,{\hat\alpha})$ is 
a graded von Neumann algebra $(\caM, \theta)$ with
an action ${\hat\alpha}$ of $G$ on $\caM$
such that
${\hat\alpha}_g$ is a linear automorphism if ${\mpp(g)}=0$ and
${\hat\alpha}_g$ is an anti-linear automorphism if ${\mpp(g)}=1$, 
satisfying ${\hat\alpha}_g\circ\theta=\theta\circ {\hat\alpha}_g$.
\end{defn}

%

We consider some key examples that will play an important role in defining our index. 
We fix a group homomorphism $\mpp:G \to \bbZ_2$ and consider projective 
unitary/anti-unitary representations $V$ of $G$ relative to $\mpp$ (see 
Section  \ref{subsec:SettingOutline} for the definition).

\begin{ex}[$\caR_{0,\caK},\Ad_{\Gamma_\caK},\Ad_{V_g}$]\label{ichi}
Let $\caK$ be a Hilbert space and set $\Gamma_{\caK}:=\unit_{\caK}\otimes \sigma_z$,
a self-adjoint unitary on $\caK\otimes \bbC^2$.\footnote{In this article we use the following notation of Pauli matrices
\begin{align*}
\sigma_x:=\begin{pmatrix}
0&1\\1&0
\end{pmatrix}, \qquad
\sigma_y:=
\begin{pmatrix}
0&-i\\i&0
\end{pmatrix}, \qquad
\sigma_z:=
\begin{pmatrix}
1&0\\0&-1
\end{pmatrix}.
\end{align*}}
We set $\caR_{0,\caK}:=\caB(\caK)\otimes \Mat_2$ and so 
$(\caR_{0,\caK}, \Ad_{\Gamma_\caK})$ is a spatially graded von Neumann algebra
acting on  $\caK\otimes \bbC^2$ with grading operator $\Gamma_\caK$.
Let 
$V$ be a projective unitary/anti-unitary  representation of $G$ on $\caK\otimes \bbC^2$
relative to $\mpp$. 
We also assume that there is a homomorphism  $\mqq:G\to \bbZ_2$ 
such that $\Ad_{V_g}(\Gamma_\caK)=(-1)^{{\mqq(g)}}\Gamma_\caK$.
We then obtain 
a graded $W^*$-$(G,\mpp)$-dynamical system $(\caR_{0,\caK},\Ad_{\Gamma_\caK},\Ad_{V_g})$.
\end{ex}
We denote the set of all $W^*$-$(G,\mpp)$-dynamical systems of the form in 
Example \ref{ichi} by $\caS_0$.

\begin{ex}[$\caR_{1,\caK},\Ad_{\Gamma_\caK},\Ad_{V_g}$]\label{ni}
Let $\caK$ be a Hilbert space and set $\Gamma_\caK:=\unit_{\caK}\otimes \sigma_z$. 
Let $\Clf$ be the subalgebra of $\Mat_2$ generated by $\sigma_x$ and set 
$\caR_{1,\caK}:=\caB(\caK)\otimes \Clf$.\footnote{We may regard $\Clf$ as 
Clifford algebra $\bbC l_1$ generated by $e_1:=\sigma_x$.}
Then $(\caR_{1,\caK}, \Ad_{\Gamma_\caK})$ is a spatially graded von Neumann algebra
acting on  $\caK\otimes \bbC^2$ with grading operator $\Gamma_\caK$.
Let 
$V$ be a projective unitary/anti-unitary representation of $G$ 
relative to $\mpp$
 such that
$\Ad_{V_g}(\unit_{\caK}\otimes \sigma_x)= (-1)^{{\mqq(g)}}\lmk \unit_\caK\otimes \sigma_x\rmk$ 
and
$\Ad_{V_g}(\Gamma_\caK)=(-1)^{{\mqq(g)}}\Gamma_\caK$
for $\mqq:G\to \bbZ_2$ a group homomorphism.
These assumptions imply that $\Ad_{V_g}\lmk \caR_{1,\caK}\rmk=\caR_{1,\caK}$ 
and so $(\caR_{1,\caK},\Ad_{\Gamma_\caK},\Ad_{V_g})$
is a graded $W^*$-$(G,\mpp)$-dynamical system.
\end{ex}
We denote the set of all $W^*$-$(G,\mpp)$-dynamical systems of the form of 
Example \ref{ni} by $\caS_1$.
Given a $W^*$-$(G,\mpp)$-dynamical systems in $\caS_1$, we can construct 
a projective representation of $G$ on $\caK$ from the projective 
representation on $\caK \otimes \bbC^2$.

We first establish some notation. Let $C$ be the complex conjugation on $\cct$ 
with respect to the standard basis. Given two group homomorphisms 
$\mqq_1, \, \mqq_2 \in \Hom(G,\bbZ_2) \cong H^1(G, \bbZ_2)$, we can define a 
group $2$-cocyle, 
\begin{align}\label{enn}
\epsilon({\mqq}_1,{\mqq}_2)(g,h)
=(-1)^{{\mqq_1(g)}{\mqq}_2(h)},\quad g,h\in G.
\end{align}
\begin{rem}\label{homo}
Note that $ [\epsilon ({\mqq_1},\mqq_2)] = [\epsilon ({\mqq_2},\mqq_1)] \in H^2(G,U(1)_{\mqq_1})$.
\end{rem}

\begin{lem}\label{zv}
For $(\caR_{1,\caK},\Ad_{\Gamma_\caK},\Ad_{V_g})\in\caS_1$, there is a
unique projective unitary/anti-unitary representation $V^{(0)}$ of $G$ on $\caK$ 
relative to $\mpp$ such that
$V_g=V_g^{(0)}\otimes C^{{\mpp(g)}}\sigma_y^{\mqq(g)}$. 
If $[\tilde\upsilon]$ and $[\upsilon]$ are the second cohomology classes  
associated to $V$ and $V^{(0)}$ respectively, then 
$[\tilde{\upsilon}] = [\upsilon \, \epsilon(\mqq,\mpp)] \in H^2(G, U(1)_\mpp)$.
\end{lem}
\begin{proof}
Because $\Ad_{V_g} \circ \mathrm{Ad}_{\Gamma_\caK} = \mathrm{Ad}_{\Gamma_\caK}\circ \Ad_{V_g}$, 
we have
$\Ad_{V_g}(\caB(\caK)\otimes\bbC\unit_{\bbC^2})
=
\caB(\caK)\otimes\bbC\unit_{\bbC^2}
$.
Therefore, $\Ad_{V_g}$ induces a linear/anti-linear $*$-automorphism
on $\bk$. Applying Wigner's Theorem, there is a unitary/anti-unitary
$\tilde V_g^{(0)}$ on $\caK$ such that 
\begin{align}
\Ad_{V_g}\lmk x\otimes \unit_{\cct}\rmk
=\Ad_{\tilde V_g^{(0)}} (x)\otimes \unit_{\bbC^2},\quad x\in\bk.
\end{align}
It is clear that $\tilde V^{(0)}$ gives a unitary/anti-unitary projective representation 
relative to $\mpp$.
Note that $V_g^*\big(\tilde V_g^{(0)}\otimes C^{{\mpp(g)}}\sigma_y^{{\mqq(g)}}\big)$
is a unitary which commutes with $\bk\otimes\bbC\unit_{\cct}$, 
$\unit_\caK\otimes \sigma_x$, $\unit_\caK\otimes \sigma_z$ and therefore 
commutes with $\bk\otimes\Mat_2$.
Therefore, there is a $c(g)\in \bbT$ such that
$V_g=c(g)\lmk\tilde V_g^{(0)}\otimes C^{{\mpp(g)}}\sigma_y^{{\mqq(g)}}\rmk$.
Setting $V_g^{(0)}:=c(g)\tilde V_g^{(0)}$, we obtain
$V_g=V_g^{(0)}\otimes C^{{\mpp(g)}}\sigma_y^{\mqq(g)}$.
Clearly $V^{(0)}$ satisfies the required conditions.
Because $\sigma_y^{{\mqq(g)}} C^{{\mpp(h)}}=(-1)^{{\mqq(g)}{\mpp(h)}} C^{{\mpp(h)}}\sigma_y^{{\mqq(g)}}$, 
we obtain the last statement.
\end{proof}
We introduce the following equivalence relation on graded $W^*$-$(G,\mpp)$-dynamical systems.
\begin{defn}
Let $G$ be a finite group and $\mpp: G\to\bbZ_2$ be a group homomorphism.
We say that two graded $W^*$-$(G,\mpp)$-dynamical systems
$(\caM_1, \theta_1, {\hat\alpha}^{(1)})$,
$(\caM_2, \theta_2, {\hat\alpha}^{(2)})$
are equivalent and write $(\caM_1, \theta_1, {\hat\alpha}^{(1)})\sim(\caM_2, \theta_2, {\hat\alpha}^{(2)})$
if
there is a $*$-isomorphism $\iota: \caM_1\to\caM_2$
such that 
\begin{align}
&\iota\circ{\hat\alpha}_g^{(1)}={\hat\alpha}_g^{(2)}\circ\iota,\quad g\in G\label{eone}\\
&\iota\circ\theta_1=\theta_2\circ\iota\label{etwo}.
\end{align}
\end{defn}
Clearly, this is an equivalence relation.

Using equivalence of $W^*$-$(G,\mpp)$-dynamical systems, we can reduce all 
type I balanced central graded $W^*$-$(G,\mpp)$-dynamical systems to the case of either
Example \ref{ichi} or {\ref{ni}}.
\begin{prop}\label{casebycase}
Let $(\caM,\theta,{\hat\alpha})$ be a graded $W^*$-$(G,\mpp)$-dynamical systems
with $(\caM,\theta)$ balanced, central and type I.
Then there is a $\kappa\in\bbZ_2$ and $(\caR_{\kappa, \caK},\Ad_{\Gamma_\caK},\Ad_{V_g})\in\caS_{\kappa}$
such that
$(\caM,\theta,{\hat\alpha})\sim (\caR_{\kappa, \caK},\Ad_{\Gamma_\caK},\Ad_{V_g})$.
\end{prop}

\begin{proof}
Because $(\caM,\theta)$ is central, by Lemma \ref{jh} 
either $\caM$ is a factor 
or $Z(\caM)$
has an odd self-adjoint unitary $b\in Z(\caM)\cap \caM^{(1)}$ such that
\begin{align}\label{bb}
Z(\caM)\cap \caM^{(1)}=\bbC b.
\end{align}
We set $\kappa=0$ for the former case, and $\kappa=1$ for the latter case.

\noindent{\bf (Case: $\kappa=0$)}
Suppose $\caM$ is a type I factor. Because $(\caM,\theta)$ is balanced,
there is an odd self-adjoint unitary $U\in \caM^{(1)}$.

We claim that $\caM^{(0)}$ is not a factor.
If $\caM^{(0)}$ is a factor, by Lemma \ref{gradei}, it is of type I. 
Note then that $\Ad_U\vert_{\caM^{(0)}}$ is an automorphism on the type I factor $\caM^{(0)}$.
By Wigner's Theorem, there is a unitary $u\in \caM^{(0)}$ such that
$\Ad_U(x)=\Ad_u(x)$, $x\in \caM^{(0)}$.
Therefore, $u^*U\in \lmk {\caM^{(0)}}\rmk'$. At the same time, $u^*U$ commutes with $U$ because
$\Ad_{U}(u^*)=\Ad_u(u^*)=u^*$ for $u\in \caM^{(0)}$.
Hence $u^*U\in \caM'\cap \caM=\bbC\unit$. This is a contradiction because 
$u^*U$ is non-zero and odd.
Hence we conclude that $\caM^{(0)}$ is not a factor.

Therefore, there is a projection $z$ in $Z(\caM^{(0)})$ which is not
$0$ nor $\unit$.
For such a projection, we have $z+\Ad_{U}(z)\in\caM\cap \lmk {\caM^{(0)}}\rmk'\cap\{U\}'=Z(\caM)=\bbC \unit$, 
which then implies that $z+\Ad_{U}(z)=\unit$.
(We note
that for orthogonal projections 
$p,q$ satisfying $p+q=t\unit$ with $t\in\bbR$, either $p+q=\unit$ or $p=0,\,\unit$ holds,
by considering the spectrum of $p=t\unit-q$.)

We claim $Z(\caM^{(0)})=\bbC z+\bbC\unit$. 
Now, for any projection $s$ in $Z(\caM^{(0)})$, 
$zs$ is a projection in $Z(\caM^{(0)})$. Therefore either $zs=0$ or $zs+\Ad_{U}(zs)=\unit$.
The latter is possible only if $zs=z$ because $z+\Ad_{U}(z)=\unit$. 
Similarly, we have $(\unit-z)s=0$ or $(\unit-z)s=\unit-z$.
Hence we have $Z(\caM^{(0)})=\bbC z+\bbC\unit$, proving the claim.

Combining this with $\Ad_U(z)=\unit-z$, 
$\caM^{(0)}$ is a direct sum of two same-type factors $\caM^{(0)}z$ and $\caM^{(0)}(\unit-z)$.
Applying Lemma \ref{gradei}, we see that $\caM^{(0)}$ is of type I, and 
$\caM^{(0)}z$ and $\caM^{(0)}(\unit-z)$ are type I factors.

Set $\Gamma:=z-(\unit-z)$.
Note that $\Ad_\Gamma$ and $\theta$ are identity on $\caM^{(0)}$.
We also have $\Ad_U(\Gamma)=(\unit-z)-z=-\Gamma$, hence $\Ad_\Gamma(U)=-U=\theta(U)$.
Therefore, we get
\begin{align}\label{tgr}
\theta(x)=\Ad_\Gamma(x),\quad x\in \caM.
\end{align}

Next we claim that there is a Hilbert space $\caK$ and a $*$-isomorphism 
$\iota:\caM\to \caB(\caK)\otimes \Mat_2$
such that
\begin{align}\label{clclcl}
\iota\circ \theta=\Ad_{\Gamma_\caK}\circ \iota,\quad\text{and}\quad
\iota(\Gamma)=\unit_\caK\otimes \sigma_z=:\Gamma_\caK.
\end{align}
As $\caM$ is a type I factor, there is a Hilbert space $\hat \caK$ and 
a $*$-isomorphism $\hat\iota:\caM\to\caB(\hat\caK)$.
Let $\hat\iota(\Gamma)=Q_0-Q_1$ be the spectral decomposition of
a self-adjoint unitary $\hat\iota(\Gamma)$, with orthogonal projections $Q_0,Q_1$,
corresponding to eigenvalues $1,-1$.
As we have $\Ad_{\hat\iota(\Gamma)}\circ \hat\iota(x)=\hat\iota\circ\Ad_\Gamma(x)=\hat\iota\circ\theta(x)$
for $x\in\caM$ by \eqref{tgr}, we have
$
\hat\iota(\caM^{(0)})=\caB(Q_0\hat\caK)\oplus\caB(Q_1\hat\caK)$.
Because $\Ad_\Gamma(U)=-U$, we have
$\Ad_{\hat\iota(U)}\lmk\hat \iota(\Gamma)\rmk=-\hat\iota(\Gamma)$.
From the spectral decomposition,
we then have $\Ad_{\hat\iota(U)}(Q_0)=Q_1$ and $\Ad_{\hat\iota(U)}(Q_1)=Q_0$.
We therefore see that $v:=Q_0 \hat\iota (U)Q_1$ is a unitary from $Q_1\hat\caK$ onto
$Q_0\hat\caK$.
We set $\caK:=Q_0\hat\caK$ and define a unitary $W:\hat \caK\to \caK\otimes \bbC^2$ by
\begin{align}
W\begin{pmatrix}
\xi_0\\\xi_1
\end{pmatrix}
=\xi_0\otimes e_0+v\xi_1\otimes e_1,\quad\xi_0\in Q_0\hat \caK,\quad\xi_1\in Q_1\hat\caK.
\end{align}
Here $\{e_0,e_1\}$ is the standard basis of $\bbC^2$.
Note that $\Ad_W\circ\hat\iota(\Gamma)=\unit_\caK\otimes\sigma_z=\Gamma_{\caK}$.
Then $\iota:=\Ad_W\circ\hat\iota: \caM\to \caB(\caK)\otimes \Mat_2$
is a $*$-isomorphism satisfying \eqref{clclcl}, proving the claim.

Next we consider the action of $G$.
Because $Z(\caM^{(0)})=\bbC z+\bbC(\unit-z)$, 
$\Gamma=z-(\unit-z)$ and $-\Gamma=-z+(\unit-z)$ are the only self-adjoint unitaries
in $Z(\caM^{(0)})\setminus \bbC\unit$.
As ${\hat\alpha}_g$ preserves $\caM^{(0)}$, 
 ${\hat\alpha}_g(\Gamma)$ is a self-adjoint unitary in $Z(\caM^{(0)})\setminus \bbC\unit$ 
 and so 
${\hat\alpha}_g(\Gamma)=(-1)^{{\mqq(g)}}\Gamma$
for ${\mqq(g)}=0$ or ${\mqq(g)}=1$.
Clearly, ${\mqq}:G\to\bbZ_2$ is a group homomorphism.

Because $\iota\circ{\hat\alpha}_g\circ \iota^{-1}$ is a linear/anti-linear automorphism on
 $\caB(\caK)\otimes\Mat_2$, by Wigner's Theorem 
there is a projective representation $V$ satisfying 
\begin{align}\label{vdef}
\Ad_{V_g} (x)=\iota\circ{\hat\alpha}_g\circ \iota^{-1}(x),\quad x\in \caB(\caK)\otimes \Mat_2,\quad
g\in G,
\end{align}
and where $V_g$ is unitary/anti-unitary depending on $\mpp(g)$.
Because ${\hat\alpha}_g(\Gamma)=(-1)^{{\mqq(g)}}\Gamma$, we have
\begin{align}
\Ad_{V_g}\lmk\Gamma_\caK\rmk
=(-1)^{{\mqq(g)}} \Gamma_\caK,\quad g\in G.\label{vpar}
\end{align}
Hence we obtain
$
(\caR_{0,\caK},\Ad_{\Gamma_\caK},\Ad_{V_g})\in\caS_0
$.
By \eqref{clclcl} and \eqref{vdef}, we also have
$(\caM,\theta,{\hat\alpha})\sim (\caR_{0,\caK},\Ad_{\Gamma_\caK},\Ad_{V_g})$.

\vspace{0.1cm}

\noindent{\bf (Case: $\kappa=1$)}
Suppose that $\caM$ has a self-adjoint unitary $b\in Z(\caM)\cap \caM^{(1)}$ satisfying \eqref{bb}.
Set $P_\pm :=\frac{1\pm b}{2}$, where $P_\pm$ are orthogonal projections in $Z(\caM)$ such that
$P_++P_-=\unit$.
By \eqref{bb}, $Z(M)=\bbC b +\bbC\unit=\bbC P_++\bbC P_-$.
As $\caM$ is type I,
$\caM$ is a direct sum of the
type I factors, $\caM P_+$ and $\caM P_-$.

We claim that $\caM^{(0)}$ is a type I factor.
For any $x\in Z\lmk \caM^{(0)}\rmk$,
we have $x\in \caM^{(0)}\cap \lmk \caM^{(0)}\rmk'\cap\{b \}'
=\caM^{(0)}\cap\caM'=Z(\caM)\cap\caM^{(0)}=\bbC\unit$,
because $b$ is a self-adjoint unitary in
$Z(\caM)\cap \caM^{(1)}$.
Hence $Z(\caM^{(0)})=\bbC\unit$ and by Lemma \ref{gradei}, $\caM^{(0)}$ is a type I factor.

Next we claim that 
there is a Hilbert space $\caK$ and a $*$-isomorphism 
$\iota:\caM\to \caB(\caK)\otimes \Clf$
such that
\begin{align}\label{clc}
\iota\circ\theta=\Ad_{\Gamma_\caK}\circ\iota,
\quad
\iota(b)=\unit_\caK\otimes \sigma_x,
\end{align}
for $\Gamma_\caK=\unit_\caK\otimes \sigma_z$.
(Recall Example \ref{ni} for $\Clf$.)
Because
$\caM^{(0)}$ is a type I factor, there is a Hilbert space $\caK$ and a $*$-isomorphism
$\iota_0: \caM^{(0)}\to \caB(\caK)$.
As $\caM=\caM^{(0)}\oplus \caM^{(0)}b$, 
we may define a linear map $\iota: \caM\to \bk\otimes \Clf$ by
\begin{align}
\iota(x+yb):=\iota_0(x)\otimes \unit +\iota_0(y)\otimes \sigma_x,\quad
x,y\in\caM^{(0)}.
\end{align}
It can be easily checked that $\iota$ is a $*$-isomorphism satisfying \eqref{clc}.

Now we consider the group action.
Because $Z(\caM)\cap \caM^{(1)}=\bbC b$, 
$b$ and $-b$ are the only self-adjoint unitaries
in $Z(\caM)\cap\caM^{(1)}$.
As ${\hat\alpha}_g$ commutes with the grading automorphism,
${\hat\alpha}_g(b)$ is a self-adjoint unitary
in $Z(\caM)\cap\caM^{(1)}$.
Therefore, 
${\hat\alpha}_g(b)=(-1)^{{\mqq(g)}}b$
with ${\mqq}:G\to\bbZ_2$ a group homomorphism.

Because ${\hat\alpha}_g(\caM^{(0)})=\caM^{(0)}$ and $\iota(\caM^{(0)})=\bk\otimes\bbC\unit$
by \eqref{clc},
$\iota\circ{\hat\alpha}_g\circ\iota^{-1} $ induces
 a linear/anti-linear automorphism on  $\bk$ which is 
 implemented by a unitary/antiunitary 
 $V_g^{(0)}$ on $\caK$ by Wigner's Theorem. That is,
 \begin{align}
 \iota\circ{\hat\alpha}_g\circ\iota^{-1}(a\otimes \unit_{\cct})
 =\Ad_{V_g^{(0)}}( a)\otimes \unit_{\bbC^2},\quad
 a\in \bk,\quad g\in G.
 \end{align}
with $V^{(0)}$ a projective unitary/anti-unitary representation of $G$ on $\caK$ 
relative to $\mpp$.
 Set
 $ V_g:= V_g^{(0)}\otimes  C^{{\mpp(g)}}\sigma_y^{{\mqq(g)}}$,
 with the complex conjugation $C$ on $\bbC^2$
 with respect to the standard basis.
 Clearly $V$ is also a projective unitary/anti-unitary representation of $G$ on $\caK\otimes\bbC^2$ 
 relative to $\mpp$.
 We then have
 \begin{align}
\Ad_{V_g} (a\otimes \ut)
&=\iota\circ{\hat\alpha}_g\circ \iota^{-1}(a\otimes\ut),\quad a\in \caB(\caK),\\
 \Ad_{V_g}(\unit_\caK\otimes \sigma_x)
 &=(-1)^{{\mqq(g)}}(\unit_\caK\otimes \sigma_x)
 =\iota\circ{\hat\alpha}_g (b)
 =\iota\circ{\hat\alpha}_g\circ \iota^{-1}(\unit_\caK\otimes \sigma_x).
 \end{align} 
 Combining these identities, we obtain
 \begin{align}\label{vvve}
 \Ad_{V_g}\circ\iota(x)=\iota\circ{\hat\alpha}_g(x),\quad x\in\caM. 
 \end{align}
 We also have
 \begin{align}
 \Ad_{V_g}(\Gamma_\caK)
 =(-1)^{{\mqq(g)}}\Gamma_\caK.
 \end{align}
Hence we obtain $(\caR_{1,\caK},\Ad_{\Gamma_\caK},\Ad_{V_g})\in\caS_1$ 
such that 
$(\caM,\theta,{\hat\alpha})\sim (\caR_{1,\caK},\Ad_{\Gamma_\caK},\Ad_{V_g})$.
\end{proof}
\begin{defn}\label{index}
Let $(\caM,\theta,{\hat\alpha})$ be a graded $W^*$-$(G,\mpp)$-dynamical system 
with $(\caM,\theta)$ balanced, central and type I.
By Proposition \ref{casebycase}, 
there is a $\kappa\in\bbZ_2$ and $(\caR_{\kappa, \caK},\Ad_{\Gamma_\caK},\Ad_{V_g})\in\caS_{\kappa}$
such that
$(\caM,\theta,{\hat\alpha})\sim (\caR_{\kappa, \caK},\Ad_{\Gamma_\caK},\Ad_{V_g})$.
Let ${\mqq}:G\to \bbZ_2$ be a group homomorphism such that
$\Ad_{V_g}(\Gamma_\caK)=(-1)^{{\mqq(g)}}\Gamma_\caK$ and 
$[\upsilon]$ the second cohomology class associated to the projective representation $V_g$ if
$\kappa=0$, and $V_g^{(0)}$ (from Lemma \ref{zv}) if $\kappa=1$.
We define an index of $(\caM,\theta,{\hat\alpha})$ by
\begin{align}
\Ind(\caM,\theta,{\hat\alpha}):=(\kappa, {\mqq}, [\upsilon])\in\bbZ_2\times H^1(G,\bbZ_2)\times H^2(G, U(1)_{\mpp}).
\end{align}
\end{defn}

\begin{lem} \label{lem:index_indep_of_R_iso}
The quantity $\Ind(\caM,\theta,{\hat\alpha})$ is well-defined and independent of the choice of 
 $(\caR_{\kappa, \caK},\Ad_{\Gamma_\caK},\Ad_{V_g})\in\caS_{\kappa}$ 
 such that $(\caM,\theta,{\hat\alpha})\sim (\caR_{\kappa, \caK},\Ad_{\Gamma_\caK},\Ad_{V_g})$.
\end{lem}
\begin{proof}
Suppose that both 
$(\caR_{\kappa_1,\caK_1},\Ad_{\Gamma_{\caK_1}},\Ad_{V_g^{(1)}})\in\caS_{\kappa_1}$ 
and
$(\caR_{\kappa_2, \caK_2},\Ad_{\Gamma_{\caK_2}},\Ad_{V_g^{(2)}})\in\caS_{\kappa_2}$
are equivalent to $(\caM,\theta,{\hat\alpha})$, via
$*$-isomorphisms $\iota_i: \caM\to \caR_{\kappa_i, \caK_i}$, $i=1,2$, respectively.
Then $\iota_2\circ\iota_1^{-1}: \caR_{\kappa_1,\caK_1}\to \caR_{\kappa_2, \caK_2}$
is a $*$-isomorphism such that for all $g\in G$,
\begin{align}\label{iiin}
&\iota_2\circ\iota_1^{-1}\circ\Ad_{V_g^{(1)}}=\Ad_{V_g^{(2)}}\circ \iota_2\circ\iota_1^{-1}, 
&&\iota_2\circ\iota_1^{-1}\circ\Ad_{\Gamma_{\caK_1}}=
\Ad_{\Gamma_{\caK_2}}\circ \iota_2\circ\iota_1^{-1}.
\end{align}
Let $(\kappa_i, {\mqq}_i, [\upsilon_i])$ be indices obtained from 
$(\caR_{\kappa_i,\caK_i},\Ad_{\Gamma_{\caK_i}},\Ad_{V_g^{(i)}})$, for $i=1,2$.
Because of the $*$-isomorphism $\iota_2\circ\iota_1^{-1}$,
we clearly have $\kappa_1=\kappa_2$.
If  $\kappa_1=\kappa_2=0$, then both of $\iota_i^{-1}(\unit_{\caK_i}\otimes\sigma_z)$, $i=1,2$,
are self-adjoint unitaries in $Z(\caM^{(0)})\setminus\bbC\unit$. From the proof of Proposition \ref{casebycase}, 
this means that $\iota_2\circ\iota_1^{-1}(\unit_{\caK_1}\otimes\sigma_z)=\pm 
(\unit_{\caK_2}\otimes\sigma_z)$.
Hence we get
\begin{align}
&(-1)^{{\mqq_1(g)}}\iota_2\circ\iota_1^{-1}(\unit_{\caK_1}\otimes\sigma_z)
=\iota_2\circ\iota_1^{-1}\circ\Ad_{V_g^{(1)}}(\unit_{\caK_1}\otimes\sigma_z)
=\Ad_{V_g^{(2)}}\circ \iota_2\circ\iota_1^{-1}(\unit_{\caK_1}\otimes\sigma_z)\nonumber\\
&\hspace{1.5cm} =\pm \Ad_{V_g^{(2)}}(\unit_{\caK_2}\otimes\sigma_z)
=\pm (-1)^{{\mqq}_2(g)}
(\unit_{\caK_2}\otimes\sigma_z)
=(-1)^{{\mqq}_2(g)}\iota_2\circ\iota_1^{-1}(\unit_{\caK_1}\otimes\sigma_z)
\end{align}
We therefore obtain that ${\mqq_1(g)}={\mqq}_2(g)$.
When $\kappa_1=\kappa_2=1$, an analogous argument for 
$\iota_i^{-1}(\unit_{\caK_i}\otimes\sigma_x)\in Z(\caM)\cap\caM^{(1)}$, $i=1,2$ 
implies ${\mqq_1(g)}={\mqq}_2(g)$.

If  $\kappa_1=\kappa_2=0$, the $*$-isomorphism 
$\iota_2\circ\iota_1^{-1}:\caB(\caK_1)\otimes\Mat_2\to\caB(\caK_2)\otimes\Mat_2$ is implemented by
a unitary $W:\caK_1\otimes \bbC^2\to\caK_2\otimes \bbC^2$.
Hence we see  from \eqref{iiin} that $\Ad_{WV_g^{(1)}W^*}(x)=\Ad_{V_g^{(2)}}(x)$
for all  $x\in \caB(\caK_1)\otimes\Mat_2$.
This means that $[\upsilon_1]=[\upsilon_2]$.
If  $\kappa_1=\kappa_2=1$, 
the restriction of the $*$-isomorphism
$\iota_2\circ\iota_1^{-1}$ onto 
${\caB(\caK_1)}$ induces a $*$-isomorphism from
$\caB(\caK_1)$ to $\caB(\caK_2)$.
Therefore, there is 
a unitary $W:\caK_1\to\caK_2$
such that 
$
\iota_2\circ\iota_1^{-1}(x\otimes \unit)
=\Ad_W(x)\otimes \unit$,
for all $ x\in\caB(\caK_1)$.
Therefore,  from \eqref{iiin} we have 
 $\Ad_{W(V_g^{(1)})^{(0)}W^*}(x)=\Ad_{(V_g^{(2)})^{(0)}}(x)$
 for all $x\in\caB(\caK_1)$.
This means that $[\upsilon_1]=[\upsilon_2]$.
\end{proof}
Proposition \ref{casebycase}, Lemma \ref{lem:index_indep_of_R_iso} and the fact 
that equivalence of $W^*$-$(G,\mpp)$-dynamical systems is an equivalence relation 
gives us the following.
\begin{prop}\label{wsta}
Let
$(\caM_1,\theta_1,{\hat\alpha}_1)$, $(\caM_2,\theta_2,{\hat\alpha}_2)$
be graded $W^*$-$(G,\mpp)$-dynamical systems of 
balanced, central and type I graded von Neumann algebras.
If $(\caM_1,\theta_1,{\hat\alpha}_1)\sim(\caM_2,\theta_2,{\hat\alpha}_2)$, then
$\Ind(\caM_1,\theta_1,{\hat\alpha}_1)=\Ind(\caM_2,\theta_2,{\hat\alpha}_2)$.
\end{prop}

\subsection{The index for pure split states}
We now define an index to fermionic SPT phases. 
For each $\Theta$-invariant and $\alpha$-invariant state $\caA$,
$(\pi_{\varphi}(\al_R)'', \Ad_{\hat \Gamma_\varphi}, \hat\alpha_\varphi)$
is a graded $W^*$-$(G,\mpp)$-dynamical system.

We first review the split property and recent results of Matsui~\cite{Matsui20} 
that relate the split property to unique gapped ground states of the CAR-algebra. 
Given a state $\varphi$ on $\al$, $\varphi\vert_{\al_R}$ denotes the restriction 
of $\varphi$ to $\al_R$ and $\pi_{\varphi\vert_{\al_R}}$ is the GNS 
representation of $\al_R$ from this restricted state.

\begin{defn}
Let $\varphi$ be 
a pure $\Theta$-invariant state on $\al$.
We say that $\varphi$ satisfies the split property if  
$\pi_{\varphi\vert_{\al_R}}(\al_R)''$ is a type I von Neumann algebra.
\end{defn}

If $\varphi$ is a pure $\Theta$-invariant state satisfying the split property, then there is an 
approximate statistical independence between the half-infinite restrictions $\varphi\vert_{\al_R}$ and 
$\varphi\vert_{\al_L}$. It is shown in~\cite{Matsui13} that pure states whose entanglement 
entropy is \emph{uniformly} bounded on finite regions satisfy the split property. Hence, the split property of 
pure states is closely related to the area law of entanglement entropy in one-dimensional systems.
See~\cite{OT, OTT} for further applications of the split property to Lieb--Schulz--Mattis-type theorems 
in the setting of quantum spin chains.

Recall the notation $\caB_f^e$ which denotes the set of all 
finite-range even interactions that satisfy the bound \eqref{fi}. 
Similarly, $\caG_{f}^{e,\alpha}$ denotes 
the set of all $\alpha$-invariant interactions $\Phi\in\caB_f^{e}$,
with a unique gapped ground state.

\begin{thm}[\cite{Matsui20}]
Let $\varphi$ be a unique gapped $\tau^\Phi$-ground state of an 
interaction $\Phi\in{\caB}_f^e$.
Then $\varphi$ satisfies the split property.
\end{thm}
To apply Matsui's result to graded $W^*$-$(G,\mpp)$-dynamical systems, we must first 
relate the split ground state of an interaction $\Phi\in {\caG}_f^{e,\alpha}$ to 
balanced and central graded type I von Neumann algebras.
To show this, we first note
the following.
\begin{lem}\label{hs}
Let $\varphi$ be a $\Theta$-invariant pure state on $\al$. Then
\begin{enumerate}
    \item[(i)] $Z(\pi_\varphi(\al_R)'')\cap \lmk \pi_\varphi(\al_R)''\rmk^{(0)}=\bbC\unit$, 
    \item[(ii)] The representation $\pi_{\varphi\vert_{\al_R}}$ and 
    $(\pi_\varphi)\vert_{\al_R}$, the restriction of $\pi_{\varphi}$ to $\al_R$, are quasi-equivalent.
\end{enumerate}
\end{lem}
\begin{proof}
(i) We have that
\begin{align}
Z(\pi_\varphi(\al_R)'')\cap \lmk \pi_\varphi(\al_R)''\rmk^{(0)}
\subset
\pi_\varphi(\al_L)'\cap\pi_\varphi(\al_R)'=\pi_\varphi(\al)'=\bbC\unit,
\end{align}
where the last equality is because $\varphi$ is pure.

\vspace{0.1cm}

\noindent (ii)
Let ${\hat \Gamma_\varphi}$ be a self-adjoint unitary 
on $\caH_\varphi$ given by 
${\hat \Gamma_\varphi}\pi_\varphi(A)\Omega_\varphi=\pi_\varphi\circ\Theta(A)\Omega_\varphi$,
$A\in\al$.
Let $p$ denote the the orthogonal projection onto $\overline{\pi_\varphi(\al_R)\Omega_\varphi}$. 
Then 
$(p\caH_\varphi, \pi_\varphi(\cdot )\vert_{\al_R}p,\Omega_\varphi)$ is a GNS triple of $\varphi\vert_{\al_R}$.
To show (ii), it suffices to show that
$\tau: \pi_\varphi(\al_R)''\to (\pi_\varphi(\al_R)p)''$ defined by 
$\tau(x)=xp$ is a $*$-isomorphism.
It is standard to see that $\tau$ is a surjective $*$-homomorphism.
To see that $\tau$ is injective, note that
from (i) and Lemma \ref{jh}, 
either $\pi_\varphi(\al_R)''$ is factor or $Z(\pi_\varphi(\al_R)'')=\bbC\unit+\bbC b$
with
some self-adjoint unitary $b\in Z(\pi_\varphi(\al_R)'')\cap (\pi_\varphi(\al_R)'')^{(1)}$.
For the former case, $\tau$ is clearly injective.
For the latter case, let $b=P_+-P_-$ be the spectral decomposition.
Because $b$ is odd, we have $\Ad_{\hat \Gamma_\varphi}(P_\pm)=P_{\mp}$.
If $\tau$ is not injective,
the kernel of $\tau$ is either $\pi_\varphi(\al_R)''P_+$ or $\pi_\varphi(\al_R)''P_-$.
If $\tau(P_+)=0$, then we have
$P_+\Omega_\varphi$=0. 
We then have
\begin{align}
P_-\Omega_\varphi={\hat \Gamma_\varphi} P_+{\hat \Gamma_\varphi}\Omega_\varphi={\hat \Gamma_\varphi} P_+\Omega_\varphi=0.
\end{align}
Hence we obtain $\Omega_\varphi=(P_++P_-)\Omega_\varphi=0$, which is a contradiction.
Similarly, we have $\tau(P_-)\neq 0$.
Therefore, $\tau$ is injective.
\end{proof}


\begin{lem}\label{deru}
Let $\varphi$ be a split pure  $\Theta$-invariant and $\alpha$-invariant state on $\al$.
Then $\pi_{\varphi}(\al_R)''$ is balanced, central with respect to 
the grading given by $\hat\Gamma_\varphi$,
and it is type I. 
The triple 
$(\pi_{\varphi}(\al_R)'', \Ad_{\hat \Gamma_\varphi}, \hat\alpha_\varphi)$
is a graded $W^*$-$(G,\mpp)$-dynamical system.
\end{lem}
\begin{proof}
Because $\varphi$ is pure and $\Theta$-invariant,
$\pi_{\varphi}(\al_R)''$ is central by part (i) of Lemma \ref{hs}.
Because $\varphi$ is split, $\pi_{\varphi\vert_{\al_R}}(\al_R)''$ is type I 
by definition. Because 
$(\pi_{\varphi})\vert_{\al_R}$ is quasi-equivalent to $\pi_{\varphi\vert_{\al_R}}$
by part (ii) of Lemma \ref{hs}, $\pi_{\varphi}(\al_R)''$ is also type I.
It is also balanced because $\al_R$ has an odd self-adjoint unitary.
Because $\alpha_g\circ\Theta=\Theta\circ\alpha_g$, for all $g\in G$ then
we have $\lmk \hat\alpha_\varphi\rmk_g\circ\Ad_{\hat \Gamma_\varphi}=\Ad_{\hat \Gamma_\varphi}\circ\lmk \hat\alpha_\varphi\rmk_g$.
\end{proof}
\begin{rem}\label{sds}
Consider the setting of Lemma \ref{deru}.
Let $\varphi_R:=\varphi\vert_{\al_R}$. Then 
$(\pi_{\varphi_R}(\al_R)'', \Ad_{\hat \Gamma_{\varphi_R}}, \hat\alpha_{\varphi_R})$
is also 
 a graded $W^*$-$(G,\mpp)$-dynamical system of 
a balanced, central and type I graded von Neumann algebra with 
\begin{align}
(\pi_{\varphi}(\al_R)'', \Ad_{\hat \Gamma_\varphi}, \hat\alpha_\varphi)
\sim
(\pi_{\varphi_R}(\al_R)'', \Ad_{\hat \Gamma_{\varphi_R}}, \hat\alpha_{\varphi_R}).
\end{align}
\end{rem}

From Lemma \ref{deru}, we see that our index of 
$W^*$-$(G,\mpp)$-dynamical systems can be applied to split, pure, $\Theta$-invariant and 
$\alpha$-invariant states on $\al$.
In particular, we may define an index for
 $\Phi\in {\caG}_f^{e,\alpha}$.
\begin{defn}
Let $\varphi$ be a $\Theta$-invariant, $\alpha$-invariant, split and 
pure state on $\al$ with $\varphi_R:=\varphi\vert_{\al_R}$. We set
\begin{align}
\ind\varphi:=\Ind(\pi_{\varphi}(\al_R)'', \Ad_{\hat \Gamma_\varphi}, \hat\alpha_\varphi) 
 = \Ind(\pi_{\varphi_R}(\al_R)'', \Ad_{\hat \Gamma_{\varphi_R}}, \hat\alpha_{\varphi_R}).
\end{align}
For interactions $\Phi\in {\caG}_f^{e,\alpha}$,
we define the index of $\Phi$ by 
$\ind(\Phi):=\ind(\varphi_\Phi)$,
with $\varphi_\Phi$ the unique ground state of $\Phi$.
\end{defn}

\section{The stability of the index}\label{stability}
In this section we prove that $\ind(\Phi)$ is an invariant of the 
classification of SPT phases. That is, for a path of interactions $\{\Phi(s)\}_{s\in[0,1]}$
satisfying Assumption \ref{assump} below, we show that 
$\ind(\Phi(0)) = \ind(\Phi(1))$.


 For each $N\in\bbN$, we denote $[-N,N]\cap \bbZ$ by $\Lambda_N$.
Let $\bbE_{N}:\caA\to \caA_{\Lambda_N}$ be the conditional expectation with respect to the trace state, 
see \cite{am}.
We consider the following subset of $\caA$.
\begin{defn}
Let $f:(0,\infty)\to (0,\infty)$ be a continuous decreasing function 
with $\lim_{t\to\infty}f(t)=0$.
For each $A\in\caA$, let
\begin{align}
\lV A\rV_f:=\lV A\rV
+ \sup_{N\in \nan}\lmk\frac{\lV
A-\bbE_{N}(A)
\rV}
{f(N)}
\rmk.
\end{align}
We denote by $\caD_f$ the set of all $A\in\caA$ such that
$\lV A\rV_f<\infty$.
\end{defn}
We consider a path in $\caG_f^{e, \alpha}$ satisfying the following conditions.
\begin{assum}\label{assump}
Let  
$[0,1]\ni s\mapsto \Phi ( s) \in \caB_f^e$ 
be a path of interactions on $\al$.
We assume the following:
\begin{enumerate}
  \item[(i)]
For each $X\in{\mathfrak S}_{\bbZ}$, the map
$[0,1]\ni s\mapsto \Phi(X;s)\in\caA_{X}$ is continuous and piecewise $C^1$.
We denote by $\dot{\Phi}(X;s)$ 
the corresponding derivatives.
The interaction obtained by differentiation is denoted by $\dot\Phi(s)$, for each $s\in[0,1]$.
\item[(ii)]
There is a number $R\in\nan$
such that $X \in {\mathfrak S}_{\bbZ}$ and $\diam{X}\ge R$ implies $\Phi(X;s)=0$ for all $s\in[0,1]$.
\item[(iii)]
For each $s\in[0,1]$,
$\Phi(s)\in \caG^{e, \alpha}_f$.
We denote the unique $\tau^{\Phi(s)}$-ground state by
$\varphi_s$. 
\item[(iv)] Interactions are bounded as follows
\begin{align}
\sup_{s\in[0,1]}\sup_{X\in {\mathfrak S}_{\bbZ}}
\lmk
\lV
\Phi\lmk X;s\rmk
\rV+|X|\lV
\dot{\Phi} \lmk X;s\rmk
\rV
\rmk<\infty.
\end{align}
\item[(v)]
Setting
\begin{align}
b(\varepsilon):=\sup_{Z\in{\mathfrak S}_{\bbZ}}
\sup_{\substack{s,s_0 \in[0,1], \\ 0<| s-s_0|<\varepsilon}}
\lV
\frac{\Phi(Z;s)-\Phi(Z;s_0)}{s-s_0}-\dot{\Phi}(Z;s_0)
\rV
\end{align}
for each $\varepsilon>0$, we have
$\lim_{\varepsilon\to 0} b(\varepsilon)=0$.
\item[(vi)] 
There exists a $\gamma>0$ such that
$\sigma(H_{\varphi_s,\Phi(s)})\setminus\{0\}\subset [\gamma,\infty)$ for
all $s\in[0,1]$, where  $\sigma(H_{\varphi_s,\Phi(s)})$ is the spectrum of $H_{\varphi_s,\Phi(s)}$.
\item[(vii)]
There exists $0<\beta<1$ satisfying the following:
Set $\zeta(t):=e^{-t^{ \beta}}$.
Then for each $A\in D_\zeta$, 
$\varphi_s(A)$ is differentiable with respect to $s$, and there is a constant
$C_\zeta$ such that
\begin{align}\label{dcon}
\lv
\dot{\varphi_s}(A)
\rv
\le C_\zeta\lV A\rV_\zeta,
\end{align}
for any $A\in D_\zeta$.
\end{enumerate}
\end{assum}

The main result of this section is the following.
\begin{thm}\label{c1t}
Let $[0,1]\ni s\mapsto \Phi ( s) \in \caB_f^e$ 
be a path of interactions on $\al$
satisfying Assumption \ref{assump}. 
Then
$\ind(\Phi(0))=\ind(\Phi(1))$.
\end{thm}

The proof relies on the idea introduced in \cite{OgataTRI}, 
that is, using the factorization property of automorphic equivalence.
Namely, we note the following.
\begin{prop}\label{aep}
Let $[0,1]\ni s\mapsto \Phi ( s) \in \caB_f^e$ 
be a path of interactions on $\al$
satisfying Assumption \ref{assump}. 
Let $\varphi_s$ be the unique $\tau^{\Phi(s)}$-ground state, for each $s\in[0,1]$. 
Then
there is an automorphism $\Xi$ on $\al$ and a unitary $u\in\al$ such that 
for all $g\in G$,
\begin{align*}
&\Xi(\al_L)=\al_L,  &&\Xi(\al_R)=\al_R,  &&\Xi\circ\Theta=\Theta\circ\Xi, 
&&\Xi\circ\alpha_g=\alpha_g\circ\Xi, \\
&\Theta(u)=u,  &&\alpha_g(u)=u, 
&&\varphi_1=\varphi_0\circ\Ad_u\circ\Xi.
\end{align*}
\end{prop}
In Appendix \ref{flr}, we prove the Lieb-Robinson bound and a locality estimate for 
lattice fermion systems. 
Having them, the proof of
Proposition \ref{aep} is the same
 as that of~\cite[Theorem 1.3]{mo} and~\cite[Proposition 3.5]{OgataTRI}. 

To prove Theorem \ref{c1t}, we first prove a preperatory lemma.
\begin{lem}\label{hi}
Let $\varphi_1,\varphi_2$ be pure $\Theta$-invariant states on $\al$.
If $\varphi_1$ and $\varphi_2$ are quasi-equivalent, 
then 
 $\varphi_1\vert_{\al_R}$ and $\varphi_2\vert_{\al_R}$ are quasi-equivalent.
\end{lem}
\begin{proof}
Let $\pi_i$, $\pi_{i,R}$ be GNS representations of $\varphi_i$ and $\varphi_i\vert_{\al_R}$ 
respectively for $i=1,2$.
By Lemma \ref{hs}, there are $*$-isomorphisms
$\tau_{i}:\pi_i(\al_R)''\to \pi_{i,R}(\al_R)''$, for $i=1,2$
such that $\tau_{i}\circ\pi_i(A)=\pi_{i,R}(A)$, $A\in\al_R$.
Because $\varphi_1$ and $\varphi_2$ are quasi-equivalent, 
there is a $*$-isomorphism $\tau: \pi_1(\al)''\to\pi_2(\al)''$
 such that $\tau\circ\pi_1(A)=\pi_2(A)$, for
$A\in\al$.
The restriction of $\tau$ to $\pi_1(\al_R)''$ gives a $*$-isomorphism
$\tau_R:\pi_1(\al_R)''\to \pi_2(\al_R)''$.
Hence we obtain a $*$-isomorphism
$\hat\tau:=\tau_2\circ\tau_R\circ\tau_1^{-1}: \pi_{1,R}(\al_R)''\to \pi_{2,R}(\al_R)''$
such that $\hat\tau\circ\pi_{1,R}(A)=\pi_{2,R}(A)$, $A\in\al_R$.
Therefore, $\varphi_1\vert_{\al_R}$ and $\varphi_2\vert_{\al_R}$ are quasi-equivalent.
\end{proof}
Now we are ready to prove the Theorem.
\begin{proofof}[Theorem \ref{c1t}]
Let $(\caH_i,\pi_i,\Omega_i)$ be the GNS triple of the states $\varphi_i\vert_{\al_R}$ for $i=0,1$.
Let $\Gamma_i$ be a self-adjoint unitary given by $\Gamma_i\pi_i(A)\Omega_i=\pi_i\circ\Theta(A)\Omega_i$,
$A\in\al_R$.
Let $\hat\alpha_i$ be the extension of $\alpha\vert_{\al_R}$ to $\pi_i(\al_R)''$.
From Proposition \ref{wsta} and Remark \ref{sds}, it suffices to show that
$(\pi_0(\al_R)'', \Ad_{\Gamma_0}, \hat\alpha_0)\sim
(\pi_1(\al_R)'', \Ad_{\Gamma_1}, \hat\alpha_1)$.
Recalling the $\ast$-automorphism $\Xi$ from Proposition \ref{aep}, 
$\Xi(\al_R)=\al_R$ and so $\Xi_R:=\Xi\vert_{\al_R}$ defines a $*$-automorphism
on $\al_R$.
Note that $(\caH_0,\pi_0\circ\Xi_R,\Omega_0)$ is
 a GNS triple of $\varphi_0\vert_{\al_R}\circ\Xi_R$.
The state $\varphi_1=\varphi_0\circ\Ad_u\circ\Xi$
is quasi-equivalent to  $\varphi_0\circ\Xi$.
Because $\Xi\circ\Theta=\Theta\circ\Xi$,
both $\varphi_0\circ\Xi$ and $\varphi_1$ are $\Theta$-invariant pure states.
Applying Lemma \ref{hi},
$\varphi_1\vert_{\al_R}$ and $\varphi_0\circ\Xi\vert_{\al_R}=\varphi_0\vert_{\al_R}\circ\Xi_R$
are quasi-equivalent.
Hence there is a $*$-isomorphism 
\begin{align}
  \tau: \pi_0\circ\Xi_R(\al_R)''=\pi_0(\al_R)''\to \pi_1(\al_R)'', \qquad 
  \tau\circ\pi_0\circ\Xi_R(A)=\pi_1(A), \quad A \in \al_R.
\end{align}
Using properties of the quasi-equivalence $\tau$ and automorphism $\Xi_R$, we see that 
\begin{align}
\tau\circ\hat\alpha_{0,g}\circ \pi_0\circ\Xi_R(A)
&=\tau\circ\pi_0\circ\alpha_g\circ \Xi_R(A)
=\tau\circ\pi_0\circ \Xi_R\circ\alpha_g(A)  \nonumber \\
&=\pi_1\circ\alpha_g(A) 
=\hat\alpha_{1,g}\circ \pi_1(A)
=\hat\alpha_{1,g}\circ \tau\circ\pi_0\circ\Xi_R(A),\\
\tau\circ\Ad_{\Gamma_0}\circ\pi_0\circ\Xi_R(A)
&=\tau\circ\pi_0\circ\Theta\circ\Xi_R(A)
=\tau\circ\pi_0\circ\Xi_R\circ\Theta(A)\nonumber\\
&=\pi_1\circ\Theta(A)
=\Ad_{\Gamma_1}\circ\pi_1(A)
=\Ad_{\Gamma_1}\circ\tau\circ\pi_0\circ\Xi_R(A)
\end{align}
for all $A\in\al_R$.
Hence we obtain
\begin{align}
\tau\circ\hat\alpha_{0,g}(x)=\hat\alpha_{1,g}\circ \tau(x),\quad
\tau\circ\Ad_{\Gamma_0}(x)=
\Ad_{\Gamma_1}\circ\tau(x),\quad
x\in \pi_0(\al_R)''.
\end{align}
This completes the proof.
\end{proofof}
\section{Stacking and group law of fermionic SPT phases}\label{stacksec}

\subsection{The graded tensor product} \label{subsec:graded_product_def}
Let  $(\caM_1,\Ad_{\Gamma_1})$ and $(\caM_2,\Ad_{\Gamma_2})$ be spatially
graded von Neumann algebras acting on 
on $\caH_1$, $\caH_2$
with grading operators $\Gamma_1$, $\Gamma_2$.
We define a product and involution on the algebraic tensor product $\caM_1\odot \caM_2$
by
\begin{align}
(a_1\hox b_1)(a_2\hox b_2)
&=(-1)^{\partial b_1\partial a_2}(a_1a_2\hox b_1b_2), \nonumber \\
(a\hox b)^* &=(-1)^{\partial a\partial b} a^*\hox b^*.
%
\end{align}
for homogeneous elementary tensors.
The algebraic tensor product with this multiplication and involution 
is a $*$-algebra, denoted $\caM_1\,\hat{\odot}\,\caM_2$.
On the Hilbert space $\caH_1\otimes\caH_2$,
\begin{align}\label{pprep}
\pi (a\hox b)
:=a\Gamma_1^{\partial b}\otimes b
\end{align}
for homogeneous $a\in \caM_1$, $b\in\caM_2$
defines a faithful $*$-representation of $\caM_1\,\hat{\odot}\,\caM_2$.
We call the von Neumann algebra generated by $\pi(\caM_1\,\hat{\odot}\,\caM_2)$
the graded tensor product of 
$(\caM_1,\caH_1,\Gamma_1)$ and $(\caM_2,\caH_2,\Gamma_2)$ 
and denote it by $\calM_1 \hox \calM_2$. 
It is simple to check that 
$\calM_1 \hox \calM_2$ is a spatially graded von Neumann algebra with a grading operator
$\Gamma_1\otimes \Gamma_2$.

For  $a\in \caM_1$ and homogeneous $b\in\caM_2$,
we denote $\pi(a\hox b)$ by $a{\hox} b$, embedding $\caM_1\,\hat{\odot}\,\caM_2$
in $\caM_1\hox\caM_2$.
Note that $\partial(a\hox b) = \partial(a) + \partial(b)$ for homogeneous
$a\in \caM_1$ and $b\in \caM_2$.

Fix a finite group $G$ and a homomorphism $\mpp: G\to\bbZ_2$.
Let $(\caM_1,\Ad_{\Gamma_1},\alpha_1)$ and  $(\caM_2,\Ad_{\Gamma_2},\alpha_2)$ 
be graded $W^*$-$(G,\mpp)$-dynamical systems, where $(\caM_1, \mathrm{Ad}_{\Gamma_1})$ and 
$(\caM_2,\mathrm{Ad}_{\Gamma_2})$ are spatially graded, balanced, central and type I.
We may define an action
$\alpha_1\hox\alpha_2$ of $G$ on $\caM_1\hox\caM_2$
 by
\begin{align}
\lmk \alpha_1\hox\alpha_2\rmk_g
(a\hox b)
=\alpha_{1,g}(a){\hox} \alpha_{2,g}(b),\quad g\in G
\end{align}
for all  homogeneous $a\in \caM_1$ and $b\in\caM_2$, 
see Lemma \ref{ga8}.

\subsection{Stacking and the group law}

In this section, we show that $W^*$-$(G,\mpp)$-dynamical systems of 
balanced, central, type I and spatially graded von Neumann algebras are 
closed under graded tensor products. Furthermore, our index from 
Definition \ref{index} obeys a twisted group law (a generalized Wall group law) under this operation.

\begin{thm}\label{stlaw}
Let $(\caM_1,\Ad_{\Gamma_1},\alpha_1)$, $(\caM_2,\Ad_{\Gamma_2},\alpha_2)$ be
graded $W^*$-$(G,\mpp)$-dynamical systems
with balanced, central and spatially graded type I von Neumann algebras.
Then the triple $(\caM_1\hox \caM_2, \Ad_{\Gamma_1\otimes\Gamma_2}, \alpha_1\hox \alpha_2)$
is a graded $W^*$-$(G,\mpp)$-dynamical system
with a balanced, central and spatially graded type I von Neumann algebra.
If $\Ind(\caM_i,\Ad_{\Gamma_i},\alpha_i)=(\kappa_i,\mqq_i,[\upsilon_i])$, $i=1,2$, 
then 
\begin{align}\label{totind}
\Ind (\caM_1\hox \caM_2, \Ad_{\Gamma_1\otimes\Gamma_2}, \alpha_1\hox\alpha_2)
=\big(\kappa_1+\kappa_2,\, \mqq_1+\mqq_2+\kappa_1\kappa_2\mpp, \, [\upsilon_1\, \upsilon_2 \, \epsilon_\mpp(\kappa_1,\mqq_1,\kappa_2,\mqq_2)] \big),
\end{align}
where $\epsilon_\mpp(\kappa_1,\mqq_1,\kappa_2,\mqq_2)$ is a group $2$-cocycle defined by 
\begin{align}
\epsilon_\mpp(\kappa_1,\mqq_1,\kappa_2,\mqq_2)(g,h)
=(-1)^{{\mqq_1(g)}\mqq_2(h)+(\kappa_1-\kappa_2)(\kappa_1\mqq_2(g)+\kappa_2{\mqq_1(g)})\cdot{\mpp(h)}},\quad
g,h\in G.
\end{align}
\end{thm}
\begin{rems}
\begin{enumerate}
    \item[(i)] One can check that \eqref{totind} gives an abelian group law, which is not surprising because of
the corresponding properties of the graded tensor product.
    \item[(ii)] The group law \eqref{totind} is a little cumbersome in full generality, but simplifies 
    in many examples of interest. For example, if $\alpha_1$ and $\alpha_2$ are linear group actions, 
    $\mpp(g)=0$ for all $g\in G$, we recover the more familiar twisted sum formula,
    \begin{align}
        (\kappa_1,\, \mqq_1,\, [\upsilon_1]) \cdot (\kappa_2,\, \mqq_2,\, [\upsilon_2]) 
        = (\kappa_1 +\kappa_2,\, \mqq_1+\mqq_2,\, [\upsilon_1\, \upsilon_2 \, \epsilon(\mqq_1,\mqq_2)] ).
    \end{align}
\end{enumerate}
\end{rems}
\begin{proof}
By Lemma \ref{oisoiso} and Lemma \ref{wsta},
we may assume that 
\begin{align}
(\caM_i,\Ad_{\Gamma_i},\alpha_i)=(\caR_{\kappa_i, \caK_i},\Ad_{\Gamma_{\caK_i}},\Ad_{V_{ i}})\in\caS_{\kappa_i}.
\end{align}
Let
\begin{align}
\Ind\big(  \caR_{\kappa_i,\caK_i},\Ad_{\Gamma_{\caK_i}},\Ad_{V_{ i}} \big)
 =(\kappa_i,\mqq_i, [\upsilon_i]),\quad i=1,2.
\end{align}
We would like to show that
\begin{align}\label{eqvrr}
\lmk\caR_{\kappa_1,\caK_1}\hox\caR_{\kappa_2, \caK_2},
\Ad_{\Gamma_{\caK_1}}{\hox}\Ad_{\Gamma_{\caK_1}},
\Ad_{ V_{1}}\hox\Ad_{V_{2}}
\rmk
\sim 
(\caR_{\kappa, \caK},\Ad_{\Gamma_{\caK}},\Ad_{V})\in\caS_{\kappa},
\end{align}
for suitably chosen $\kappa=0,1$, Hilbert space $\caK$ and projective representation
$V$ on $\caK\otimes\cct$,
satisfying
\begin{align}\label{tind}
\Ind(\caR_{\kappa, \caK},\Ad_{\Gamma_{\caK}},\Ad_{V_{g}}) 
=\big(\kappa_1+\kappa_2,\, \mqq_1+\mqq_2+\kappa_1\kappa_2\mpp, \, [\upsilon_1\, \upsilon_2 \, 
\epsilon_\mpp(\kappa_1,\mqq_1,\kappa_2,\mqq_2)] \big).
\end{align}

\noindent{\bf (Case: $\kappa_1=0$ or $\kappa_2=0$) }

We set the following notation,
\begin{align}
  &\caK := \caK_1\otimes\caK_2\otimes\cct, 
  &&\lambda = \begin{cases} 1, &\text{if}\,\, \kappa_1=\kappa_2 = 0, \\ 2, &\text{if}\,\, \kappa_1=1, \,\, \kappa_2 = 0, \\
       3, & \text{if}\,\,\kappa_1 = 0, \,\, \kappa_2 = 1, \end{cases}
\end{align}
and define the unitary $v:\bbC^2\otimes \caK_2\to \caK_2\otimes \cct$, 
\begin{align}\label{swap}
v(\xi\otimes \eta)=\eta\otimes \xi,\quad \xi\in \cct,\quad \eta\in\caK_2.
\end{align}

Using the standard basis $\{e_0,e_1\}$ of $\cct$,
we define the unitaries $w_1,w_2,w_3$ on $\cct\otimes\cct$
by
\begin{align*}
&w_1(e_0\otimes e_0)=e_0\otimes e_0,
&&w_1(e_1\otimes e_1)=e_1\otimes e_0,
&&w_1(e_1\otimes e_0)=e_0\otimes e_1,
&&w_1(e_0\otimes e_1)=e_1\otimes e_1,\\
&w_2(e_0\otimes e_0)=e_0\otimes e_0,
&&w_2(e_1\otimes e_1)=e_1\otimes e_0,
&&w_2(e_1\otimes e_0)=e_0\otimes e_1,
&&w_2(e_0\otimes e_1)=-e_1\otimes e_1,\\
&w_3(e_0\otimes e_0)=e_0\otimes e_0,
&&w_3(e_1\otimes e_1)=e_1\otimes e_0,
&&w_3(e_1\otimes e_0)=e_1\otimes e_1,
&&w_3(e_0\otimes e_1)=e_0\otimes e_1.
\end{align*}
By direct calculation, we may check
\begin{align}
&\Ad_{w_\lambda}\lmk\sigma_z\otimes\sigma_z\rmk
=\unit_{\cct}\otimes\sigma_z,\quad \lambda=1,2,3,\label{ttz}\\
&\Ad_{w_2}\lmk \sigma_x\otimes\sigma_z\rmk
=\unit_{\cct}\otimes\sigma_x,\quad
\Ad_{w_3}\lmk\unit_{\cct}\otimes\sigma_x\rmk=\unit_{\cct}\otimes\sigma_x.\label{ttx}
\end{align}

We now define unitary $U_\lambda:\caK_1\otimes\cct\otimes\caK_2\otimes \cct\to 
\caK\otimes \cct$
such that 
\begin{align}
U_\lambda:=(\unit_{\caK_1}\otimes \unit_{\caK_2}\otimes w_\lambda)(\unit_{\caK_1}\otimes v\otimes \unit_{\cct}),\quad\lambda=1,2,3.
\end{align}
By \eqref{ttz},  we have
\begin{align}\label{uttz}
\Ad_{U_\lambda}\lmk\Gamma_{\caK_1}\otimes \Gamma_{\caK_2}\rmk=\Gamma_{\caK},\quad\lambda=1,2,3,
\end{align}
hence 
\begin{align}\label{tu}
\Ad_{U_\lambda}\circ\lmk \Ad_{\Gamma_{\caK_1}}\hox\Ad_{\Gamma_{\caK_2}}\rmk(x)
=\Ad_{\Gamma_{\caK}}\circ\Ad_{U_\lambda}(x),
\quad x\in \caR_{\kappa_1, \caK_1}\hox\caR_{\kappa_2,\caK_2},\quad
\lambda=1,2,3.
\end{align}
By \eqref{ttx},  
when $\lambda=2$,
for $\lmk \unit_{\caK_1}\otimes \sigma_x\rmk{\hox}
\lmk \unit_{\caK_2}\otimes \sigma_z\rmk\in 
\caR_{\kappa_1,\caK_1}{\hox} \caR_{\kappa_2,\caK_2}
$,
we have
\begin{align}\label{uttx2}
\Ad_{U_2}\lmk\lmk \unit_{\caK_1}\otimes \sigma_x\rmk{\hox}
\lmk \unit_{\caK_2}\otimes \sigma_z\rmk
\rmk=\unit_{\caK}\otimes \sigma_x.
\end{align}
Similarly, when $\lambda=3$,
for $
\lmk \unit_{\caK_1}\otimes \sigma_z\rmk{\hox}
\lmk \unit_{\caK_2}\otimes \sigma_x\rmk
\in\caR_{\kappa_1, \caK_1}{\hox} \caR_{\kappa_2, \caK_2}
$,
\begin{align}\label{uttx3}
\quad
\Ad_{U_3}\lmk\lmk \unit_{\caK_1}\otimes \sigma_z\rmk{\hox}
\lmk \unit_{\caK_2}\otimes \sigma_x\rmk
\rmk=\unit_{\caK}\otimes \sigma_x.
\end{align}

Let $[\tilde \upsilon_i]$ be the second cohomology class associated to the projective representation
$V_{i}$, $i=1,2$.
We set 
\begin{align}
V_g:=\Ad_{U_\lambda}\lmk V_{1,g}\otimes V_{2,g}\Gamma_{\caK_2}^{{\mqq_1(g)}}\rmk,\quad g\in G,\quad\lambda=1,2,3.
\end{align}
This gives a projective unitary/anti-unitary representation $V$ of $G$ on $\caK\otimes\cct$ relative to $\mpp$. Using that 
$\Ad_{V_{2,g}}\lmk\Gamma_{\caK_2}\rmk=(-1)^{\mqq_2(g)}\Gamma_{\caK_2}$ for $g\in G$,
the second cohomology class associated to $V$ is equal to
$[\tilde \upsilon_1 \,\tilde \upsilon_2 \, \epsilon(\mqq_1,\mqq_2)] \in H^2(G, U(1)_\mpp)$, where 
$\epsilon(\mqq_1,\mqq_2)$ is given in \eqref{enn}.
By Lemma \ref{isoiso} and \ref{aisoaiso}, we have 
that for $x\in \caR_{\kappa_1, \caK_1}{\hox} \caR_{\kappa_2, \caK_2}$, $g\in G$ 
and any $\lambda=1,2,3$, 
\begin{align}\label{uv}
\Ad_{V_g}\circ\Ad_{U_\lambda}(x)=\Ad_{U_\lambda}\circ\Ad_{ V_{1,g}\otimes V_{2,g}\Gamma_{\caK_2}^{{\mqq_1(g)}}}(x)
=\Ad_{U_\lambda}\circ\lmk\alpha_{1,g}{\hox} \alpha_{2,g}\rmk\lmk
x
\rmk .
\end{align}
In particular, for $\lambda=2,3$, we also have
\begin{align}\label{uttxl}
\Ad_{V_g} (\unit_{\caK}\otimes \sigma_x)=
(-1)^{{\mqq_1(g)}+\mqq_2(g)}
\lmk \unit_{\caK}\otimes \sigma_x\rmk,\quad g\in G,
\end{align}
from \eqref{uttx2} and \eqref{uttx3}.

By \eqref{uttz}, we have
\begin{align}\label{vzz}
\Ad_{V_g}\lmk\Gamma_\caK\rmk
=\Ad_{U_\lambda}\circ\Ad_{ V_{1,g}\otimes V_{2,g}\Gamma_{\caK_2}^{{\mqq_1(g)}}}\lmk
\Gamma_{\caK_1}\otimes\Gamma_{\caK_2}
\rmk
=(-1)^{{\mqq_1(g)}+\mqq_2(g)}\Gamma_\caK,\quad g\in G.
\end{align}

Having set up the required preliminaries, 
we now consider the $W^*$-$(G,\mpp)$-dynamical system 
$(\caR_{\kappa, \caK},\Ad_{\Gamma_{\caK}},\Ad_{V_{g}})\in \caS_{\kappa}$ and 
show equivalence with the graded tensor product in
the three cases where $\kappa_1$ or $\kappa_2=0$.

\vspace{0.1cm}

\noindent{(i)-1}
For $\lambda=1$ (i.e. $\kappa_1=\kappa_2=0$), we set $\kappa=0$ 
and 
note from \eqref{vzz} that $(\caR_{\kappa, \caK},\Ad_{\Gamma_{\caK}},\Ad_{V_{g}})\in \caS_{\kappa}$. 
In this case,  $[\tilde \upsilon_i]= [\upsilon_i]$ and 
$\epsilon_\mpp(0,\mqq_1,0,\mqq_2)=\epsilon(\mqq_1,\mqq_2)$. 
Hence the second cohomology class of $V$ is 
$[\upsilon_1\upsilon_2 \epsilon_\mpp(0,\mqq_1,0,\mqq_2)]$.
With this and (\ref{vzz}), the index of $(\caR_{\kappa, \caK},\Ad_{\Gamma_{\caK}},\Ad_{V_{g}})$
is given by \eqref{tind}.
So we just 
need to show equivalence of the $W^*$-$(G,\mpp)$-dynamical system with the 
graded tensor product. 
The equivalence is given by a $*$-isomorphism
\begin{align}
\iota:=\Ad_{U_1}: \caB(\caK_1\otimes\cct\otimes\caK_2\otimes\cct)
=\caR_{\kappa_1, \caK_1}{\hox} \caR_{\kappa_2, \caK_2}\to 
\caB(\caK\otimes\cct)=\caR_{0,\caK}.
\end{align}
By \eqref{tu} and \eqref{uv}, $\iota$ satisfies the required conditions \eqref{eone} and \eqref{etwo}
for equivalence of $W^*$-$(G,\mpp)$-dynamical systems.

\vspace{0.1cm}

\noindent{(i)-2}
For $\lambda=2$ (i.e. $\kappa_1=1,\, \kappa_2=0$),  set $\kappa=1$.
By \eqref{uttxl} and \eqref{vzz}, we see that
$(\caR_{\kappa, \caK},\Ad_{\Gamma_{\caK}},\Ad_{V_{g}})\in \caS_{\kappa}$.
Note that $[\tilde \upsilon_1]= [\upsilon_1\, \epsilon (\mqq_1,\mpp)] \in {H^2(G,U(1)_\mpp)}$, 
see  Lemma \ref{zv} and Definition \ref{index}, 
with $\tilde \upsilon_2=\upsilon_2$.
Hence the second cohomology associated to our projective representation 
$V$
is 
\begin{align}
 \big[ \tilde \upsilon_1 \, \tilde \upsilon_2 \, \epsilon(\mqq_1,\mqq_2)\big] =
\big[ \upsilon_1 \,\upsilon_2 \,\epsilon (\mqq_1,\mpp) \, \epsilon(\mqq_1,\mqq_2)\big].
\end{align}
Combining this and \eqref{vzz}, the second cohomology associated to the projective representation 
$V^{(0)}$ (cf. Lemma \ref{zv} and Definition \ref{index}) is
\begin{align*}
\big[\tilde \upsilon_1 \,\tilde \upsilon_2 \, \epsilon(\mqq_1,\mqq_2) \,\epsilon(\mqq_1+\mqq_2,\mpp) \big] 
= \big[ \upsilon_1 \,\upsilon_2 \,\epsilon (\mqq_1,\mpp) \, \epsilon(\mqq_1,\mqq_2) 
   \, \epsilon(\mqq_1+\mqq_2,\mpp) \big] 
= \big[\upsilon_1 \,\upsilon_2 \,\epsilon_\mpp(1,\mqq_1,0,\mqq_2) \big].
\end{align*}
From this and \eqref{vzz}, we see the index of 
$(\caR_{\kappa, \caK},\Ad_{\Gamma_{\caK}},\Ad_{V_{g}})\in \caS_{\kappa}$ is given by \eqref{tind}.

Now we show  $(\caR_{\kappa, \caK},\Ad_{\Gamma_{\caK}},\Ad_{V_{g}})$ is equivalent 
to the graded tensor product (\ref{eqvrr}).
From Lemma \ref{comgra}, the commutant of $\caR_{\kappa_1, \caK_1}{\hox} \caR_{\kappa_2,\caK_2}$
is 
$\bbC\unit_{\caK_1\otimes\cct\otimes\caK_2\otimes\cct}+\bbC\unit_{\caK_1}\otimes\sigma_x\otimes\unit_{\caK_2}\otimes\sigma_z$.
Note that 
by \eqref{uttx2}, 
$\Ad_{U_2}$ maps the commutant to 
$\bbC\unit_{\caK}\otimes \unit_{\cct}+\bbC\unit_{\caK}\otimes \sigma_x=(\caR_{\kappa,\caK})'$.
Therefore we have
$\Ad_{U_2}(\caR_{\kappa_1, \caK_1}{\hox}\caR_{\kappa_2, \caK_2})
=\caR_{\kappa, \caK}$.
Hence $\iota:=\Ad_{U_2}\vert_{\caR_{\kappa_1, \caK_1}{\hox}\caR_{\kappa_2,\caK_2}}$
defines a $*$-isomorphism 
$\iota: \caR_{\kappa_1, \caK_1}{\hox} \caR_{\kappa_2,\caK_2}\to \caR_{\kappa,\caK}$.
By \eqref{tu} and \eqref{uv},
$\iota$ satisfies the required conditions of an equivalence of $W^*$-$(G,\mpp)$-dynamical systems.

\vspace{0.1cm}

\noindent{(i)-3}
For $\lambda=3$ (i.e. $\kappa_1=0,\, \kappa_2=1$), we set $\kappa=1$.
By \eqref{vzz} and \eqref{uttxl}, we see that
$(\caR_{\kappa, \caK},\Ad_{\Gamma_{\caK}},\Ad_{V_{g}})\in \caS_{\kappa}$. 
We also have that  $[\tilde \upsilon_1]=[\upsilon_1]$ 
and $[\tilde \upsilon_2]=[\upsilon_2\,\epsilon (\mqq_2,\mpp)]$.
Hence the second cohomology class associated to $V$
is 
\begin{align}
 \big[ \tilde \upsilon_1 \,\tilde \upsilon_2 \,\epsilon(\mqq_1,\mqq_2)\big] =
\big[ \upsilon_1 \,\upsilon_2 \, \epsilon (\mqq_2,\mpp) \, \epsilon(\mqq_1,\mqq_2)\big]
.
\end{align}
Hence,  from \eqref{vzz} the cohomology class associated to $V^{(0)}$ is
\begin{align}
\big[ \tilde \upsilon_1\, \tilde \upsilon_2 \, \epsilon(\mqq_1,\mqq_2)
 \, \epsilon(\mqq_1+\mqq_2,\mpp)\big]
= \big[ \upsilon_1 \, \upsilon_2 \, \epsilon_\mpp(0,\mqq_1,1,\mqq_2)\big]
\end{align}
and the index of $(\caR_{\kappa, \caK},\Ad_{\Gamma_{\caK}},\Ad_{V_{g}})\in \caS_{\kappa}$ 
is given by \eqref{tind}.

We now show that $(\caR_{\kappa, \caK},\Ad_{\Gamma_{\caK}},\Ad_{V_{g}})$ is 
equivalent to the graded tensor product.
From Lemma \ref{comgra}, the commutant of $\caR_{\kappa_1,\caK_1}{\hox}\caR_{\kappa_2,\caK_2}$
is 
$\bbC\unit_{\caK_1\otimes\cct\otimes\caK_2\otimes\cct}+
\bbC\unit_{\caK_1\otimes\cct\otimes\caK_2}
\otimes\sigma_x$, 
which by \eqref{uttx3} is mapped to 
$\bbC\unit_{\caK}\otimes \unit_{\cct}+\bbC\unit_{\caK}\otimes \sigma_x=(\caR_{\kappa,\caK})'$ 
by $\Ad_{U_3}$.
Therefore, 
$\Ad_{U_3}(\caR_{\kappa_1, \caK_1}{\hox}\caR_{\kappa_2, \caK_2})
=\caR_{\kappa,\caK}$
and $\iota:=\Ad_{U_3}\vert_{\caR_{\kappa_1,\caK_1}{\hox}\caR_{\kappa_2,\caK_2}}$
defines a $*$-isomorphism 
$\iota: \caR_{\kappa_1,\caK_1}{\hox}\caR_{\kappa_2,\caK_2}\to \caR_{\kappa,\caK}$ and 
implements an equivalence of $W^*$-$(G,\mpp)$-dynamical systems.

\vspace{0.1cm}

\noindent{\bf (Case: $\kappa_1=\kappa_2=1$) }

Set $\kappa:=0$ and $\caK:=\caK_1\otimes\caK_2$.
We define a projective representation $V$ of $G$ on $\caK\otimes\cct$ relative to $\mpp$ by
\begin{align}
V_g:=V_{1, g}^{(0)}\otimes V_{2, g}^{(0)}\otimes C^{{\mpp(g)}}\sigma_y^{{\mqq_1(g)}}\sigma_x^{\mqq_2(g)+{\mpp(g)}},\quad
g\in G.
\end{align}
Here $V_{i}^{(0)}$ is the projective representation on $\caK_1$
such that $V_{i,g}=V_{i,g}^{(0)}\otimes C^{{\mpp(g)}}\sigma_y^{\mqq_i(g)}$ for $i=1,2$ 
(see Lemma \ref{zv}).
Then we have
\begin{align}
\Ad_{V_g}\lmk {\Gamma_\caK}\rmk= (-1)^{{\mqq_1(g)}+\mqq_2(g)+{\mpp(g)}}{\Gamma_\caK},\quad
g\in G.
\end{align}Hence $(\caR_{\kappa, \caK},\Ad_{\Gamma_{\caK}},\Ad_{V_{g}})\in\caS_\kappa$.
Because $\sigma_y$ anti-commutes with $\sigma_x$ and $C$, while
$\sigma_x$ commutes with $C$, the second cohomology class
associate to the projective representation $V$
is 
\begin{align}
\big[ \upsilon_1 \, \upsilon_2 \, \epsilon (\mqq_1,\mqq_2)\big] 
= \big[ \upsilon_1 \, \upsilon_2 \, \epsilon_\mpp(1,\mqq_1,1,\mqq_2)\big],
\end{align}
where we recall that $[\epsilon (\mqq_1,\mqq_2)] = [\epsilon (\mqq_2,\mqq_1)]$.
Hence the triple 
$(\caR_{\kappa, \caK},\Ad_{\Gamma_{\caK}},\Ad_{V_{g}})\in \caS_{\kappa}$
has index given by \eqref{tind}.

Now we show \eqref{eqvrr} for the constructed $(\caR_{\kappa, \caK},\Ad_{\Gamma_{\caK}},\Ad_{V_{g}})$.
Regarding $\Clf$ as a graded von Neumann algebra $(\Clf,\Ad_{\sigma_z})\subset \Mat_2$, 
there is a $\ast$-isomorphism 
$\iota_0:\Clf{\hox}\Clf\to\Mat_2$ such that 
\begin{align}\label{ioz}
\iota_0(\unit{\hox}\unit)=\unit,\quad
\iota_0(\sigma_x{\hox}\unit):=\sigma_x,\quad
\iota_0(\unit{\hox}\sigma_x):=\sigma_y,\quad
\iota_0(\sigma_x{\hox}\sigma_x):=i\sigma_z.
\end{align}
Noting $\Ad_{\unit_{\caK_1}\otimes v\otimes \unit_{\cct}}
\big(\caR_{\kappa_1,\caK_1}{\hox}\caR_{\kappa_2,\caK_2}\big)
=\caB(\caK_1)\otimes \caB(\caK_2)\otimes \lmk \Clf{\hox}\Clf\rmk$ 
with $v$ in \eqref{swap}, we obtain a $*$-isomorphism
$\iota: \caR_{\kappa_1,\caK_1}{\hox}\caR_{\kappa_2,\caK_2}\to 
\bk\otimes \Mat_2=\caR_{\kappa,\caK}$
given by
\begin{align}
\iota(x):=\lmk \id_{\caK}\otimes\iota_0\rmk\circ\Ad_{\unit_{\caK_1}\otimes v\otimes \unit_{\cct}}(x),\quad
x\in \caR_{\kappa_1,\caK_1}{\hox}\caR_{\kappa_2,\caK_2}.
\end{align}
We then have
\begin{align*}
&\Ad_{V_g}\circ\iota\lmk \lmk a\otimes\sigma_x\rmk{\hox}  \lmk b\otimes\unit_{\cct}\rmk\rmk
=\Ad_{V_g}\lmk
a\otimes b\otimes \sigma_x
\rmk
=\Ad_{V_{1, g}^{(0)}}(a)\otimes \Ad_{V_{2, g}^{(0)}}(b)
\otimes (-1)^{{\mqq_1(g)}}\sigma_x \\
&\hspace{1.5cm} =\iota\lmk
\Ad_{V_{1, g}}\lmk a\otimes \sigma_x\rmk
\hat
\otimes \lmk \Ad_{V_{2, g}}\lmk b\otimes \unit_{\cct}\rmk
\rmk
\rmk
=\iota\circ\lmk\alpha_{1,g}{\hox} \alpha_{2,g}\rmk\lmk
\lmk a\otimes\sigma_x\rmk{\hox}  \lmk b\otimes\unit_{\cct}\rmk
\rmk,
\end{align*}
and 
\begin{align*}
&\Ad_{V_g}\circ\iota\lmk \lmk a\otimes\unit_{\cct}
\rmk{\hox}  \lmk b\otimes\sigma_x\rmk\rmk
=\Ad_{V_g}\lmk
a\otimes b\otimes \sigma_y
\rmk
=\Ad_{V_{1, g}^{(0)}}(a)\otimes \Ad_{V_{2, g}^{(0)}}(b)
\otimes (-1)^{\mqq_2(g)}\sigma_y \\
&\hspace{1.5cm} =\iota\lmk
\Ad_{V_{1, g}}\lmk a\otimes \unit_{\cct}\rmk
\hat
\otimes \lmk \Ad_{V_{2, g}}\lmk b\otimes \sigma_x\rmk
\rmk
\rmk
=\iota\circ\lmk\alpha_{1,g}{\hox} \alpha_{2,g}\rmk\lmk
\lmk a\otimes\unit_{\cct}\rmk{\hox}  \lmk b\otimes\sigma_x\rmk
\rmk
\end{align*}
for all $a\in\caB(\caK_1)$, $b\in \caB(\caK_2)$.
Because the  elements 
$(a\otimes\sigma_x){\hox} (b\otimes\unit_{\cct})$
and $(a\otimes\unit_{\cct}) \hox (b \otimes \sigma_x)$
generate $\caR_{\kappa_1,\caK_1}{\hox}\caR_{\kappa_2,\caK_2}$,
we see that
$\Ad_{V_g}\circ\iota(x)=\iota\circ\lmk\alpha_{1,g}{\hox} \alpha_{2,g}\rmk(x)$
for $x\in \caR_{\kappa_1,\caK_1}{\hox}\caR_{\kappa_2,\caK_2}$.
We also see from \eqref{ioz}
that
$\Ad_{\Gamma_\caK}\circ\iota(x)=\iota\circ\lmk\Ad_{\Gamma_1}{\hox} \Ad_{\Gamma_2}\rmk(x)$
for $x\in \caR_{\kappa_1,\caK_1}{\hox}\caR_{\kappa_2,\caK_2}$.
Hence we obtain \eqref{eqvrr}.
\end{proof}
\begin{ex}[Time-reversal symmetry and the $\bbZ_8$-classification]
As a simple example, let us consider fermionic SPT phases with time-reversal symmetry. 
That is, we take $G=\bbZ_2 =\{0,1\}$ with with ${\mpp(1)}=1$. We let $\alpha= \alpha_1$ 
be the anti-linear $\ast$-automorphism of order $2$ from the non-trivial element. 
Therefore, if $G$ acts on a balanced, central 
and type I von Neumann algebra, then $\alpha$ is implemented on a graded 
Hilbert space $\caK$ by $\mathrm{Ad}_R$ with 
$R$ anti-unitary. Following~\cite{OgataTRI}, we can ensure that 
$R^2 = \pm \unit_\caK$ and so the group $2$-cocycle is determined by 
the sign of $R^2$. 

The data $\Z_2 \times H^1(\bbZ_2,\bbZ_2)\times H^2(\bbZ_2,U(1)_\mpp)$ from 
Theorem \ref{stlaw} is wholly determined by the 
triple $[\kappa; \varepsilon, \pm]$, where $\varepsilon=\mqq(1) \in \bbZ_2$ and 
$\pm$ is the sign of $R^2$. Our choice of notation is so that our results 
can easily be compared with~\cite[Appendix A]{MoutuouBrauer} 
and~\cite{Wall}. 
Following \eqref{totind}, the triple has the (abelian) composition law under 
stacking
\begin{align*}
  [0;\varepsilon_1, \xi_1] [0, \varepsilon_2,\xi_2] &= [0; \varepsilon_1+\varepsilon_2, (-)^{\varepsilon_1\varepsilon_2} \xi_1 \xi_2] \\
  [0;\varepsilon_1, \xi_1] [1, \varepsilon_2,\xi_2] &= [1; \varepsilon_1+\varepsilon_2, (-)^{\varepsilon_1+\varepsilon_1\varepsilon_2} \xi_1 \xi_2] \\
  [1;\varepsilon_1, \xi_1] [1, \varepsilon_2,\xi_2] &= [0; \varepsilon_1+\varepsilon_2+1, (-)^{\varepsilon_1\varepsilon_2} \xi_1 \xi_2].
\end{align*}
One therefore sees that $\Z_2 \times H^1(\bbZ_2,\bbZ_2)\times H^2(\bbZ_2,U(1)_\mpp) \cong \bbZ_8$ 
with generator $[1;0,+]$. Hence we recover and extend the $\bbZ_8$-classification of time-reversal 
symmetric fermionic SPT phases in one dimension considered for finite systems in~\cite{fk,FK2,BWHV}.
\end{ex}


\section{Translation invariant states}\label{transsec}
In this section, we derive a representation of pure, split, 
translation invariant and $\alpha$-invariant states in terms of a
finite set of operators on Hilbert spaces.
The idea of the proof is the same as quantum spin case, cf.~\cite{bjp, Matsui3}, 
although anti-commutativity results in richer structures.

Recall the integer shift $S_x$ on $l^2(\bbZ)\otimes\bbC^d$, $x\in \bbZ$, which defines the 
$\ast$-automorphism $\beta_{S_x} \in \Aut(\al)$.
Let $\omega$ be a pure, split, $\alpha$-invariant and translation invariant 
state on ${\al}$. In particular, such 
states are $\Theta$-invariant (see~\cite[Example 5.2.21]{BR2}).
By Proposition \ref{casebycase} and Lemma \ref{deru}
the graded $W^*$-$(G,\mpp)$-dynamical system 
$(\pi_\omega(\al_R)'', \Ad_{\Gamma_\omega}, \hat\alpha_\omega)$ associated to $\omega$
is equivalent to some $(\caR_{\kappa, \caK},\Ad_{\Gamma_\caK},\Ad_{V_g})\in\caS_{\kappa}$.
We denote this $\kappa$ by $\kappa_\omega$.
The space translation lifts to an endomorphism on $\pi_\omega(\al_R)''$.

\begin{lem}\label{iorho}
Let $\omega$ be a pure,  split, 
$\alpha$-invariant and translation invariant state on $\al$.
Suppose that the graded $W^*$-$(G,\mpp)$-dynamical system
$({\pi_\omega}(\al_R)'', \Ad_{\Gamma_\omega}, \hat\alpha_\omega)$ associated to $\omega$
is equivalent to $(\caR_{\kappa, \caK},\Ad_{\Gamma_\caK},\Ad_{V_g})\in\caS_{\kappa}$, via
a $*$-isomorphism $\iota: {\pi_\omega}(\al_R)''\to \caR_{\kappa, \caK}$.
Then
there is an injective $*$-endomorphism $\rho$ on $\caR_{\kappa, \caK}$,
such that
\begin{align}\label{ipr}
\iota\circ{\pi_\omega}\circ\beta_{S_1}(A)=\rho\circ\iota\circ{\pi_\omega}(A),\quad
A\in\al_R.
\end{align}
Furthermore, we have
\begin{align}\label{ababy}
a\rho(b)-(-1)^{\partial a\partial b}\rho(b)a=0,
\end{align}
for homogeneous $a\in\iota\circ{\pi_\omega}\big( \al_{\{0\}}\big)$
and $b\in \caR_{\kappa, \caK}$.
\end{lem}
\begin{proof}
By the translation invariance of $\omega$,
the space translation $\beta_{S_1}$ 
is lifted to an automorphism $\hat\beta_{S_1}$ on ${\pi_\omega}(\al)''$.
Restricting $\hat\beta_{S_1}$ to ${\pi_\omega}(\al_R)''$, we obtain
an injective $*$-endomorphism $\tilde \beta$ on 
${\pi_\omega}(\al_R)''$.
We then see that
$\rho:=\iota\circ \tilde \beta \circ\iota^{-1}:  \caR_{\kappa, \caK}\to \caR_{\kappa, \caK}$
is an injective endomorphism on $\caR_{\kappa, \caK}$ satisfying \eqref{ipr}.
Because $\beta_{S_1}(\A_R) \subset \A_{\bbZ \geq 1}$, we see that 
$a_0 \beta_{S_1}(a_1) -(-1)^{\partial a_0\partial a_1}\beta_{S_1}(a_1)a_0=0$ 
for homogeneous $a_0 \in \al_{\{0\}}$
and $a_1 \in \al_R$. Then, because 
$\rho\lmk \caR_{\kappa, \caK}\rmk=\lmk \iota\circ{\pi_\omega}\circ\beta_{S_1}\lmk \al_{R}\rmk\rmk''$, 
equation \eqref{ababy} follows.
\end{proof}
Let $\caP$ be 
the power set $\calP = \calP(\{1,\ldots,d\}) = 2^{\{1,\ldots,d\}}$ of $\{1,\ldots,d\}$.
We denote
 the parity of the number of the elements in $\mu\in\caP$ by
$|\mu| = \# \mu \,\mathrm{mod} \, 2$. 
We denote by $\{\psi_\mu\}_{\mu\in\caP}$
the standard basis of $\calF(\bbC^d)$.
Namely, with the Fock vacuum $\Omega_d$ of  $\calF(\bbC^d)$ and the standard basis
$\{e_{i}\}_{i=1}^d$ of $\bbC^d$,
$\psi_\mu$ for $\mu\neq\emptyset$ is given by
$\psi_\mu=C_\mu a^*(e_{\mu_1})a^*(e_{\mu_2})\cdots a^*(e_{\mu_l})\Omega_d$
with $l=\#\mu$, $\mu=\{\mu_1,\mu_2,\ldots,\mu_l\}$ with $\mu_1<\mu_2\cdots<\mu_l$ 
and a suitable normalization factor $C_\mu\in\bbC\setminus \{0\}$.
For the empty set $\mu=\emptyset$, we set $\psi_\emptyset:=\Omega_d$.
 
We denote the matrix units of $\al_{\{0\}}\simeq \caB(\caF(\bbC^d)) \simeq \Mat_{2^d}$
associated to the standard basis $\{\psi_\mu\}_{\mu\in\caP}$
by $\{ E_{\mu,\nu}^{(0)}\}$, $\mu,\nu\in\calP$.
Because $\Theta$ is implemented by the second quantization of $-\unit_{\bbC^d}$,
\begin{align}
\mathfrak\Gamma(-\unit)=\sum_{\mu\in \caP}(-1)^{|\mu|} E_{\mu\mu}^{(0)}\in\caA_{\{0\}},
\end{align}
we see that 
\begin{align}
\Theta(E_{\mu,\nu}^{(0)})=(-1)^{|\mu|+|\nu|} E_{\mu,\nu}^{(0)},\qquad
\mu,\nu\in \caP.
\end{align}
We set $E_{\mu,\nu}^{(x)}:=\beta_{S_x}\big( E_{\mu,\nu}^{(0)}\big)$ for general $x\in\bbZ$.
Clearly, $\{ E_{\mu,\nu}^{(x)}\}_{\mu,\nu\in\caP}$ are matrix units of $\al_{\{x\}}$.

\begin{lem} \label{lem:trans_weak_conv_to_state}
Let $\omega$ be a pure, split and translation invariant state on ${\al}$
and $\hat\beta_{S_n}$ be the extension of $\beta_{S_n}$
to $\pi_\omega(\al)''$, i.e. $\hat\beta_{S_n}\circ\pi_\omega(A)=\pi_\omega\circ\beta_{S_n}(A)$,
$A\in\al$.
\begin{enumerate}
\item[(i)]
If $x\in\lmk\pi_\omega(\al_R)''\rmk^{(0)}$, then
$\sigma\mathrm{\hbox{-}weak}\lim_{n\to\infty}\hat\beta_{S_n}(x)
=\braket{\Omega_\omega}{x\Omega_\omega}\unit_{\Ho}$.
\item[(ii)]
If $x\in\lmk\pi_\omega(\al_R)''\rmk^{(1)}$ and $\pi_\omega(\al_R)''$ is a factor,
then 
\begin{align}\label{app}
\sigma\mathrm{\hbox{-}weak}\lim_{n\to\infty}
\pi_\omega\lmk
{\mathfrak\Gamma}(-\unit)
\beta_{S_1}
\lmk
{\mathfrak\Gamma}(-\unit)\rmk\cdots
\beta_{S_{n-1}}
\lmk
{\mathfrak\Gamma}(-\unit)\rmk\rmk
\hat \beta_{S_n}(x)
=0=\braket{\Omega_\omega}{x\Omega_\omega}\unit_{\caH_\omega}.
\end{align}
\item[(iii)]
If $x\in\lmk\pi_\omega(\al_R)''\rmk^{(1)}$ and $Z\lmk\pi_\omega(\al_R)''\rmk\cap \lmk\pi_\omega(\al_R)''\rmk^{(1)}\neq \{0\}$,
then
$\sigma\mathrm{\hbox{-}weak}\lim_{n\to\infty}\hat\beta_{S_n}(x)=0
=\braket{\Omega_\omega}{x\Omega_\omega}$.
\end{enumerate}
\end{lem}
\begin{proof}
First we note from the $\sigma$-weak continuity of $\hat\beta_{S_n}$ that
\begin{align}\label{marui}
\hat\beta_{S_n}\lmk\lmk \pi_\omega(\al_R)''\rmk^{(\sigma)}\rmk
\subset 
\lmk\lmk  \pi_\omega\circ\beta_{S_n}(\al_R)\rmk''\rmk^{(\sigma)},\quad n\in\bbN,\quad \sigma=0,1.
\end{align}
(i) 
By \eqref{marui}, we have 
$\hat\beta_{S_n}\lmk\lmk \pi_\omega(\al_R)''\rmk^{(0)}\rmk
\subset \pi_\omega(\al_{[0,n-1]})'$.
Therefore for any $x\in\lmk\pi_\omega(\al_R)''\rmk^{(0)}$,
any $\sigma$-weak accumulation point $z$ of $\{\hat\beta_{S_n}(x)\}$
belongs to $\pi_\omega(\al_R)'\cap \lmk\pi_\omega(\al_R)''\rmk^{(0)}$.
But $\pi_\omega(\al_R)'\cap \lmk\pi_\omega(\al_R)''\rmk^{(0)}=\bbC\unit_{\caH_\omega}$
by Lemma \ref{hs}.
Hence we have $z\in \bbC\unit_{\caH_\omega}$.
Because $\braket{\Omega_\omega}{\hat\beta_{S_n}(x)\Oo}=\braket{\Oo}{ x\Oo}$,
this means $z=\braket{\Oo}{ x\Oo}\unit_{\Ho}$.
As this holds for any accumulation point, we obtain $\sigma\mathrm{\hbox{-}weak}\lim_{n\to\infty}\hat\beta_{S_n}(x)
=\braket{\Omega_\omega}{x\Omega_\omega}\unit_{\Ho}$.

\vspace{0.1cm}

\noindent
(ii)
Suppose that $\pi_\omega(\al_R)''$ is a factor and 
set $Y_n:={\mathfrak\Gamma}(-\unit)
\beta_{S_1}({\mathfrak\Gamma}(-\unit))\cdots \beta_{S_{n-1}}({\mathfrak\Gamma}(-\unit))$.
Note that
$\Ad_{Y_n}(B)=\Theta(B)$ for any $B\in\al_{[0,n-1]}$.
Therefore, by \eqref{marui}, we have 
$\po (Y_n) \hat\beta_{S_n}\big( (\po(\al_R)'')^{(1)}\big) 
\subset \po(\al_{[0,n-1]})'$.
Hence, for any $x\in\lmk\pi_\omega(\al_R)''\rmk^{(1)}$,
any $\sigma$-weak accumulation point $z$ of the set $\{\po\lmk Y_n\rmk\hat\beta_{S_n}(x)\}$
belongs to $\pi_\omega(\al_R)'\cap \lmk\pi_\omega(\al_R)''\rmk^{(1)}=\{0\}$. 
As such, $z=0$.
As this holds for any accumulation point, we obtain (ii).

\vspace{0.1cm}

\noindent
(iii)
 Suppose that $Z\lmk\pi_\omega(\al_R)''\rmk\cap \lmk\pi_\omega(\al_R)''\rmk^{(1)}\neq \{0\}$.
 By \eqref{marui}, we have that
 \begin{align}
     \hat\beta_{S_n}\big( (\pi_\omega(\al_R)'')^{(1)}\big)
\subset \po(\al_R)''\cap \pi_\omega(\al_{[0,n-1]})'\Gamma_\omega.
 \end{align}
Therefore for any $x\in\lmk\pi_\omega(\al_R)''\rmk^{(1)}$,
any $\sigma$-weak accumulation point $z$ of $\{\hat\beta_{S_n}(x)\}$
belongs to $\lmk \pi_\omega(\al_R)'\Gamma_\omega\rmk\cap \lmk\pi_\omega(\al_R)''\rmk^{(1)}$.
Because $Z\lmk\pi_\omega(\al_R)''\rmk$ has an odd element, 
$\pi_\omega(\al_R)''$ is not a factor. 
Lemma \ref{jh2} 
then implies that
$\pi_\omega(\al_R)'\Gamma_\omega \cap \pi_\omega(\al_R)''=\{0\}$.
Hence we have $z=0$.
As this holds for any accumulation point, we obtain (iii).
\end{proof}
Before stating the result, we fix some notation. Given the operators $\{W_\mu\}_{\mu\in\calP}$ we define the CP map $T_{\bf W}$ by
$$
   T_{\bf W}(x) = \sum_{\mu\in\calP} W_\mu x W_\mu^*.
$$

Because the algebraic structure of the von Neumann algebra of interest changes 
depending on whether $\kappa_\omega=0,1$, we treat each case separately, though the 
general strategy of proof is the same.
\subsection{Case: $\kappa_\omega=0$}

Recall  that $\mathfrak{\Gamma}(U_g)$
denotes the second quantization of $U_g$ on $\calF(\C^d)$.
In this subsection we prove the following.

\begin{thm}\label{cuntzthm0}
Let $\omega$ be 
 a pure $\alpha$-invariant and 
translation invariant split state on $\al$.
Suppose that the graded $W^*$-$(G,\mpp)$-dynamical system
$(\pi_\omega(\al_R)'', \Ad_{\Gamma_\omega}, \hat\alpha_\omega)$ associated to $\omega$
is equivalent to $(\caR_{0,\caK},\Ad_{\Gamma_\caK},\Ad_{V_g})\in\caS_{0}$, via
a $*$-isomorphism $\iota: \pi_\omega(\al_R)''\to \caB(\caK\otimes\cct)$.
Let $\rho$ be the $*$-endomorphism on $\caR_{0,\caK}$ given in Lemma \ref{iorho}.
Then there is a set of isometries $\{B_\mu\}_{\mu \in \calP}$ on $\caK\otimes\cct$
such that
$B_\nu^* B_\mu = \delta_{\mu,\nu} \unit$,
  \begin{align}\label{bimp}
  \rho\circ\iota\circ\pi_\omega(A)
  = \sum_{\mu\in\calP} \mathrm{Ad}_{B_\mu\Gamma_{\caK}^{|\mu|}}\circ\iota\circ\pi_\omega (A)
  = \sum_{\mu\in\calP} \mathrm{Ad}_{\Gamma_{\caK}^{|\mu|}B_\mu}\circ\iota\circ\pi_\omega (A),\quad
  A\in\al_R,
\end{align}
and 
\begin{align}\label{bbbo}
\iota\circ\pi_\omega\lmk E_{\mu_0,\nu_0}^{(0)} E_{\mu_1,\nu_1}^{(1)} 
\cdots E_{\mu_N,\mu_N}^{(N)}\rmk = 
  (-1)^{\sum\limits_{k=1}^N (|\mu_k|+|\nu_k|) \sum\limits_{j=0}^{k-1} |\nu_j|} 
  B_{\mu_0} \cdots B_{\mu_N} B_{\nu_N}^* \cdots B_{\nu_0}^*
\end{align}
for all $N\in\bbN\cup\{0\}$ and $\mu_0,\ldots\mu_N,\nu_0,\ldots,\nu_N\in\caP$.
The operators $B_\mu$ have 
  homogeneous parity and are such that 
  $\Ad_{\Gamma_\caK}\lmk B_\mu\rmk = (-1)^{|\mu|+\sigma_0}B_\mu$,
  with some uniform $\sigma_0\in\{0,1\}$.
Furthermore, 
\begin{align}\label{sat}
\sigma\mathrm{\hbox{-}weak}\lim_{N\to\infty}T_{\bf B}^N\circ\iota(x)
=\braket{\Oo}{ x\Oo}\unit_{\caK\otimes\cct},\quad x\in\po(\caA_R)''
\end{align}
and for each $g\in G$, there is some 
  $c_g \in \mathbb{T}$ such that 
   \begin{align}\label{gvb}
   \sum_{\mu\in\calP} \langle \psi_\mu, \, \mathfrak{\Gamma}(U_g) \psi_\nu \rangle B_\mu = 
     c_g V_g B_\nu V_g^*.
   \end{align}
\end{thm}

We will prove this result in several steps. First we note some properties of 
endomorphisms of operators on graded Hilbert spaces and the Cuntz algebra.
\begin{prop}\label{cuntzpropi}
Let $\caH$ be a Hilbert space with a self-adjoint unitary $\Gamma$ 
that gives a grading for $\calB(\calH)$. 
Let $\caM$ be a finite type I von Neumann subalgebra of $\caB(\caH)$
with  matrix units $\{E_{\mu,\nu}\}_{\mu,\nu\in\calP} \subset \caM$ spanning 
$\caM$.
Assume that 
\begin{align}\label{grae}
\Ad_{\Gamma}(E_{\mu,\nu})=(-1)^{|\mu|+|\nu|} E_{\mu,\nu},\qquad
\mu,\nu\in \caP
\end{align}
and set $\Gamma_0:=\sum_{\mu\in\caP} (-1)^{|\mu|}E_{\mu\mu}$.
Let $\rho: \calB(\calH) \to \calB(\calH)$ be a graded, unital $\ast$-endomorphism 
such that $\rho(a)b- (-1)^{\partial a \partial b} b \rho(a) = 0$ for 
$a \in \calB(\calH)$ $b\in \caM$ with homogeneous 
grading. Suppose further that $\calB(\calH) = \rho(\calB(\calH)) \vee \calM$.
Then there exist isometries $\{S_\mu\}_{\mu \in \calP}$ on $\caH$ with the property that
\begin{align}
S_\nu^* S_\mu = \delta_{\mu,\nu} \unit, 
    \qquad 
    \rho(x) = \sum_\mu S_\mu x S_\mu^* 
\end{align}
   for all $\mu,\nu \in\calP$ and $x \in \calB(\calH)$. The operators $S_\mu$ have 
  homogeneous parity and are such that $\Ad_\Gamma\lmk S_\mu\rmk = (-1)^{|\mu|+\sigma_0}S_\mu$
  with some uniform $\sigma_0\in\{0,1\}$.
  Furthermore, setting $B_\mu:=(\Gamma_0\Gamma)^{|\mu|} S_\mu$, for $\mu\in\caP$,
  we have $B_\nu^* B_\mu = \delta_{\mu,\nu} \unit$,
  \begin{align}\label{rhotrans}
  \rho(x) 
  = \sum_{\mu\in\calP} \mathrm{Ad}_{B_\mu} \circ \mathrm{Ad}_{\Gamma^{|\mu|}} (x),\quad
  x\in\bh,
\end{align}
  and
\begin{align} \label{eq:rho_and_E_using_V_even}
  E_{\mu_0,\nu_0} \rho(E_{\mu_1,\nu_1}) \cdots \rho^N (E_{\mu_N,\mu_N}) = 
  (-1)^{\sum\limits_{k=1}^N (|\mu_k|+|\nu_k|) \sum\limits_{j=0}^{k-1} |\nu_j|} 
  B_{\mu_0} \cdots B_{\mu_N} B_{\nu_N}^* \cdots B_{\nu_0}^* 
\end{align}
for all $N\in\bbN\cup\{0\}$ and $\mu_0,\ldots,\mu_N,\nu_0,\ldots,\nu_N\in\caP$.
The operators $B_\mu$ have 
  homogeneous parity such that
  $\Ad_\Gamma\lmk B_\mu\rmk = (-1)^{|\mu|+\sigma_0}B_\mu$,
  with the same $\sigma_0$ as above.
If there are isometries $\{T_\mu\}_{\mu \in \calP}$ such that
\begin{align}\label{buni}
 T_\nu^* T_\mu = \delta_{\mu,\nu} \unit, \quad T_\mu T_\nu^*= E_{\mu,\nu},
\quad\rho(x) 
  = \sum_{\mu\in\calP} \mathrm{Ad}_{T_\mu} \circ \mathrm{Ad}_{\Gamma^{|\mu|}} (x),\quad
  x\in\bh,
\end{align}
then there is some $c\in\bbT$ such that $T_\mu=cB_\mu$, for all $\mu\in \caP$.
\end{prop}
To study the situation, we note the following general property.
\begin{lem}\label{lem35}
Let $\caH$ be a Hilbert space with a self-adjoint unitary $\Gamma$ 
that gives a grading for $\calB(\calH)$. 
Let $\caM_1$, $\caM_2$  be $\Ad_{\Gamma}$-invariant
 von Neumann subalgebras of $\bh$ with $\caM_1 \vee \caM_2 = \bh$.
Suppose that $\caM_1$ is a type I factor with a self-adjoint unitary $\Gamma_1\in \caM_1$
such that $\Ad_{\Gamma_1}(x)=\Ad_{\Gamma}(x)$ for all $x\in\caM_1$.
Suppose further that
\begin{align}\label{anti}
ab- (-1)^{\partial a \partial b} b a = 0,\quad \text{for homogeneous}\quad a\in\caM_1, b \in\caM_2.
\end{align}
Then there are Hilbert spaces $\caH_1,\caH_2$ and a unitary
$V:\caH\to\caH_1\otimes\caH_2$ such that
\begin{align}
\Ad_{V}(\caM_1)=\caB(\caH_1)\otimes \bbC\unit_{\caH_2}.\quad \label{mi}
\end{align}
Furthermore, there are self-adjoint
unitaries $\tilde\Gamma_i$ on $\caH_i$ with $i=1,2$
such that 
\begin{align}\label{ggg}
\Ad_V(\Gamma)=\tilde\Gamma_1\otimes\tilde\Gamma_2,\quad
\Ad_V(\Gamma_1)=\tilde\Gamma_1\otimes\unit_{\caH_2}.
\end{align}
The commutant of $\caM_2$ is given by
\begin{align}\label{mdes}
\caM_2'
=\caM_1^{(0)}+\caM_1^{(1)}\Gamma_1\Gamma.
\end{align}
If $p$ is an even minimal projection in $\caM_1$
then $\caM_2\cdot p=\caB(p\caH)$.
\end{lem}
We note that if $\caM_1$ is a type I factor, 
Wigner's Theorem guarantees the existence of a self-adjoint unitary $\Gamma_1\in \caM_1$
such that $\Ad_{\Gamma_1}(x)=\Ad_{\Gamma}(x)$ for all $x\in\caM_1$.
\begin{proof}
Because $\caM_1$ is a type I factor, by~\cite[Chapter V, Theorem 1.31]{Takesaki1}
 there are Hilbert spaces $\caH_1,\,\caH_2$ and a unitary
$V:\caH\to\caH_1\otimes\caH_2$ satisfying \eqref{mi}.
Because $\Gamma_1\in\caM_1$ and
 $\Gamma\Gamma_1\in \caM_1'$,  there are self-adjoint
unitaries $\tilde\Gamma_i$ on $\caH_i$ with $i=1,2$
satisfying \eqref{ggg}.
Clearly $\Ad_{\Gamma_1}(\Gamma_1)=\Gamma_1$ and so $\Gamma_1$
is an even element of $\caM_1$.

Note that $\caN:=\caM_2^{(0)}+\caM_2^{(1)}\Gamma_1$
is a von Neumann subalgebra of $\caM_1'$ by \eqref{anti}.
Therefore, $\Ad_V(\caN)$ is a von Neumann subalgebra of 
$\unit_{\caH_1}\otimes \caB(\caH_2)$.
Because
\begin{align*}
  &\caM_2=\caM_2^{(0)}+\caM_2^{(1)}\Gamma_1\Gamma_1\subset \caM_1\vee \caN, 
  &&\caM_1\subset\caM_1\vee \caN, 
  &&\caM_1\vee\caM_2=\caB(\caH),
\end{align*}
we have $\caM_1\vee \caN=\caB(\caH)$.
Combining with \eqref{mi}, this means 
\begin{align}\label{mini}
\Ad_{V}(\caM_2^{(0)}+\caM_2^{(1)}\Gamma_1)=\Ad_V(\caN)=\bbC\unit_{\caH_1}\otimes\caB(\caH_2).
\end{align}

Now we associate the grading given by $\tilde\Gamma_i$ to
$\caB(\caH_i)$ for $i=1,2$, and regard $\caB(\caH_1)\otimes\caB(\caH_2)$
as $\caB(\caH_1){\hox} \caB(\caH_2)$, the graded tensor product of 
$(\caB(\caH_1),\caH_1,\tilde{\Gamma}_1)$ and $(\caB(\caH_2),\caH_2,{\tilde\Gamma_2})$.
Because 
$\Ad_V(\Gamma)=\tilde\Gamma_1\otimes\tilde\Gamma_2$,
$\Ad_V: \caB(\caH)\to \caB(\caH_1){\hox} \caB(\caH_2)$
is a graded $*$-isomorphism.
Considering the even and odd subspaces of \eqref{mini},
we obtain
\begin{align}
\Ad_{V}(\caM_2^{(0)})=\bbC\unit_{\caH_1}\otimes\caB(\caH_2)^{(0)},\quad
\Ad_{V}(\caM_2^{(1)})\Ad_{V}(\Gamma_1)=
\bbC\unit_{\caH_1}\otimes\caB(\caH_2)^{(1)}.
\end{align}
and so 
\begin{align}\label{m2}
\Ad_{V}(\caM_2)=\Ad_{V}(\caM_2^{(0)}+\caM_2^{(1)})
=\bbC\unit_{\caH_1}\otimes\caB(\caH_2)^{(0)}
+\bbC\tilde\Gamma_1\otimes\caB(\caH_2)^{(1)}
= \bbC\unit_{\caH_1}{\hox} \caB(\caH_2)
\end{align}
where $\bbC\unit_{\caH_1}{\hox} \caB(\caH_2)$ is a graded tensor product of
$(\bbC\unit_{\caH_1},\caH_1, \tilde\Gamma_1)$ and 
$(\caB(\caH_2), \caH_2,\tilde\Gamma_2)$.


We now consider the commutant of $\caM_2$. Applying  Lemma \ref{comgra},
we see that
\begin{align}
\Ad_{V}(\caM_2')
=\caB(\caH_1)^{(0)}\otimes \bbC\unit_{\caH_2}
+\caB(\caH_1)^{(1)}\otimes \bbC\tilde\Gamma_2
=\Ad_{V}\big( \caM_1^{(0)}+\caM_1^{(1)}\Gamma_1\Gamma \big).
\end{align}
Hence we obtain \eqref{mdes}.

Let $p$ be a minimal projection in $\caM_1$ and suppose that it is even.
Then 
$\Ad_V(p)$ is a minimal projection in $\caB(\caH_1)\otimes\bbC\unit_{\caH_2}$.
Therefore, there is a rank-one projection $r$ on $\caH_1$
such that $\Ad_V(p)=r\otimes\unit_{\caH_2}$.
Because $p$ is even and $\Ad_V$ is a graded $*$-isomorphism,
we have
$\Ad_{\tilde \Gamma_1}(r)=r$.
As $r$ is rank-one, this means that $\tilde \Gamma_1 r=\pm r$.
Therefore, using \eqref{m2}, we have
\begin{align}
\Ad_V(\caM_2 p)
&=
\bbC r\otimes\caB(\caH_2)^{(0)}
+\bbC\tilde\Gamma_1r\otimes\caB(\caH_2)^{(1)} \nonumber \\
&=\bbC r\otimes \big( \caB(\caH_2)^{(0)}\pm \caB(\caH_2)^{(1)}\big) 
=\bbC r\otimes \caB(\caH_2)
=\Ad_V\lmk p \bh p\rmk. 
\end{align}
Hence we obtain $\caM_2 p=p\bh p=\caB(p\caH)$.
\end{proof}
\begin{lem} \label{lem:endo_and_commutant_even}
 Consider the setting of Proposition \ref{cuntzpropi}.
 Then the following hold.
\begin{enumerate}
  \item[(i)] $\rho(\bh)'=\caM^{(0)}+\caM^{(1)}\Gamma_0\Gamma$.
  \item[(ii)] Let $\hat{E}_{\mu,\nu} = E_{\mu,\nu} (\Gamma_0 \Gamma)^{|\mu|+|\nu|}$. Then 
  $\{\hat{E}_{\mu,\nu}\}_{\mu,\nu\in\calP}$ are matrix units in 
  $\rho(\calB(\calH))'$  
  spanning $\rho(\calB(\calH))'$,
  \item[(iii)] For all $\mu \in \calP$, the map 
  \begin{align}\label{rmu1}
    \rho_\mu: \calB(\calH) \ni x \mapsto \rho(x) E_{\mu,\mu} \in \calB( E_{\mu,\mu} \calH) 
  \end{align}  is a $\ast$-isomorphism.
\end{enumerate}
\end{lem}
\begin{proof}
Note that $\Ad_{\Gamma}(x)=\Ad_{\Gamma_0}(x)$ for $x\in \caM$.
Applying Lemma \ref{lem35} with $\caM_1=\caM$, $\caM_2=\rho(\bh)$
and $\Gamma_1=\Gamma_0$, we immediately obtain (i).
Because $\{E_{\mu,\nu}\}_{\mu,\nu\in\calP}$ are matrix units spanning $\caM$
and satisfying \eqref{grae},
we see from (i)
that
\begin{align}
\rho(\caB(\caH))'=\caM^{(0)}+\caM^{(1)}\Gamma_0\Gamma
= \mathrm{span}_{\mu,\nu\in \caP}\left\{
E_{\mu,\nu}\lmk \Gamma_0\Gamma\rmk^{|\mu|+|\nu|}\right\}
= \mathrm{span}_{\mu,\nu\in \caP}
\left\{\hat E_{\mu,\nu}\right\}.
\end{align}
Because
$\Gamma_0\Gamma$ commutes with $E_{\mu,\nu}$,
it is straight forward to check that $\{\hat E_{\mu,\nu}\}_{\mu,\nu\in\caP}$
are matrix units.
Hence we obtain (ii).

For part (iii), we first note that because $E_{\mu,\mu}$ is even, 
$[\rho(x), E_{\mu,\mu}] = 0$ for all $x \in \calB(\calH)$. Therefore 
 there is a well-defined 
$\ast$-homomorphism 
$$
  \rho_\mu : \calB(\calH) \to \calB(E_{\mu,\mu} \calH), \qquad 
  \rho_\mu(x) = \rho(x) E_{\mu,\mu}, \,\,\, x\in \calB(\calH).
$$
Because $\calB(\calH)$ is a factor, $\rho_\mu$ is injective.
To see that $\rho_\mu$ is surjective, we note that $E_{\mu\mu}$ is a minimal projection
of $\caM$ and it is even. 
Then applying Lemma \ref{lem35} with $\caM_1=\caM$ and $\caM_2=\rho(\caB(\caH))$,
we obtain
$
\rho(\bh)\cdot E_{\mu\mu}=\caB(E_{\mu\mu}\caH)
$
and so $\rho_\mu$ is surjective.
\end{proof}
We now prove Proposition \ref{cuntzpropi}, which we split into two lemmas. 
We recall the matrix units $\{E_{\mu,\nu}\}_{\mu,\nu\in\calP}\subset \caM$ 
and $\hat{E}_{\mu,\nu}= E_{\mu,\nu} (\Gamma_0 \Gamma)^{|\mu|+|\nu|}$ from 
Lemma \ref{lem:endo_and_commutant_even}.
\begin{lem} [First part of Proposition \ref{cuntzpropi}]\label{lemma:Even_cuntz_property}
 Consider the setting of Proposition \ref{cuntzpropi}.
Then there exist isometries $\{S_\mu\}_{\mu \in \calP}$ on $\caH$  with the property that
  for all $\mu,\nu \in\calP$ and $x \in \calB(\calH)$,
  \begin{align}\label{sprop}
  S_\nu^* S_\mu = \delta_{\mu,\nu} \unit, \qquad S_\mu S_\nu^* = \hat{E}_{\mu,\nu}, \qquad 
     \rho(x) E_{\mu,\mu}= S_\mu x S_\mu^*, \qquad \rho(x) = \sum_{\mu\in\calP} S_\mu x S_\mu^*.
  \end{align}
  The operators $S_\mu$ have 
  homogeneous parity and are such that 
  $\Ad_\Gamma\lmk S_\mu\rmk = (-1)^{|\mu|+\sigma_0}S_\mu$,
  with some uniform $\sigma_0\in\{0,1\}$.
They also satisfy $\Gamma_0 S_\mu = (-1)^{|\mu|}S_\mu$.
\end{lem}
\begin{proof}
By part (iii) of Lemma \ref{lem:endo_and_commutant_even}, 
$\rho_\mu$ in \eqref{rmu1} is a $*$-isomorphism $\calB(\calH) \xrightarrow{\rho_\mu} \calB( {E}_{\mu,\mu} \calH)$. 
Therefore we can apply Wigner's Theorem to obtain a unitary 
$w_\mu: \calH \to {E}_{\mu,\mu} \calH$ such that 
$ \rho_\mu = \mathrm{Ad}_{w_\mu}$.
Note that
$$
    w_\mu^* w_\nu = w_\mu^* E_{\mu,\mu} E_{\nu,\nu} w_\nu =  \delta_{\mu,\nu} \, \unit_{\calH}, \quad
   \mu,\,\nu \in \calP.  
$$
We also see that, because $\sum_\mu E_{\mu,\mu} = \unit$,
\begin{align*}
 \rho(x) = \sum_\mu \rho(x) E_{\mu,\mu} = \sum_\mu \rho_\mu(x) = \sum_\mu w_\mu x w_\mu^*,\quad
 x\in\bh.
\end{align*}
We use the above property to compute that for any $x \in\calB(\calH)$,
\begin{align*}
  w_\mu w_\nu^* \rho(x) &= w_\mu w_\nu^* \big( \sum_\lambda w_\lambda x w_\lambda^* \big) 
   = w_\mu x w_\nu^* 
   = \big( \sum_{\lambda} w_\lambda x w_\lambda^*\big) w_\mu w_\nu^* 
   = \rho(x) w_\mu w_\nu^*.
\end{align*}
Therefore $w_\mu w_\nu^* \in \rho(\calB(\calH))'$  for any $\mu, \nu \in \calP$.

Summarizing our results so far, we have obtained a collection of operators 
$\{w_\mu w_\nu^*\}_{\mu,\nu \in \calP}$ in 
$\rho(\calB(\calH))'$ such that 
\begin{equation} \label{eq:Smu_prop}
  \hat{E}_{\mu,\mu} w_\mu w_\nu^* \hat{E}_{\nu,\nu} = w_\mu w_\nu^* .
\end{equation}
From \eqref{eq:Smu_prop}, 
and (ii) of Lemma \ref{lem:endo_and_commutant_even},
that there is some 
$c_{\mu\nu} \in \C$ such that 
$$
   w_\mu w_\nu^* = c_{\mu \nu} \hat{E}_{\mu,\nu}.
$$
Note that $c_{\mu\nu}=\overline{c_{\nu\mu}}$.
Because of the definition, we have $w_\mu w_\mu^*=\hat E_{\mu\mu}$ and
we see that $c_{\mu\mu} = 1$. 
On the other hand, because of $w_\nu^* w_\nu=\unit_{\caH}$,
we have
$$
  c_{\mu \lambda} \hat{E}_{\mu,\lambda} = w_\mu w_\lambda^* 
  = w_\mu w_\nu^* w_\nu w_\lambda^* = c_{\mu \nu} c_{\nu \lambda} \hat{E}_{\mu,\lambda}
$$
and so $c_{\mu \lambda} = c_{\mu \nu} c_{\nu \lambda}$. 
In particular, $1=c_{\mu \mu} = c_{\mu \nu} c_{\nu \mu}=|c_{\mu\nu}|^2$ and so
 $c_{\mu\nu}\in\bbT$.
Now setting $\mu_0 :=\emptyset\in \calP$ and define $S_\mu = c_{\mu_0 \mu} w_\mu$ for 
every $\mu \in \calP$. Then because of the above properties of $c_{\mu\nu}$,
the collection $\{S_\mu\}_{\mu\in \calP}$ has the same algebraic 
properties as $\{w_\mu\}$ as well as that 
$S_\mu S_\nu^* = \hat{E}_{\mu,\nu}$ as required.
Hence we obtain \eqref{sprop}.

Next, we recall the grading operator $\Gamma_0 = \sum_\mu (-1)^{|\mu|} E_{\mu,\mu}$ of 
$\calM$.  Because $S_\mu$ is an isometry onto $E_{\mu,\mu} \calH$
$$
\Gamma_0 S_\mu = \Gamma_0  E_{\mu,\mu} S_\mu = (-1)^{|\mu|} E_{\mu,\mu} S_\mu = (-1)^{|\mu|} S_\mu.
$$
We now consider the grading of $S_\mu$, $\mathrm{Ad}_{\Gamma}(S_\mu)$. We compute 
that for any $x \in\calB(\calH)$,
\begin{align*}
  \Gamma S_\mu x S_\mu^* \Gamma &= \Gamma \rho(x) E_{\mu,\mu} \Gamma 
  = \Gamma \rho(x) \Gamma E_{\mu,\mu} = \rho( \Gamma x \Gamma) E_{\mu,\mu} = 
  S_\mu \Gamma x \Gamma S_\mu^*
\end{align*}
as $E_{\mu,\mu}$ is even and $\rho$ commutes with the grading. 
Multiplying $\Gamma S_\mu^*$ from the left and
$\Gamma S_\mu$ from the right of the equation,  we see 
that $\Gamma S_\mu^* \Gamma S_\mu \in \calB(\calH)' = \C \unit_\calH$. 
Note that  $\Gamma S_\mu^* \Gamma S_\mu$ is unitary because
$\Ad_\Gamma(E_{\mu\mu})=E_{\mu\mu}$.
So 
$S_\mu^* \Gamma S_\mu= e^{i\varphi} \Gamma$ with some $e^{i\varphi}\in \bbT$. 
Multiplying $S_\mu$ from the left,  and by $\Gamma$ from the right,
we obtain
$\Gamma S_\mu\Gamma=E_{\mu\mu}\Gamma S_\mu\Gamma= S_\mu S_\mu^* \Gamma S_\mu\Gamma=e^{i\varphi} S_\mu$.
But because $(\Ad_{\Gamma})^2=\id$, $(e^{i\varphi})^2=1$ and 
 $\Ad_{\Gamma}(S_\mu) = (-1)^{b_\mu} S_\mu$ with some $b_\mu=0,1$.

Let us further examine the grading of the operator $S_\mu$. 
We compute that
\begin{align*}
    \Gamma \hat{E}_{\mu,\nu} \Gamma 
  &= \Gamma E_{\mu,\nu} (\Gamma_0 \Gamma)^{|\mu|+|\nu|} \Gamma 
  = \Gamma E_{\mu,\nu} \Gamma (\Gamma_0 \Gamma)^{|\mu|+|\nu|}  \\
  &= (-1)^{|\mu|+|\nu|} E_{\mu,\nu}  (\Gamma_0 \Gamma)^{|\mu|+|\nu|}  
  = (-1)^{|\mu|+|\nu|} \hat{E}_{\mu,\nu}
\end{align*}
while we also find 
\begin{align*}
 \Gamma \hat{E}_{\mu,\nu} \Gamma  &= \Gamma S_\mu \Gamma \Gamma S_\nu^* \Gamma 
 = (-1)^{b_\mu} (-1)^{b_\nu} S_\mu S_\nu^* 
  = (-1)^{b_\mu+ b_\nu} \hat{E}_{\mu,\nu}.
\end{align*}
Therefore $|\mu|+|\nu| = b_\mu+b_\nu \in  \Z_2$. 
By setting $\mu_0:=\emptyset\in \calP$ and $\sigma_0:=b_{\mu_0}$,
we have that $\Gamma S_\mu \Gamma$ = $(-1)^{|\mu|+\sigma_0}S_\mu$ for 
all $\mu \in \calP$.
\end{proof}
\begin{lem}[Second half of Proposition \ref{cuntzpropi}] \label{lem:Even_Cuntz2}
Consider the setting of Proposition \ref{cuntzpropi}.
For $S_\mu$ of Lemma \ref{lemma:Even_cuntz_property},
set $B_\mu:=(\Gamma_0\Gamma)^{|\mu|} S_\mu$, for $\mu\in\caP$.
Then $B_\nu^* B_\mu = \delta_{\mu,\nu} \unit$,
$B_\mu B_\nu^*= E_{\mu,\nu}$,
\begin{align}
  \rho(x) 
  = \sum_{\mu\in\calP} \mathrm{Ad}_{B_\mu} \circ \mathrm{Ad}_{\Gamma^{|\mu|}} (x)
  &=\sum_{\mu\in\calP} \mathrm{Ad}_{\Gamma^{|\mu|}}\circ\mathrm{Ad}_{B_\mu} ,\quad
  x\in\bh,  \nonumber \\
 \label{eq:rho_and_E_using_V_even_again}
  E_{\mu_0,\nu_0} \rho(E_{\mu_1,\nu_1}) \cdots \rho^N (E_{\mu_N,\mu_N}) &= 
  (-1)^{\sum\limits_{k=1}^N (|\mu_k|+|\nu_k|) \sum\limits_{j=0}^{k-1} |\nu_j|} 
  B_{\mu_0} \cdots B_{\mu_N} B_{\nu_N}^* \cdots B_{\nu_0}^*,
\end{align}
for all $N\in\bbN\cup\{0\}$ and $\mu_0,\ldots\mu_N,\nu_0,\ldots,\nu_N\in\caP$.
The operators $B_\mu$ have 
  homogeneous parity and are such that 
  $\Ad_\Gamma\lmk B_\mu\rmk = (-1)^{|\mu|+\sigma_0}B_\mu$,
  with the same $\sigma_0\in\{0,1\}$
  as in Lemma
  \ref{lemma:Even_cuntz_property}.
If there are isometries $\{T_\mu\}_{\mu \in \calP}$ satisfying \eqref{buni}, then 
there is some $c\in\bbT$ such that $T_\mu=cB_\mu$, for all $\mu\in \caP$.
\end{lem}
\begin{proof}
From Lemma \ref{lemma:Even_cuntz_property},
we check that 
\begin{align}\label{bbn}
  B_\mu^* B_\nu = S_\mu^* (\Gamma_0 \Gamma)^{|\mu|+|\nu|} S_\nu 
  = S_\mu^* S_\nu  \Gamma^{|\mu|+|\nu|}(-1)^{(|\nu|+\sigma_0)(|\mu|+|\nu|)} (-1)^{|\nu|(|\mu|+|\nu|)} 
  = \delta_{\mu,\nu} \one.
\end{align}
We also have from Lemma \ref{lemma:Even_cuntz_property} that 
\begin{align}\label{bbe}
  B_\mu B_\nu^* &= (\Gamma_0 \Gamma)^{|\mu|} S_\mu S_\nu^* (\Gamma_0 \Gamma)^{|\nu|} 
  = (\Gamma_0 \Gamma)^{|\mu|} E_{\mu,\nu} (\Gamma_0\Gamma)^{|\mu|+|\nu|} (\Gamma_0 \Gamma)^{|\nu|} 
  = E_{\mu,\nu}
\end{align}
because $\Gamma_0 \Gamma$ commutes with $\caM$.
Because $S_\mu$ has homogenous parity and $\Gamma_0 \Gamma$ is even,
$B_\mu$ has the same homogeneous parity as $S_\mu$. 
In particular, 
  $\Ad_\Gamma\lmk B_\mu\rmk = (-1)^{|\mu|+\sigma_0}B_\mu$,
   with the same $\sigma_0\in\{0,1\}$
  as in Lemma
  \ref{lemma:Even_cuntz_property}.
This implies that the 
endomorphism $\mathrm{Ad}_{B_\mu}$ respects the grading 
on $\calB(\caH)$, i.e., $\Ad_\Gamma\circ{\Ad}_{B_\mu}={\Ad}_{B_\mu}\circ\Ad_\Gamma$. 
Furthermore, using that $\Gamma_0 S_\mu = (-1)^{|\mu|}S_\mu$, 
$\mathrm{Ad}_{S_\mu} = \mathrm{Ad}_{\Gamma^{|\mu|}B_\mu}= \mathrm{Ad}_{B_\mu \Gamma^{|\mu|}} $. 
We therefore see that for $x\in\calB(\caH)$
\begin{align*}
  \rho(x) &= \sum_{\mu \in\calP} S_\mu x S_\mu^* 
  = \sum_{\mu\in\calP} \mathrm{Ad}_{B_\mu} \circ \mathrm{Ad}_{\Gamma^{|\mu|}} (x).
\end{align*}
A simple induction argument using that $\mathrm{Ad}_{B_\mu}$ commutes with 
$\mathrm{Ad}_\Gamma$ gives that 
\begin{equation}  \label{eq:iterated_rho_even_V}
  \rho^N(x) = \sum_{\lambda_0,\ldots,\lambda_{N-1} \in\calP} \!\! 
  \mathrm{Ad}_{B_{\lambda_0}\cdots B_{\lambda_{N-1}}} \circ 
   \mathrm{Ad}_{\Gamma^{|\lambda_0|+\cdots |\lambda_{N-1}|}} (x).
\end{equation}
We now consider $\rho(E_{\mu,\nu})$. Recalling \eqref{bbe} and that 
$\mathrm{Ad}_\Gamma (E_{\mu,\nu}) = (-1)^{|\mu|+|\nu|} E_{\mu,\nu}$, we see that 
$$
 \rho(E_{\mu,\nu}) = \sum_\lambda B_\lambda \Gamma^{|\lambda|} E_{\mu,\nu} \Gamma^{|\lambda|} B_\lambda^* 
 = \sum_\lambda (-1)^{|\lambda|(|\mu|+|\nu|)} B_\lambda B_\mu B_\nu^* B_\lambda^*.
$$
From this, \eqref{bbe} and \eqref{bbn}, we have
$$
  E_{\mu_0,\nu_0} \rho( E_{\mu_1,\nu_1} )
  = B_{\mu_0} B_{\nu_0}^* \sum_\lambda (-1)^{|\lambda|(|\mu_1|+|\nu_1|)} 
    B_\lambda B_{\mu_1} B_{\nu_1}^* B_\lambda^* 
  = (-1)^{|\nu_0|(|\mu_1|+|\nu_1|)} B_{\mu_0}  B_{\mu_1} B_{\nu_1}^* B_{\nu_0}^*.
$$
This proves Equation \eqref{eq:rho_and_E_using_V_even_again} in the case of $N=1$. We 
now assume the equality is true for $N$ and consider $N+1$. Using 
Equation \eqref{bbn}, \eqref{bbe}, \eqref{eq:iterated_rho_even_V}, we compute that
\begin{align*}
  &E_{\mu_0,\nu_0} \rho(E_{\mu_1,\nu_1}) \cdots \rho^N(E_{\mu_N,\nu_N}) \rho^{N+1}(E_{\mu_{N+1},\nu_{N+1}}) \\
  &=
    (-1)^{\sum\limits_{k=1}^N (|\mu_k|+|\nu_k|) \sum\limits_{j=0}^{k-1} |\nu_j|} 
  B_{\mu_0} \cdots B_{\mu_N} B_{\nu_N}^* \cdots B_{\nu_0}^* \, \rho^{N+1}(E_{\mu_{N+1},\nu_{N+1}})  \\
  &= (-1)^{\sum\limits_{k=1}^N (|\mu_k|+|\nu_k|) \sum\limits_{j=0}^{k-1} |\nu_j|} 
  B_{\mu_0} \cdots B_{\mu_N} B_{\nu_N}^* \cdots B_{\nu_0}^* 
  \Big( \sum_{\lambda_0,\ldots,\lambda_N} \mathrm{Ad}_{B_{\lambda_0}\cdots B_{\lambda_N}} 
  \circ \mathrm{Ad}_{\Gamma^{\sum\limits_{j=0}^N |\lambda_j|}} ( E_{\mu_{N+1},\nu_{N+1}}) \Big) \\
  &= (-1)^{\sum\limits_{k=1}^N (|\mu_k|+|\nu_k|) \sum\limits_{j=0}^{k-1} |\nu_j|} 
  B_{\mu_0} \cdots B_{\mu_N}     
  \big( (-1)^{(|\mu_{N+1}|+|\nu_{N+1}|)\sum\limits_{j=0}^N |\nu_j|} B_{\mu_{N+1}}B_{\nu_{N+1}}^*\big)  
    B_{\nu_N}^* \cdots B_{\nu_0}^* \\
  &= (-1)^{\sum\limits_{k=1}^{N+1} (|\mu_k|+|\nu_k|) \sum\limits_{j=0}^{k-1} |\nu_j|} 
    B_{\mu_0} \cdots B_{\mu_N}B_{\mu_{N+1}}B_{\nu_{N+1}}^* B_{\nu_N}^* \cdots B_{\nu_0}^* 
\end{align*}
as required.

To show the last statement, suppose $\{{T}_\mu\}_{\mu\in\calP} \subset \calB(\caH)$
satisfy \eqref{buni}.
 Because
  \begin{align}
  \sum_{\lambda\in\calP} \mathrm{Ad}_{T_\lambda} \circ \mathrm{Ad}_{\Gamma^{|\lambda|}} (x)
  =\rho(x) 
  = \sum_{\lambda\in\calP} \mathrm{Ad}_{B_\lambda} \circ \mathrm{Ad}_{\Gamma^{|\lambda|}} (x),\quad
  x\in\bh,
\end{align}
multiplying $T_\nu^*$ from the left and by $B_\nu$ from the right,
we obtain
 \begin{align}
   \mathrm{Ad}_{\Gamma^{|\nu|}} (x)\cdot T_\nu^*B_\nu
  = T_\nu^* B_\nu\cdot \mathrm{Ad}_{\Gamma^{|\nu|}} (x),\quad
  x\in\bh.
\end{align}
Hence we obtain $T_\nu^* B_\nu\in \bbC\unit_{\caH}$, i.e., we have
$T_\nu^* B_\nu=c_\mu\unit_{\caH}$ for some $c_\mu\in\bbC$.
We then have
\begin{align}
B_\nu=
E_{\nu\nu}B_{\nu}=T_\nu T_\nu^*B_\nu
=c_\nu T_\nu.
\end{align}
By $B_\nu^* B_\nu=T_\nu^* T_\nu=\unit_{\caH}$,
we see that $c_\nu\in\bbT$.
Furthermore, 
from $B_\mu B_\nu^*=T_\mu T_\nu^*=E_{\mu\nu}$,
we see that $c_\mu=c_\nu=:c\in\bbT$.
\end{proof}
Lemma \ref{lemma:Even_cuntz_property} and \ref{lem:Even_Cuntz2} 
complete the proof of Proposition \ref{cuntzpropi}. 
We are ready to show Theorem \ref{cuntzthm0}.
\begin{proofof}[Theorem \ref{cuntzthm0}]
We fix a $W^*$-$(G,\mpp)$-dynamical system 
$(\caR_{0,\caK},\Ad_{\Gamma_\caK},\Ad_{V_g})\in\caS_{0}$ 
that is equivalent to $(\pi_\omega(\al_R)'', \Ad_{\Gamma_\omega}, \hat\alpha_\omega)$ and 
the endomorphism 
$\rho$ of Lemma \ref{iorho}.
Then the Hilbert space $\caK\otimes\cct$,  self-adjoint unitary $\Gamma_\caK$,
finite type I factor $\iota\circ\pi_\omega(\al_{\{0\}})$
with matrix units
$\{\iota\circ\pi_\omega\circ \big( E_{\mu,\nu}^{(0)}\big) \}_{\mu,\nu\in\calP}
\subset \calB(\caK \otimes \bbC^2)$
and $\rho$ satisfy the hypothesis of Proposition \ref{cuntzpropi}.
Applying Proposition \ref{cuntzpropi}, we obtain the isometries $\{B_\mu\}$ 
such that $B_\mu^* B_\nu = \delta_{\mu,\nu}\unit$ and 
which satisfy \eqref{bimp} and \eqref{bbbo} from the statement of the Theorem. 

To show  \eqref{sat},
set $\Gamma_0 := \iota\circ\pi_\omega \big( \mathfrak\Gamma(-\unit) \big)=
\sum_\mu (-1)^{|\mu|} \iota\circ\pi_\omega \big( E_{\mu,\mu}^{(0)} \big)$.
We claim for a homogeneous $x\in \caR_{0,\caK}$, 
and $N\in\bbN$ that
\begin{align}\label{pn}
T^N_{\bf B}(x) = {\Gamma}_0^{{\partial x}} \rho({\Gamma}_0^{{\partial x}} ) \cdots \rho^{N-1}({\Gamma}_0^{{\partial x}}) \rho^N(x).
\end{align}
First set $\Gamma_1 := \sum_\mu \Gamma_{\caK}^{|\mu|} \iota\circ\pi_\omega \big( E_{\mu,\mu}^{(0)} \big)$, 
which is a self-adjoint unitary.
Because of \eqref{bbbo} with $N=0$, we have
$\iota\circ\pi_\omega \big(E_{\mu\mu}^{(0)}\big)=B_\mu B_\mu^*$.
Therefore, we have
 \begin{align}
  \rho\circ\iota\circ\pi_\omega(A)
  = \sum_{\mu\in\calP} \mathrm{Ad}_{\Gamma_{\caK}^{|\mu|}B_\mu}\circ\iota\circ\pi_\omega (A)
  =\Ad_{\Gamma_1}\circ T_{\bf B}\circ\iota\circ\pi_\omega (A),\quad
  A\in\al_R.
\end{align}
Hence we obtain for any homogeneous $x\in\caR_{0,\caK}$, 
\begin{align}\label{tbx}
T_{\bf B}(x)=\Ad_{\Gamma_1}\circ \rho(x)
&=\sum_{\mu,\nu}
\Gamma_{\caK}^{|\mu|} 
\lmk\iota\circ\pi_\omega \big( E_{\mu,\mu}^{(0)} \big) \rmk
\rho(x)
\Gamma_{\caK}^{|\nu|}\lmk  \iota\circ\pi_\omega \big( E_{\nu,\nu}^{(0)} \big) \rmk \\
&=\sum_{\mu}\iota\circ\pi_\omega\big( E_{\mu,\mu}^{(0)}\big)
\Gamma_{\caK}^{|\mu|} 
\rho(x)
\Gamma_{\caK}^{|\mu|}. \nonumber \\
&=\sum_{\mu}\iota\circ\pi_\omega \big( E_{\mu,\mu}^{(0)}\big)
\rho\circ\Ad_{\Gamma_{\caK}^{|\mu|}}(x). \nonumber \\
&=\sum_{\mu}\iota\circ\pi_\omega \big( E_{\mu,\mu}^{(0)} \big)
(-1)^{|\mu|\partial x}
\rho(x)
=\Gamma_0^{\partial x}\rho(x),  \nonumber
\end{align}
where in the third equality we used that $ \iota\circ\pi_\omega \big( E_{\nu,\nu}^{(0)} \big)$ commutes
with $\Gamma_\caK$ and elements from $\rho(\caR_{0,\caK})$.
This proves  \eqref{pn} for the case $N=1$.
Now we proceed by induction and 
suppose that \eqref{pn} holds for $N$.
Then using \eqref{tbx} and the induction assumption, for any homogeneous $x\in\caR_{0,\caK}$, 
\begin{align}
T^{N+1}_{\bf B}(x) 
&= T_{\bf B}\lmk
{\Gamma}_0^{{\partial x}} \rho({\Gamma}_0^{{\partial x}} ) \cdots \rho^{N-1}({\Gamma}_0^{{\partial x}}) \rho^N(x)\rmk\nonumber\\
&=\Gamma_0^{\partial \lmk{\Gamma}_0^{{\partial x}} \rho({\Gamma}_0^{{\partial x}} ) \cdots \rho^{N-1}({\Gamma}_0^{{\partial x}}) \rho^N(x)\rmk } 
\rho( {\Gamma}_0^{{\partial x}} )
\rho^2({\Gamma}_0^{{\partial x}} ) \cdots \rho^{N}({\Gamma}_0^{{\partial x}}) \rho^{N+1}(x) \nonumber \\
&=
\Gamma_0^{\partial x }\rho({\Gamma}_0^{{\partial x}} )
\rho^2({\Gamma}_0^{{\partial x}} ) \cdots \rho^{N}({\Gamma}_0^{{\partial x}}) \rho^{N+1}(x).
\end{align}
Hence  \eqref{pn} holds for $N+1$ and proves the claim.

Now we show \eqref{sat}.
Because $\kappa_\omega = 0$, $\pi_\omega(\al_R)''$ is a factor. 
Therefore, for 
 any homogeneous $x\in\po(\al_R)''$, the sequence 
\begin{align}
T^N_{\bf B}\circ\iota(x)
&= {\Gamma}_0^{{\partial x}} \rho({\Gamma}_0^{{\partial x}} ) \cdots \rho^{N-1}({\Gamma}_0^{{\partial x}}) \rho^N\circ\iota(x) \nonumber\\
&=\iota\circ\lmk
\pi_\omega\lmk
{\mathfrak\Gamma}(-\unit)^{\partial x}
\beta_{S_1}\big({\mathfrak\Gamma}(-\unit)^{\partial x}\big) \cdots
\beta_{S_{N-1}}\big({\mathfrak\Gamma}(-\unit)^{\partial x}\big) \rmk \hat \beta_{S_N}(x)\rmk
\end{align}
converges to
$\braket{\Oo}{x\Oo}\unit_{\caK\otimes\cct}$ in the $\sigma$-weak topology
by Lemma \ref{lem:trans_weak_conv_to_state}.
This proves \eqref{sat}.

To prove \eqref{gvb} set
\begin{align}
 T_\nu := \sum_{\lambda\in\calP} \ol{ \langle \psi_\lambda , \mathfrak{\Gamma}(U_g)^* \psi_\nu \rangle }^{{\mpp(g)}}
    V_g B_\lambda V_g^* ,\quad \nu\in\caP.
   \end{align}
 Recall that  for $c\in\bbC$, $\overline{c}^{\mpp(g)} = c$ for $\mpp(g)=0$ 
 and  $\bar c$ if ${\mpp}(g)=1$. 
     We claim that $\{T_\mu\}_{\mu\in\caP}$ satisfies \eqref{buni}
     with $E_{\mu\nu}$ and $\Gamma$ replaced by
     $\iota\circ\pi_\omega(E_{\mu\nu}^{(0)})$ and $\Gamma_\caK$ respectively.
We compute 
\begin{align*}
  T_\mu^* T_\nu &= \sum_{\lambda,\zeta} \ol{ \langle \psi_\lambda , \mathfrak{\Gamma}(U_g)^* \psi_\mu \rangle }^{{\mpp(g)}+1} 
    \ol{ \langle \psi_\zeta , \mathfrak{\Gamma}(U_g)^* \psi_\nu \rangle }^{{\mpp(g)}}  V_g B_\lambda^* B_\zeta V_g^* \\
   &= \sum_{\lambda} \ol{ \langle \mathfrak{\Gamma}(U_g)^* \psi_\mu , \psi_\lambda \rangle \, 
     \langle \psi_\lambda , \mathfrak{\Gamma}(U_g)^* \psi_\nu \rangle}^{{\mpp(g)}}\unit 
   = \delta_{\mu,\nu} \, \one.
\end{align*}
To see the second property of \eqref{buni},
note that
 \begin{align*}
   {\mathfrak{\Gamma}(U_g)}^* E_{\mu,\nu}^{(0)} {\mathfrak{\Gamma}(U_g)} &= 
   \sum_{\lambda,\zeta\in\calP} \ol{ \langle \psi_\nu, {\mathfrak{\Gamma}(U_g)} \psi_\zeta\rangle }^{\mpp(g)} \, 
   \langle \psi_\lambda, {\mathfrak{\Gamma}(U_g)}^* \psi_\mu \rangle E_{\lambda,\zeta}^{(0)}.
 \end{align*}
 Using this, we obtain
\begin{align*}
  T_\mu T_\nu^* &= \sum_{\lambda,\zeta} \ol{ \langle \psi_\lambda , \mathfrak{\Gamma}(U_g)^* \psi_\mu \rangle 
    \langle \mathfrak{\Gamma}(U_g)^* \psi_\nu, \psi_\zeta \rangle }^{{\mpp(g)}} 
    V_g B_\lambda B_\zeta^* V_g^* \\
    &= \sum_{\lambda,\zeta} \ol{ \langle \psi_\lambda , \mathfrak{\Gamma}(U_g)^* \psi_\mu \rangle 
    \langle \mathfrak{\Gamma}(U_g)^* \psi_\nu, \psi_\zeta \rangle }^{{\mpp(g)}} 
    \iota\circ \pi_\omega\big( \mathfrak{\Gamma}(U_g) E_{\lambda,\zeta}^{(0)}\mathfrak{\Gamma}(U_g)^* \big) \\
   &= \iota \circ \pi_\omega \circ \mathrm{Ad}_{\mathfrak{\Gamma}(U_g)} \Big( 
     \sum_{\lambda,\zeta} { \langle \psi_\lambda , \mathfrak{\Gamma}(U_g)^* \psi_\mu \rangle 
    \langle \mathfrak{\Gamma}(U_g)^* \psi_\nu, \psi_\zeta \rangle } E_{\lambda,\zeta}^{(0)} \Big) \\
  &= \iota \circ \pi_\omega \circ \mathrm{Ad}_{\mathfrak{\Gamma}(U_g)} \big( 
    \mathfrak{\Gamma}(U_g)^* E_{\mu,\nu}^{(0)} \mathfrak{\Gamma}(U_g) \big) 
   = \iota\circ \pi_\omega( E_{\mu,\nu}^{(0)}).
\end{align*}

To check the third property of \eqref{buni},
note that 
 $\langle \psi_\mu, {\mathfrak{\Gamma}(U_g)}^* \psi_\nu\rangle = 0$ if $|\mu|\neq |\nu|$, 
 because ${\mathfrak{\Gamma}(U_g)}$ commutes with $\mathfrak{\Gamma}(-\unit_{\bbC^d})$.
Using this, we check that 
\begin{align*}
  &\sum_{\mu\in\calP} \mathrm{Ad}_{T_\mu} \circ \mathrm{Ad}_{\Gamma_{\caK}^{|\mu|}} (\iota \circ \pi_\omega(A) )  \\
  &= \sum_{\mu} \sum_{\lambda,\zeta} \langle  \ol{\psi_\lambda, \mathfrak{\Gamma}(U_g)^* \psi_\mu\rangle }^{{\mpp(g)}} 
    \ol{ \langle \psi_\zeta, \mathfrak{\Gamma}(U_g)^* \psi_\mu \rangle }^{{\mpp(g)}+1}  
      \delta_{|\mu|,|\lambda|}
       V_g B_\lambda V_g^* \Gamma_{\caK}^{|\mu|} (\iota\circ \pi_\omega)(A) \Gamma_{\caK}^{|\mu|} 
        V_g B_\zeta^* V_g^* \\
  &= \sum_{\lambda,\zeta} \sum_\mu \ol{ \langle \psi_\lambda, \mathfrak{\Gamma}(U_g)^* \psi_\mu\rangle 
    \langle \mathfrak{\Gamma}(U_g)^* \psi_\mu, \psi_\zeta \rangle }^{{\mpp(g)}} 
    V_g B_\lambda V_g^* \Gamma_{\caK}^{|\lambda|} (\iota\circ \pi_\omega)(A) \Gamma_{\caK}^{|\lambda|} 
        V_g B_\zeta^* V_g^* \\
  &= \sum_\lambda V_g B_\lambda V_g^* \Gamma_{\caK}^{|\lambda|} (\iota\circ \pi_\omega)(A) \Gamma_{\caK}^{|\lambda|} 
        V_g B_\lambda^* V_g^* \\
  &= \sum_{\lambda} V_g   \big(\mathrm{Ad}_{B_\lambda} \circ \mathrm{Ad}_{\Gamma_{\caK}^{|\lambda|}}\big) 
   \big( \iota\circ \pi_\omega \circ \alpha_g^{-1}(A) \big) V_g^*
\end{align*}
and recalling \eqref{rhotrans}, 
\begin{align*}
\sum_{\mu\in\calP} \mathrm{Ad}_{T_\mu} \circ \mathrm{Ad}_{\Gamma_{\caK}^{|\mu|}} (\iota \circ \pi_\omega(A) ) 
  &= \mathrm{Ad}_{V_g} \circ \rho\big( \iota\circ \pi_\omega \circ \alpha_g^{-1}(A) \big) \\
  &= \mathrm{Ad}_{V_g} \circ \iota\circ \pi_\omega \big(\beta_{S_1} \circ \alpha_g^{-1}(A) \big) \\
  &=  \iota\circ \pi_\omega \circ \alpha_g \circ \beta_{S_1} \circ \alpha_g^{-1}(A) 
  = \rho\circ \iota\circ\pi_\omega(A),
\end{align*}
for all $A\in\al_R$.
Hence we have proven that  $\{T_\mu\}_{\mu\in\caP}$ satisfies \eqref{buni}.
Applying Proposition \ref{cuntzpropi},
there is some $c_g \in \mathbb{T}$ such that $B_\mu = c_g T_\mu$ for all $\mu\in\caP$. Therefore 
\begin{align*}
  \sum_{\mu}\ol{c_g} \langle \psi_\mu, \mathfrak{\Gamma}(U_g)\psi_\nu \rangle 
   B_\mu
  &= \sum_{\mu,\lambda} \langle \psi_\mu, \mathfrak{\Gamma}(U_g)\psi_\nu \rangle
    \ol{ \langle \psi_\lambda, \mathfrak{\Gamma}(U_g)^* \psi_\mu \rangle }^{{\mpp(g)}} V_g B_\lambda V_g^* \\
  &= \sum_{\lambda,\mu} \ol{ \langle \mathfrak{\Gamma}(U_g)^* \psi_\mu, \psi_\nu \rangle 
     \langle \psi_\lambda, \mathfrak{\Gamma}(U_g)^* \psi_\mu \rangle }^{{\mpp(g)}} V_g B_\lambda V_g^* \\
  &= V_g B_\nu V_g^*. 
\end{align*}
Hence $\sum_\mu  \langle \psi_\mu, \, \mathfrak{\Gamma}(U_g) \psi_\nu \rangle B_\mu = c_g V_g B_\nu V_g^*$, 
which completes the proof.
\end{proofof}
\subsection{Case: $\kappa_\omega=1$}

We now consider endomorphisms on $W^*$-$(G,\mpp)$-dynamical systems that are 
equivalent to $(\caR_{1,\caK},\Ad_{\Gamma_\caK},\Ad_{V_g})\in\caS_{1}$ 
from Example \ref{ni}.
Recall that $\mathfrak{\Gamma}(U_g)$
denotes the second quantization of $U_g$ on $\calF(\C^d)$. Our 
aim is to prove the following.
\begin{thm}\label{cuntzthm1}
Let $\omega$ be 
 a pure $\alpha$-invariant and 
translation invariant split state on $\al$.
Suppose that the graded $W^*$-$(G,\mpp)$-dynamical system
$(\pi_\omega(\al_R)'', \Ad_{\Gamma_\omega}, \hat\alpha_\omega)$ associated to $\omega$
is equivalent to $(\caR_{1,\caK},\Ad_{\Gamma_\caK},\Ad_{V_g})\in\caS_{1}$ via
a $*$-isomorphism $\iota: \pi(\al_R)''\to \caR_{1,\caK}$.
Let $\rho$ be the $*$-endomorphism on $\caR_{1,\caK}$ given in Lemma \ref{iorho}.
Then
there is some $\sigma_0\in\{0,1\}$ such that
$\rho\lmk\unit_\caK\otimes \sigma_x\rmk=(-1)^{\sigma_0}\iota\circ\pi_\omega\lmk 
\mathfrak\Gamma(-\unit)\rmk\lmk\unit_\caK\otimes \sigma_x\rmk$
and a set of isometries $\{S_\mu\}_{\mu \in \calP}$ on $\caK$
such that
$S_\nu^* S_\mu = \delta_{\mu,\nu} \unit_\caK$,
  \begin{align}\label{bimpn}
  \rho\circ\iota\circ\pi_\omega(A)
  = \sum_{\mu\in\calP} \mathrm{Ad}_{\hat S_\mu}\circ\iota\circ\pi_\omega (A)
  ,\quad
  A\in\al_R, 
  \quad 
\end{align}
with $\hat S_\mu:=S_\mu \otimes \sigma_z^{\sigma_0+|\mu|}$ and 
\begin{align}\label{bbbon}
\iota\circ\pi_\omega \! \lmk E_{\mu_0,\nu_0}^{(0)} E_{\mu_1,\nu_1}^{(1)} 
\cdots E_{\mu_N,\mu_N}^{(N)}\rmk = 
  (-1)^{\sum\limits_{k=1}^N (|\mu_k|+|\nu_k|) \sum\limits_{j=0}^{k-1}\lmk \sigma_0+|\nu_j|\rmk} 
  S_{\mu_0} \cdots S_{\mu_N} S_{\nu_N}^* \cdots S_{\nu_0}^*\otimes \sigma_x^{\sum\limits_{i=0}^{N} |\mu_i|+|\nu_i|} 
\end{align}
for all $N\in\bbN\cup\{0\}$ and $\mu_0,\ldots\mu_N,\nu_0,\ldots,\nu_N\in\caP$.
Furthermore,  we have
\begin{align}\label{satn}
\sigma\mathrm{\hbox{-}weak}\lim_{N\to\infty}T_{\bf \hat S}^N\circ\iota(x)
=\braket{\Oo}{ x\Oo} \, \unit_{\caK\otimes\cct},\quad x\in\po(\caA_R)''.
\end{align}
  For each $g\in G$, there is some 
  $c_g \in \mathbb{T}$ such that 
   \begin{align}\label{gvbn}
  (-1)^{\mqq(g)|\nu|} \sum_{\mu\in\calP} \langle \psi_\mu, \, \mathfrak{\Gamma}(U_g) \psi_\nu \rangle S_\mu = 
     c_g V_g^{(0)} S_\nu (V_g^{(0)})^*, 
   \end{align}
where  $V_g^{(0)}$ is given in Lemma \ref{zv}.
\end{thm}
We again will prove this theorem in several steps. Parts of the 
proof follow the same argument as the case $\kappa_\omega =0$, so some details will be omitted.
\begin{prop}\label{cuntzpropii}
Let $\caK$ be a Hilbert space and set $\Gamma_\caK:=\unit_\caK\otimes\sigma_z$
on $\caK\otimes \cct$.
We give a grading to  $\caR_{1,\caK}=\bk\otimes \Clf$ by $\Ad_{\Gamma_\caK}$.
Suppose that $\caN$ is a type I subfactor of $\caR_{1,\caK}$
with matrix units $\{E_{\mu,\nu}\}_{\mu,\nu\in\calP} \subset \caN$ spanning 
$\caN$.
Assume that 
\begin{align}\label{graen}
\Ad_{\Gamma}(E_{\mu,\nu})=(-1)^{|\mu|+|\nu|} E_{\mu,\nu},\quad\text{for}\quad
\mu,\nu\in \caP.
\end{align}
Set  $\Gamma_0:=\sum_{\mu\in\caP} (-1)^{|\mu|}E_{\mu\mu}$.
Let $\rho: \caR_{1,\caK} \to \caR_{1,\caK}$ be an injective graded, unital $\ast$-endomorphism 
such that $\rho(a)b- (-1)^{\partial a\partial b} b \rho(a) = 0$ 
for $b\in\caN, a \in \caR_{1,\caK}$ with homogeneous 
grading. Suppose further that $\caR_{1,\caK}= \rho(\caR_{1,\caK}) \vee \caN$.

Then there is some $\sigma_0\in\{0,1\}$ such that
$\rho\lmk\unit_\caK\otimes \sigma_x\rmk=(-1)^{\sigma_0}\Gamma_0\lmk\unit_\caK\otimes \sigma_x\rmk$ 
and there exist isometries $\{S_\mu\}_{\mu \in \calP}$ on $\caK$ with the property that 
\begin{align}\label{sss}
S_\nu^* S_\mu = \delta_{\mu,\nu} \, \unit_{\caK}, 
    \qquad 
    \rho(b) = \sum_\mu \Ad_{( S_\mu \otimes \sigma_z^{\sigma_0+|\mu|})}(b) 
\end{align}
   for all $\mu,\nu \in\calP$ and $b \in \caR_{1,\caK}$. 
     Furthermore, for $N\in\bbN$, $\mu_0,\ldots,\mu_{N-1},\nu_0,\ldots,\nu_{N-1}\in\calP$, the identity
     \begin{align}\label{kansha}
&     E_{\mu_0,\nu_0}\rho(E_{\mu_1,\nu_1})  \rho^2(E_{\mu_2,\nu_2}) \cdots \rho^{N-1}(E_{\mu_{N-1},\nu_{N-1}}) \nonumber\\
&\hspace{1.5cm}  =  (-1)^{\sum\limits_{j=1}^{N-1} (\sum\limits_{k=0}^{j-1}(\sigma_0+|\nu_k|) ) (|\mu_j|+|\nu_j|)}
  S_{\mu_0}\cdots S_{\mu_{N-1}}S^*_{\nu_{N-1}} \cdots S_{\nu_0}^* \otimes \sigma_x^{\sum\limits_{i=0}^{N-1} |\mu_i|+|\nu_i|} 
     \end{align}
holds.

If there are isometries $\{T_\mu\}_{\mu \in \calP}$ on $\caK$ such that
\begin{align}\label{buni2}
 T_\nu^* T_\mu = \delta_{\mu,\nu} \unit_\caK, \quad T_\mu T_\nu^*\otimes \sigma_x^{|\mu|+|\nu|}= E_{\mu,\nu},
\quad\rho(b) 
  = \sum_{\mu\in\calP} \mathrm{Ad}_{T_\mu\otimes\sigma_z^{\sigma_0+|\mu|}}
   (b),\quad
  b\in \caR_{1,\caK},
\end{align}
then there is some $c\in\bbT$ such that $T_\mu=c S_\mu$, for all $\mu\in \caP$.
\end{prop}
To study the situation, we note the following general property.
\begin{lem}\label{rmrm}
Let $\caK$ be a Hilbert space and set ${\Gamma_\caK}:=\unit_\caK\otimes\sigma_z$
on $\caK\otimes \cct$.
We give a grading to  $\caR_{1,\caK}=\bk\otimes \Clf$ by $\Ad_{{\Gamma_\caK}}$.
Let $\caN$ and $\caM$ be
$\Ad_{{\Gamma_\caK}}$-invariant von Neumann subalgebras of 
$\caR_{1,\caK}=\bk\otimes \Clf$ satisfying
\begin{align}
ab-(-1)^{\partial a\partial b} ba=0,\quad\text{for homogeneous }\quad a\in \caN, \quad b\in \caM.
\end{align}
Suppose that $\caN$ is a type I factor with a self-adjoint unitary ${\Gamma_1}\in\caN$
satisfying $\Ad_{{\Gamma_1}}(a)=\Ad_{{\Gamma_\caK}}(a)$, for all $a\in\caN$.
Suppose 
$Z(\caM)^{(1)}\neq\{0\}$ 
and $\caN\vee \caM=\bk\otimes \Clf$.
Then the following holds.
\begin{enumerate}
\item[(i)] There are Hilbert spaces $\caH_1,\caH_2$, a unitary
$U:\caK\otimes\cct\to \caH_1\otimes\caH_2\otimes\cct$
and a self-adjoint unitary ${\tilde \Gamma_1}$ on $\caH_1$
such that
\begin{align}\label{ffs}
\mathrm{Ad}_U (\calN) &= \calB(\calH_1) \otimes \bbC\one_{\calH_2} \otimes \bbC\one_{\C^2}, 
   \qquad \mathrm{Ad}_U (\bk\otimes \Clf) = \calB(\calH_1 \otimes \calH_2) \otimes \Clf,\nonumber\\
    \mathrm{Ad}_{U}( {\Gamma_\caK}) &= 
    {\tilde \Gamma_1} \otimes \one_{\calH_2} \otimes \sigma_z, 
   \quad \mathrm{Ad}_{U}( {\Gamma_1}) = 
    {\tilde \Gamma_1} \otimes \one_{\calH_2} \otimes \unit_{\cct},
   \quad \mathrm{Ad}_{U}( \one_\calK \otimes \sigma_x) 
   = \unit_{\caH_1}\otimes \one_{\calH_2} \otimes \sigma_x,
\end{align}
and 
\begin{align}
&\mathrm{Ad}_{U}( \caM) 
    = \bbC\one_{\calH_1} \otimes \calB(\calH_2) \otimes \C \unit_{\cct}
    +   \bbC{\tilde \Gamma_1}\otimes \calB(\calH_2) \otimes \C \sigma_x.
    \label{muu}
\end{align}

\item[(ii)] $\caM'=\caN^{(0)}\lmk\bbC\unit_{\caK}\otimes \Clf\rmk 
+\caN^{(1)}\lmk\bbC\unit_{\caK}\otimes \Clf\rmk{\Gamma_1}{\Gamma_\caK}$.

\item[(iii)]
For any minimal projection $p$ of $\caN$ which is even,
we have $\caM\cdot p=\caB(q \caK)\otimes\Clf$
with $q$ a projection on $\caK$ satisfying $p=q\otimes \unit_{\bbC^2}$.
(Note that even $p$ is always of this form.)
\item[(iv)]
$Z(\caM)=\bbC\one_{\caK} \otimes \unit_{\cct}
+\bbC{\Gamma_1}\lmk \unit_{\caK}\otimes \sigma_x\rmk$.
\end{enumerate}
\end{lem}
\begin{proof}
(i) As $\caN$ is a type I factor, 
there are Hilbert spaces $\caH_1,\tilde \caH_2$, and a unitary
$\tilde U:\caK\otimes\cct\to \caH_1\otimes\tilde \caH_2$
such that $\Ad_{\tilde U}(\caN)=\caB(\caH_1)\otimes\bbC\unit_{\tilde \caH_2}$.
Because ${\Gamma_1}\in\caN$, there is a self-adjoint unitary ${\tilde \Gamma_1}$ on $\caH_1$
such that $\Ad_{\tilde U}({\Gamma_1})={\tilde \Gamma_1}\otimes\unit_{\tilde\caH_2}$.
Let $\caD:= \spa_\bbC \{\unit,{\Gamma_1}{\Gamma_\caK},(\unit_\caK\otimes \sigma_x),
{\Gamma_1}{\Gamma_\caK}(\unit_\caK\otimes \sigma_x)\}$, a $*$-subalgebra of $\caN'$.
Let ${\Gamma_1}{\Gamma_\caK}=e_{00}-e_{11}$ be a spectral decomposition of the self-adjoint 
unitary ${\Gamma_1}{\Gamma_\caK}$.
Set $e_{i,1-i}:=e_{ii}(\unit_\caK\otimes \sigma_x) e_{1-i,1-i}$, $i=0,1$.
Then  because ${\Gamma_1}{\Gamma_\caK}$ and $\unit_\caK\otimes \sigma_x$ anti-commute,
we can check that $\{e_{ij}\}_{i,j=0,1}$ are matrix units in $\caD$
spanning $\caD$.
Hence $\caD$ is a type I${}_2$ factor in $\caN'$ generated by the matrix units $\{e_{ij}\}_{i,j=0,1}$.
Therefore, there is a type I${}_2$ factor $\caD_1$ on $\tilde{\caH}_2$ such that
$\Ad_{\tilde U}(\caD)=\bbC\unit_{\caH_1}\otimes \caD_1$
and the generating matrix units $\{f_{ij}\}_{i,j=0,1}$
such that $\Ad_{\tilde U}(e_{ij})=\unit_{\caH_1}\otimes f_{ij}$.
Then there is a Hilbert space $\caH_2$ and a unitary $W:\tilde \caH_2\to \caH_2\otimes \bbC^2$
such that 
\begin{align}
\Ad_{W}(f_{ij})=\unit_{\caH_2}\otimes \hat e_{ij},\quad
\Ad_{W}(\caD_1)=\bbC\unit_{\caH_2}\otimes \Mat_2.
\end{align}
Here $\hat e_{ij}$ denotes the matrix unit of $2\times 2$ 
matrices $\Mat_2$ with respect to the standard basis of $\cct$.
Setting $U:=\lmk\unit_{\caH_1}\otimes W\rmk\tilde U: \caK\otimes\cct\to 
\caH_1\otimes \caH_2\otimes\cct$, we may check directly
that $U$, $\caH_1$, $\caH_2$, ${\tilde \Gamma_1}$ satisfy the 
\eqref{ffs}.

We now prove \eqref{muu}.
 Because 
$\caM^{(0)}$ is a von Neumann subalgebra of $\caN'\cap\lmk \bk\otimes\Clf\rmk$,
$\Ad_U\lmk \caM^{(0)}\rmk$ is a von Neumann subalgebra of
\begin{align}
\Ad_U(\caN')\cap \Ad_U\lmk\bk\otimes\Clf\rmk
=\bbC\unit_{\caH_1}\otimes \caB(\caH_2)\otimes \Clf.
\end{align}
Furthermore, because elements in $\caM^{(0)}$ are even with respect to
 $\Ad_{{\Gamma_\caK}}$, elements in $\Ad_U\lmk \caM^{(0)}\rmk$ are even with respect
 to $\Ad_{\Ad_U({\Gamma_\caK})}=\Ad_{{\tilde \Gamma_1}\otimes\unit_{\caH_2}\otimes \sigma_z}$.
 Therefore we have 
 $
 \Ad_U\lmk \caM^{(0)}\rmk
 \subset
 \bbC\unit_{\caH_1}\otimes \caB(\caH_2)\otimes \bbC\unit_{\cct}
 $.
 Hence there is a von Neumann subalgebra $\tilde \caM$ of
 $\caB(\caH_2)$ such that
 \begin{align}\label{mtm}
 \Ad_U \big(\caM^{(0)} \big) =\bbC\unit_{\caH_1}\otimes \tilde \caM\otimes \bbC\unit_{\cct}.
 \end{align}
 
 Next we consider $\Ad_U\lmk \caM^{(1)}\rmk$.
 We claim $(\unit_\caK\otimes \sigma_x) {\Gamma_1}\in Z(\caM)^{(1)}$.
 To see this, let $b\in Z(\caM)^{(1)}$ be a non-zero element, which exists because of the assumption, 
 and set $\tilde b=(\unit_\caK\otimes \sigma_x) {\Gamma_1}b$.
 Because $b\in Z(\caM)^{(1)}$, $\unit_\caK\otimes \sigma_x\in Z(\bk\otimes\Clf)$
and ${\Gamma_1}$ is an even element in $\caN$ implementing the grading on $\caN$,
we see that
\begin{align}
\tilde b\in \caN'\cap\caM'\cap\{{\Gamma_\caK}\}'
=\lmk \bk\otimes\Clf\rmk'\cap \{{\Gamma_\caK}\}'
=\bbC\unit_{\caK\otimes \cct}.
\end{align}
Hence $(\unit_\caK\otimes \sigma_x) {\Gamma_1}$ 
is proportional to $b\in Z(\caM)^{(1)}$, i.e. it belongs to  $Z(\caM)^{(1)}$, proving the claim.
From this and \eqref{mtm} we have
\begin{align}\label{mtmi}
\Ad_U(\caM^{(1)})
=\Ad_U \big( \caM^{(0)}(\unit_\caK\otimes \sigma_x) {\Gamma_1} \big)
=\bbC{\tilde \Gamma_1}\otimes\tilde\caM\otimes \bbC\sigma_x 
\end{align}
for $\tilde\caM$ in \eqref{mtm}. 
From \eqref{mtm} and \eqref{mtmi}, to show \eqref{muu},
 it suffices to show that $\tilde \caM=\caB(\caH_2)$. For any $a\in\tilde\caM'$ 
\begin{align*}
\Ad_{U^*}\lmk\unit_{\caH_1}\otimes a\otimes\unit_{\cct}\rmk
\in \big( \caM^{(0)}\big)' \cap \big( \caM^{(1)}\big)' \cap \caN'\cap \{{\Gamma_\caK}\}'
=\lmk
\bk\otimes\Clf
\rmk'\cap\{{\Gamma_\caK}\}'
=\bbC\unit_{\caK\otimes\cct}.
\end{align*}
Hence we obtain $a\in \bbC\unit_{\caH_2}$.
This proves that $\tilde \caM=\caB(\caH_2)$.

\vspace{0.1cm}

\noindent
(ii)
We associate a spatial grading 
to 
$\bbC\unit_{\caH_1}$ and
$\caB(\caH_2)\otimes \Clf$  
by 
${\tilde \Gamma_1}$ and $\unit_{\caH_2}\otimes \sigma_z$ respectively.
From \eqref{muu}, we see that $\Ad_U(\caM)$
is equal to the graded tensor product
 $\bbC\unit_{\caH_1}{\hox} \lmk \caB(\caH_2)\otimes \Clf\rmk$
 of $(\bbC\unit_{\caH_1},\caH_1, \tilde\Gamma_1)$
 and $(\caB(\caH_2)\otimes\Clf,\caH_2\otimes \cct,\unit_{\caH_2}\otimes \sigma_z)$.
By Lemma \ref{comgra},
its commutant $\Ad_U(\caM')$
is equal to
\begin{align}\label{inh}
\Ad_U(\caM')
&=\caB(\caH_1)^{(0)}\otimes \bbC\unit_{\caH_2}\otimes\Clf
+\caB(\caH_1)^{(1)}\otimes \bbC\unit_{\caH_2}\otimes\Clf\sigma_z \nonumber \\
&=\Ad_U \big( \caN^{(0)}\lmk\bbC\unit_{\caK}\otimes \Clf\rmk 
+\caN^{(1)}\lmk\bbC\unit_{\caK}\otimes \Clf\rmk{\Gamma_1}{\Gamma_\caK} \big),
\end{align}
where $\caB(\caH_1)$ is given a grading by ${\tilde \Gamma_1}$.
This proves the claim.

\vspace{0.1cm}

\noindent
(iii)
Let $p$ be a
 minimal projection $\caN$ which is even and hence 
of the form $p=q\otimes \unit_{\bbC^2}$
with $q$ a projection on $\caK$.
Then because $p\in \caN$ is minimal, we have
$\Ad_U(p)=r\otimes \unit_{\caH_2}\otimes \unit_{\cct}$
with a rank-one projection $r$ on $\caH_1$.
Because $p$ is even, $r$ is even with respect to $\Ad_{{\tilde \Gamma_1}}$.
Therefore, there is a $\sigma\in\{0,1\}$ such that
 ${\tilde \Gamma_1}r=(-1)^\sigma r$.
Substituting \eqref{muu}, we then obtain
\begin{align}
\Ad_U(\caM p)
    &= \bbC r \otimes \calB(\calH_2) \otimes \C \unit_{\cct}
    +   \bbC{\tilde \Gamma_1} r\otimes \calB(\calH_2) \otimes \C \sigma_x\nonumber\\
&=\bbC r \otimes \calB(\calH_2) \otimes \C \unit_{\cct}
    +   \bbC (-1)^\sigma r\otimes \calB(\calH_2) \otimes \C \sigma_x  \nonumber \\
&=\Ad_U\lmk p \lmk \caB(\caK)\otimes \Clf\rmk p\rmk
=\Ad_U\lmk \caB(q\caK)\otimes \Clf\rmk
\end{align}
as required.

\vspace{0.1cm}

\noindent
(iv)
From \eqref{muu} and \eqref{inh}, we have
\begin{align*}
&\Ad_U \big( Z(\caM)^{(0)} \big) \\
&\quad 
=\lmk \bbC\one_{\calH_1} \otimes \calB(\calH_2) \otimes \C \unit_{\cct}\rmk 
\cap \big(
\caB(\caH_1)^{(0)}\otimes \bbC\unit_{\caH_2}\otimes\bbC\unit_{\cct}
+\caB(\caH_1)^{(1)}\otimes \bbC\unit_{\caH_2}\otimes\bbC\sigma_x\sigma_z
 \big)
=\bbC\unit,
\end{align*}
and
\begin{align*}
\Ad_U \big( Z(\caM)^{(1)} \big) &= \big( \bbC{\tilde \Gamma_1}\otimes \calB(\calH_2) \otimes \C \sigma_x \big)
\cap
\big(
\caB(\caH_1)^{(0)}\otimes \bbC\unit_{\caH_2}\otimes\bbC\sigma_x
+\caB(\caH_1)^{(1)}\otimes \bbC\unit_{\caH_2}\otimes\bbC\sigma_z
\big) \\
&=\bbC{\tilde \Gamma_1}\otimes\bbC\unit_{\caH_2}\otimes\bbC\sigma_x
=\Ad_U\lmk \bbC{\Gamma_1}(\unit_{\caK}\otimes\sigma_x)\rmk.
\end{align*}
This proves the claim.
\end{proof}
We introduce some notation. Given a self-adjoint unitary $T$ on some 
Hilbert space, we write the $\pm 1$ eigenspace projections as
\begin{equation} \label{eq:pm_espace_proj}
   \bbP_\varepsilon(T) = \frac{ \one + (-1)^{\varepsilon} T}{2}, \quad \varepsilon \in \{0,1\}.
\end{equation}
Note that because we use the presentation of $\Z_2$ as an additive group, $\bbP_1(T)$ is the projection 
onto the \emph{negative} eigenspace. We also have that 
$T \bbP_\varepsilon(T) = (-1)^\varepsilon \bbP_\varepsilon(T) = \bbP_\varepsilon(T) T$.
\begin{proofof}[Proposition \ref{cuntzpropii}]
Because $\unit_\caK\otimes \sigma_x$ belongs to $Z(\caR_{1,\caK})^{(1)}$ and $\rho$ is graded,
$\rho(\unit_\caK\otimes \sigma_x)$ belongs to $ Z(\rho(\caR_{1,\caK}))^{(1)}$.
In particular, because $\rho$ is injective,  $ Z(\rho(\caR_{1,\caK}))^{(1)}$ is not zero.
Therefore, we satisfy the hypothesis of Lemma \ref{rmrm} with 
$\caM$ and $\Gamma_1$ replaced by $\rho(\caR_{1,\caK})$ and $\Gamma_0$ respectively. 
Applying the lemma, we have that
\begin{enumerate}
\item[(i)] $Z\lmk\rho(\caR_{1,\caK})\rmk=\bbC\unit+\bbC\Gamma_0\lmk\unit_\caK\otimes \sigma_x\rmk$.
\item[(ii)] For any $\mu\in\caP$, $E_{\mu\mu}=e_{\mu\mu}\otimes\unit_{\cct}$
with $e_{\mu\mu}$ a projection on $\caK$, 
$\rho\lmk\caR_{1,\caK}\rmk E_{\mu\mu}=\caB(e_{\mu\mu}\caK)\otimes \Clf$,
\item[(iii)]
$\rho\lmk\caR_{1,\caK}\rmk'=\caN^{(0)} \lmk\bbC\unit_\caK\otimes\Clf\rmk
+\caN^{(1)}\lmk\bbC\unit_\caK\otimes\Clf\rmk\Gamma_0\Gamma_\caK$.
\end{enumerate}
Because of (i), $\rho(\unit_\caK\otimes \sigma_x)$,
an odd self-adjoint unitary in $Z\lmk \rho(\caR_{1,\caK})\rmk$,
should be either $\Gamma_0\lmk\unit_\caK\otimes \sigma_x\rmk$ or
$-\Gamma_0\lmk\unit_\caK\otimes \sigma_x\rmk$.
Therefore,
 there is $\sigma_0\in\{0,1\}$ such that
 \begin{align}\label{rcc}
 \rho\lmk\unit_\caK\otimes \sigma_x\rmk=(-1)^{\sigma_0}\Gamma_0\lmk\unit_\caK\otimes \sigma_x\rmk.
\end{align}
By (ii), \eqref{rcc}, and the fact that $E_{\mu\mu}\in \caN^{(0)}$ commutes with $\rho(\caR_{1,\caK})$, 
for each $\mu\in\caP$, we have 
\begin{align}
\rho\lmk
 \lmk \bk\otimes \bbC\unit_{\cct}\rmk\cdot  \bbP_0\lmk\unit_\caK\otimes \sigma_x\rmk
\rmk E_{\mu\mu}
=
\rho\lmk
 \caR_{1,\caK} \bbP_0\lmk\unit_\caK\otimes \sigma_x\rmk
\rmk E_{\mu\mu}
=\caB(e_{\mu\mu}\caK)\otimes\bbC \bbP_{\sigma_0+|\mu|}\lmk \sigma_x\rmk.
\end{align}
Therefore,
there is a $*$-isomorphism
$\rho_\mu:\bk\to \caB(e_{\mu\mu}\caK)$
such that
\begin{align}\label{rmu}
\rho\lmk \lmk a\otimes \unit\rmk \cdot \bbP_0\lmk \unit_\caK\otimes \sigma_x\rmk\rmk E_{\mu\mu}
=\rho_\mu(a)\otimes \bbP_{\sigma_0+|\mu|}(\sigma_x),\quad a\in \caB(\caK).
\end{align}
Applying $\Ad_{\Gamma_\caK}$, we also get that
\begin{align}\label{rmu2}
\rho\lmk \lmk a\otimes \unit\rmk \cdot \bbP_1\lmk \unit_\caK\otimes \sigma_x\rmk\rmk E_{\mu\mu}
=\rho_\mu(a)\otimes \bbP_{\sigma_0+|\mu|+1}(\sigma_x),\quad a\in \caB(\caK).
\end{align}
From \eqref{rmu} and \eqref{rmu2},
we obtain 
\begin{align}\label{sonoi}
\rho\lmk  a\otimes \unit\rmk E_{\mu\mu}
=\rho_\mu(a)\otimes \unit_{\cct},\quad a\in \caB(\caK).
\end{align}
Furthermore, by \eqref{rcc}, we have
\begin{align}\label{sonon}
 \rho\lmk\unit_\caK\otimes \sigma_x\rmk E_{\mu\mu}
 =(-1)^{\sigma_0+|\mu|}\lmk e_{\mu\mu}\otimes \sigma_x\rmk.
\end{align}

By Wigner's Theorem, for each $\mu\in\caP$,
there is a unitary $T_\mu:\caK\to e_{\mu\mu}\caK$
such that 
\begin{align}\label{ttt}
T_\mu^*T_\nu=\delta_{\mu,\nu} \unit_{\caK},\quad T_{\mu}T_\mu^*=e_{\mu\mu},\quad \mu,\nu\in \caP,\quad
\Ad_{T_\mu}\lmk a\rmk=\rho_\mu(a),\quad a\in \bk.
\end{align}
From this, \eqref{sonoi} and \eqref{sonon},
we obtain
\begin{align}
\rho(b) E_{\mu\mu}=\Ad_{T_\mu\otimes \sigma_z^{\sigma_0+|\mu|}}\lmk b\rmk,\quad
b\in \caR_{1,\caK}.
\end{align}
Summing this over $\mu$, we obtain
\begin{align}\label{rrtran}
\rho(b) =\sum_{\mu\in \caP}\Ad_{T_\mu\otimes \sigma_z^{\sigma_0+|\mu|}}\lmk b\rmk,\quad
b\in \caR_{1,\caK}.
\end{align}
Multiplying $T_\nu T_\mu^*\otimes\sigma_z^{|\mu|+|\nu|}$ from left or right of \eqref{rrtran},
we obtain the same value for any $b\in \caR_{1,\caK}$.
Therefore, $T_\nu T_\mu^*\otimes\sigma_z^{|\mu|+|\nu|}$
belongs to $\rho(\caR_{1,\caK})'$.
By (iii), we then have
\begin{align}
T_\nu T_\mu^*\otimes\sigma_z^{|\mu|+|\nu|}
\in \rho(\caR_{1,\caK})'
=\caN^{(0)} \lmk\bbC\unit_\caK\otimes\Clf\rmk
+\caN^{(1)}\lmk\bbC\unit_\caK\otimes\Clf\rmk\Gamma_0\Gamma_\caK.
\end{align}
Hence if $|\mu|=|\nu|$, 
$
T_\nu T_\mu^*\otimes \unit_{\cct} \in \caN^{(0)}, 
$
while if $|\mu|\neq |\nu|$, this means
$
T_\nu T_\mu^*\otimes \unit_{\cct} \in \caN^{(1)}\lmk \unit_\caK\otimes \sigma_x\rmk
$.
From \eqref{ttt}, 
$\{T_{\mu}T_\nu^*\otimes \bbP_0(\sigma_x)\}_{\mu,\nu\in\caP}$
 are matrix units in $\caN \lmk \unit_\caK\otimes \bbP_0(\sigma_x)\rmk$
with $ e_{\mu\mu}T_{\mu}T_\nu^* e_{\nu\nu}\otimes\bbP_0(\sigma_x)=
T_{\mu}T_\nu^*\otimes \bbP_0(\sigma_x)$.
Then as in the proof of Proposition \ref{cuntzpropi},
there are $c_\mu\in \bbT$ such that 
$S_\mu S_\nu^*\otimes \bbP_0(\sigma_x)=E_{\mu\nu} \bbP_0(\unit_\caK\otimes\sigma_x)$
for 
$S_\mu=c_\mu T_\mu$.
Applying $\Ad_{\Gamma_\caK}$, we also 
obtain $S_\mu S_\nu^*\otimes \bbP_1(\sigma_x)=
(-1)^{|\mu|+|\nu|}E_{\mu\nu} \bbP_1(\unit_\caK\otimes\sigma_x)$, which 
then implies that
\begin{align}\label{unitdesu}
\big( S_\mu\otimes \sigma_x^{|\mu|} \big)
\big( S_\nu\otimes \sigma_x^{|\nu|} \big)^*
=S_\mu S_\nu^*\otimes \sigma_x^{|\mu|+|\nu|} 
=S_\mu S_\nu^*\otimes \big( \bbP_0(\sigma_x)+(-1)^{|\mu|+|\nu|}\bbP_1(\sigma_x) \big)
=E_{\mu\nu}.
\end{align}
It is clear that $\{S_\mu\}_{\mu\in\caP}$ are isometries satisfying \eqref{sss}.
The proof of \eqref{kansha} comes from an induction argument using
\eqref{sss} and \eqref{unitdesu}. As the argument is the same as in the proof of 
Proposition \ref{cuntzpropi}, we omit the details. 
Similarly, the proof that the isometries $\{S_\mu\}_{\mu\in\caP}$ are 
unique up to scalar multiplication in $\bbT$ is the same as in Proposition \ref{cuntzpropi}.
\end{proofof}
\begin{proofof}[Theorem \ref{cuntzthm1}]
The Hilbert space $\caK$, 
finite type I factor $\iota\circ\pi_\omega(\al_{\{0\}})$
with  matrix units
$\{\iota\circ\pi_\omega\circ \big( E_{\mu,\nu}^{(0)} \big) \}_{\mu,\nu\in\calP} 
\subset \calB(\caK) \otimes \mathfrak{C}$
and $\rho$ satisfy the conditions of Proposition \ref{cuntzpropii}.
Applying the proposition, we obtain 
$\sigma_0\in\{0,1\}$ and $\{S_\mu\}$ satisfying \eqref{bimpn} and \eqref{bbbon} 
from the statement of the theorem. 
The property \eqref{satn} follows from \eqref{bimpn} and 
parts (i) and (iii) of Lemma \ref{lem:trans_weak_conv_to_state}.
For the proof of \eqref{gvbn}, we set
 \begin{align}\label{gvbn1}
 T_\nu:=  (-1)^{\mqq(g)|\nu|}\sum_{\mu\in\calP} \overline{\langle \psi_\mu, \, \mathfrak{\Gamma}(U_g) \psi_\nu \rangle}^{\mpp(g)}
   \big( V_g^{(0)}\big)^* S_\mu V_g^{(0)}.
   \end{align}
     As in the proof of Theorem \ref{cuntzthm0}, we then can check that $T_\mu$ satisfies \eqref{buni2} for $E_{\mu\nu}$ replaced by
     $\iota\circ\pi_\omega(E_{\mu\nu}^{(0)})$.
     Applying the last statement of Proposition \ref{cuntzpropii}, there is some $c_g \in \mathbb{T}$ such that $S_\mu = c_g T_\mu$ for all $\mu\in\caP$.
     The proof of \eqref{gvbn} is given 
     by the same argument as in the proof of Theorem \ref{cuntzthm0}.
\end{proofof}
%
%
%
\section{Fermionic matrix product states}\label{fmpssec}

Using our results from Section \ref{transsec}, in this section 
we consider a translation invariant split state $\omega$ of $\al$ 
whose density matrices have uniformly bounded rank on finite intervals.
Our main result 
is that such states can be 
written as the thermodynamic limit of an even or odd fermionic matrix product state (MPS) 
depending on the value $\kappa_\omega \in \bbZ_2$. See~\cite{BWHV,KapustinfMPS} for 
the basic properties of fermionic MPS in the finite setting.
The idea of the proof is the same as quantum spin case, cf.~\cite{bjp, Matsui3, classA2}, 
although anti-commutativity results in richer structures.
We start with some preliminary results.

The following Lemma is immediate because each 
$\caA_{[0,N-1]}$ is isomorphic to a matrix algebra.
\begin{lem}
Let $\omega$ be a $\Theta$-invariant state of ${\al}$. 
For each $N\in\bbN$, let $Q_N$
be the support  projection of the density matrix of $\omega\vert_{\caA_{[0,N-1]}}$,
the restriction of $\omega$ to $\al_{[0,N-1]}$.
Then $Q_N$ is even.
\end{lem}
We consider the situation where the matrices $Q_N$ have uniformly bounded rank.
\begin{lem} \label{lem:intersection_of_kernel_fin_dim}
Let $\{Q_N\}$ be a sequence of orthogonal projections with
$Q_N\in \al_{[0,N-1]}^{(0)}$.
We suppose that the rank of $Q_N$ is uniformly bounded, i.e.,
$\sup_{N\in\bbN} \rank (Q_N)<\infty$.
Let $\pi$ be
 an irreducible representation of ${\al}_R$ or ${\al}_R^{(0)}$ on 
a Hilbert space $\calH$. Set 
$\calH_0 = \bigcap\limits_{N=1}^\infty\big( \pi(Q_N) \caH\big)$. 
Then $\mathrm{dim}\,\caH_0 < \infty$.
\end{lem}
\begin{proof}
As the statement is trivial if $\caH_0=\{0\}$,
assume that $\caH_0\neq\{0\}$.
We fix a unit vector $\eta\in\caH_0$ and let 
$\{\xi_j\}_{j=1}^l \subset \caH_0$ be an orthonormal system. We let $\mathfrak{A}$ denote 
either ${\al}_R$ or ${\al}_R^{(0)}$ with $\pi: \mathfrak{A} \to \calB(\caH)$ 
irreducible and let ${\mathfrak A}_{\rm loc}$ denote local elements in ${\mathfrak A}$. 
We similarly write  $\mathfrak{A}_{[0,N-1]}$ 
to denote either $\al_{[0,N-1]}$ or its even subalgebra. 
Note that the $l\times l$ matrix
$(\braket{\xi_i}{\xi_j})_{i,j=1,\ldots,l}$
is an identity.
Because $\pi$ is 
irreducible,  approximating $\xi_i$ with elements in
$\pi({\mathfrak A}_{\rm loc})\eta$,
there exists an $N\in \mathbb{N}$ and elements $a_{j,N} \in Q_N\mathfrak{A}_{[0,N-1]}Q_N$ 
such that for the $l\times l$-matrix 
$X_N = ( \langle \pi(a_{i,N})\eta, \pi(a_{j,N})\eta \rangle )_{i,j=1,\ldots,l}$,
 \begin{align}
   \big\| X_N - \unit_{\Mat_l}\big\| < \frac{1}{2}
 \end{align}
holds.

We now claim that $\{a_{j,N}\}_{j=1}^l$ are linearly independent within 
$Q_N\mathfrak{A}_{[0,N-1]} Q_N$. So we suppose that 
$\sum_j d_j a_{j,N} = 0$ for $\{d_j\}_{j=1}^l \subset \C$. Then 
taking the vector $d=(d_1,\ldots,d_l)$,
\begin{align*}
  \langle d, X_N d \rangle &= \sum_{i,j=1}^l \langle \pi(a_{i,N})\eta, \pi(a_{j,N}) \eta\rangle \, \ol{d_i} d_j 
  = \big\| \pi\big( \sum_{j=1}^l d_j a_{j,N} \big) \eta \big\|^2 = 0.
\end{align*}
Therefore 
$$
  0 = \langle d, X_N d \rangle =  \|d\|^2 + \langle d, (X_N - \one) d \rangle  
  \geq \|d\|^2 - \frac{1}{2} \|d\|^2 = \frac{1}{2} \|d\|^2
$$
and so $d=0$ and $\{a_{j,N}\}_{j=1}^l$ are linearly independent. 

By the assumption 
we have 
$\mathrm{dim}\big( Q_N \mathfrak{A}_{[0,N-1]}  Q_N \big) \leq C^2$,
for $C:=\sup_{N\in\bbN} \rank (Q_N)<\infty$.
This tells us that $l\leq C^2$ and so $\mathrm{dim}\, \caH_0 \leq C^2$.
\end{proof}
We now consider the case of even and odd fermionic MPS separately.
\subsection{Case: $\kappa_\omega=0$ (even fermionic MPS)}

\begin{thm}\label{thmevenmps}
Let $\omega$ be a pure, split, translation invariant and 
$\alpha$-invariant state on $\al$  with
index $\Ind(\omega)=(0,\mqq,[\upsilon])$.
For each $N\in\bbN$, let $Q_N$
be the support projection of the density matrix of $\omega\vert_{\caA_{[0,N-1]}}$ 
and assume 
$\sup_{N\in\bbN} \rank (Q_N)<\infty$.
Then there is some $m\in\bbN$, a faithful density matrix $D\in \Mat_m$, a self-adjoint unitary
$\mathfrak\Theta\in\Mat_m$ and a
set of matrices $\{v_\mu\}_{\mu\in\caP}$ in $\Mat_m$ 
satisfying the following.
\begin{enumerate}
\item[(i)] For all $x\in\Mat_m$,
$\lim_{N\to\infty}T_{\bf v}^N(x)=\Tr \lmk D x\rmk\unit_{\Mat_m}$ in the norm topology.
\item[(ii)] There is some $\sigma_0=0,1$ such that
 $\Ad_\mathfrak \Theta\lmk v_\mu \rmk
 =(-1)^{|\mu|+\sigma_0}v_\mu$ for all $\mu\in\caP$.
 \item[(iii)] $\Ad_{\mathfrak \Theta}(D)=D$.
\item[(iv)]  For any $l\in\bbN\cup\{0\}$, and $\mu_0,\ldots\mu_l, \nu_0,\ldots\nu_l\in \caP$,
  \begin{align}
  \omega\big( E_{\mu_0,\nu_0}^{(0)} E_{\mu_1,\nu_1}^{(1)}\cdots E_{\mu_l,\nu_l}^{(l)} \big) 
   =(-1)^{\sum\limits_{k=1}^l (|\mu_k|+|\nu_k|) \sum\limits_{j=0}^{k-1} |\nu_j| } 
   \Tr\big( D v_{\mu_0}\cdots v_{\mu_l} v_{\nu_l}^* \cdots v_{\nu_0}^* \big).
  \end{align}
\item[(v)] There is a projective unitary/anti-unitary representation $W$ on $\bbC^m$ 
relative to $\mpp$ and $c_g \in \mathbb{T}$ such that 
   \begin{align}
   \sum_{\mu\in\calP} \langle \psi_\mu, \, \mathfrak{\Gamma}(U_g) \psi_\nu \rangle v_\mu = 
     c_g W_g v_\nu W_g^*.
   \end{align}
   The second cohomology class associated to $W$ is $[\upsilon]$ and 
   \begin{align}
   \Ad_{W_g^*}(D)=D,\quad \Ad_{W_g}\lmk {\mathfrak \Theta}\rmk=(-1)^{\mqq(g)} \mathfrak \Theta, 
   \quad g\in G.
   \end{align}
\end{enumerate}
\end{thm}

\begin{rem}[Comparison with index for even fermionic MPS] \label{rem:evenfMPS_comparison}
Given an even fermionic MPS with on-site $G$-symmetry, $H^1(G, \mathbb{Z}_2) \times H^2(G, U(1)_\mpp)$-valued indices 
are defined in~\cite{BWHV, KapustinfMPS, TurzilloYou}. Briefly, an irreducible even fermionic MPS is specified by matrices 
$\{a_\mu\}_{\mu \in \calP} \subset \Mat_m$ spanning a simple  algebra that is $\Z_2$-graded by the adjoint action of a self-adjoint unitary 
$\mathfrak \Theta \in \Mat_m$. 
The on-site group action is 
given by $\Ad_{\tilde{W}_g}$ on the generators up to a $U(1)$-phase, where $\tilde{W}$ is a projective unitary/anti-unitary 
representation of $G$. The indices $(\tilde{\mqq}, [\tilde{\upsilon}])$ defined in~\cite{BWHV, KapustinfMPS, TurzilloYou} are given by the 
grading of the representation and its second cohomology class,
\[
   \Ad_{\tilde{W}_g}( \mathfrak \Theta ) = (-1)^{\tilde{\mqq}(g)} \mathfrak \Theta,   \qquad  \tilde{W}_g \tilde{W}_h = \tilde{\upsilon}(g,h) \tilde{W}_{gh},
\]
It is therefore clear from part (v) of Theorem \ref{thmevenmps} that the the indices $(\mqq,[\upsilon])$ defined 
for $\omega$ coincide with the indices defined from the corresponding fermionic MPS.
\end{rem}

To prove Theorem \ref{thmevenmps} we start with  a preparatory lemma.
\begin{lem} \label{lem:even_fmps_from_split_main}
Consider the setting of Theorem \ref{thmevenmps}.
Suppose that the graded $W^*$-$(G,\mpp)$-dynamical system
$(\pi_\omega(\al_R)'', \Ad_{\Gamma_\omega}, \hat\alpha_\omega)$ associated to $\omega$
is equivalent to $(\caR_{0,\caK},\Ad_{\Gamma_\caK},\Ad_{V_g})\in\caS_{0}$, via
a $*$-isomorphism $\iota: \pi_\omega(\al_R)''\to \caB(\caK\otimes\cct)$.
Then the following holds.
\begin{enumerate}
\item[(i)]
There is a finite rank density operator $D$ on $\calK \otimes\cct$ such that
\begin{align}
\mathrm{Ad}_{\Gamma_{\caK}}(D) = D,\quad\text{and}\quad
\Tr_{\caK\otimes\cct}\big( D (\iota\circ \pi_\omega(A)) \big) = \omega(A)
\end{align}
 for all $A \in {\al_R}$.
For $P_{\mathrm{Supp}(D)}$, the support 
  projection of $D$, $ \mathrm{Ad}_{\Gamma_{\caK}}\lmk P_{\mathrm{Supp}(D)}\rmk=P_{\mathrm{Supp}(D)}$.
  \item[(ii)] Let $\{B_\mu\}_{\mu \in \calP}$ be the set of isometries given in Theorem \ref{cuntzthm0}.
  Then we have
  \begin{align}\label{vbv}
  v_\mu:=P_{\mathrm{Supp}(D)} B_\mu=P_{\mathrm{Supp}(D)} B_\mu P_{\mathrm{Supp}(D)},\quad
  \mu\in\caP.
  \end{align}
\item[(iii)]$P_{\mathrm{Supp}(D)}V_g = V_g P_{\mathrm{Supp}(D)}$
and $D V_g=V_g D$  for any $g\in G$.
  \end{enumerate}
\end{lem}

\begin{proof}
(i)
Given the cyclic vector $\Omega_\omega$, $\langle \Omega_\omega, \iota^{-1}(x)\Omega_\omega\rangle$ 
defines a normal state on $\calB({\caK\otimes\cct})$. 
Let
$D$ be a density operator on $\caK\otimes\cct$
such that $\Tr_{{\caK\otimes\cct}}(D x) = \langle \Omega_\omega, \iota^{-1}(x)\Omega_\omega \rangle$. 
We then see that
$$
  \Tr_{{\caK\otimes\cct}} \big( D (\iota\circ \pi_\omega)(A) \big) 
  = \langle \Omega_\omega, \pi_\omega(A) \Omega_\omega \rangle  = \omega(A), \quad 
  A\in \al_R.
$$
Because $\omega\circ \Theta = \omega$ and 
$\iota\circ \pi_\omega \circ \Theta\vert_{\al_R}= 
\mathrm{Ad}_{\Gamma_\caK} \circ \iota \circ \pi_\omega\vert_{\al_R}$, it follows 
that $\Tr_{{\caK\otimes\cct}}( \mathrm{Ad}_{\Gamma_\caK} (D) (\iota\circ \pi_\omega)(A)) 
  = \Tr_{{\caK\otimes\cct}} ( D (\iota\circ \pi_\omega)(A))$ 
for all $A\in \al_R$. As such, $\mathrm{Ad}_{\Gamma_\caK} (D) = D$.
From this, we have $ \mathrm{Ad}_{\Gamma_{\caK}}\lmk P_{\mathrm{Supp}(D)}\rmk=P_{\mathrm{Supp}(D)}$.

 Let $\calH_0 = \bigcap\limits_{N=1}^\infty  \lmk \iota \circ \pi_\omega(Q_N)\rmk\lmk{\caK\otimes\cct}\rmk$. Because $\iota\circ\pi_\omega$ is an irreducible representation of $\al_R$, from Lemma
 \ref{lem:intersection_of_kernel_fin_dim}, 
$\caH_0$ is finite-dimensional.
Because $\omega(\unit-Q_N)=0$, we have
$ \Tr_{{\caK\otimes\cct}} \big( D (\iota\circ \pi_\omega)(\unit-Q_N) \big) 
= \omega(\unit-Q_N)=0$.
This means $P_{\mathrm{Supp}(D)}$, the support 
  projection of $D$, satisfies
  $P_{\mathrm{Supp}(D)}\le \iota\circ \pi_\omega(Q_N)$
for all $N\in\nan$.
Hence we have $P_{\mathrm{Supp}(D)}\lmk\caK\otimes\cct\rmk\subset\caH_0$.
Therefore $D$ is finite rank.

\vspace{0.1cm}

\noindent
(ii) Recall the endomorphism $\rho$ satisfying \eqref{ipr} from Lemma \ref{iorho}.
Because $\omega(A) = \omega(\beta_{S_1}(A))$ for all $A\in\al_R$, 
the set of isometries $\{B_\mu\}_{\mu \in \calP}$ given in Theorem \ref{cuntzthm0}
are such that
\begin{align*}
\Tr_{\caK\otimes\bbC^2}\big( D(\iota\circ\pi_\omega)(A)\big) 
  &=   \Tr_{\caK\otimes\bbC^2}\big( D( \rho \circ\iota\circ\pi_\omega)(A)\big) \\
  &= \sum_{\mu} \Tr_{\caK\otimes\bbC^2}\big( \Ad_{B_\mu^* }
\circ  \Ad_{ \Gamma_\caK^{|\mu|} }(D) (\iota\circ \pi_\omega)(A) \big) 
\end{align*}
for all $A\in\al_R$. This implies that 
$D=\sum_{\mu}\Ad_{B_\mu^* }
\circ  \Ad_{ \Gamma_\caK^{|\mu|} }(D)=\sum_{\mu}\Ad_{B_\mu^*}(D)$ 
and so 
$$
   \sum_\mu \lmk \unit-P_{\mathrm{Supp}(D)}\rmk B_\mu^* D B_\mu 
   \lmk \unit-P_{\mathrm{Supp}(D)}\rmk = 
  \lmk\unit- P_{\mathrm{Supp}(D)} \rmk D\lmk \unit- P_{\mathrm{Supp}(D)}\rmk=0.
$$
Hence we obtain 
$ P_{\mathrm{Supp}(D)}B_\mu    \lmk \unit-P_{\mathrm{Supp}(D)}\rmk
=0$.

\vspace{0.1cm}

\noindent (iii)
For an element $A\in\al_R$ and $\mpp(g) \in\bbZ_2$, we set $A^{\mpp(g)\ast}$
as $A$ if $\mpp(g)=0$ and $A^*$ if $\mpp(g)=1$.
Because 
$\omega( \alpha_g(A^{\mpp(g)\ast})) = \omega(A) = \Tr(D (\iota\circ\pi_\omega)(A))$, $A\in\al_R$,
we have that 
\begin{align*}
  \Tr_{\caK\otimes\cct}\big( D(\iota\circ\pi_\omega)(A) \big) 
  &= \Tr_{\caK\otimes\cct}\big( D( \iota\circ \pi_\omega) \lmk \alpha_g(A^{\mpp(g)\ast})\rmk \big) 
  = \Tr_{\caK\otimes\cct}\big( D V_g \big( (\iota\circ \pi_\omega)(A^{\mpp(g)\ast} ) \big) V_g^* \big).
\end{align*}
Given an orthonomal basis $\{\xi_j\}_j$ of $\caK\otimes\cct$, 
we see that 
for any $A\in \al_R$,
\begin{align*}
  \Tr_{\caK\otimes\cct}\big( D(\iota\circ\pi_\omega)(A) \big) 
  &=\Tr_{\caK\otimes\cct} \! \lmk
   DV_g \big( (\iota\circ \pi_\omega)(A^{\mpp(g)\ast}) \big) V_g^*
  \rmk
  = \sum_j \langle V_g \xi_j, D V_g (\iota\circ \pi_\omega)(A^{\mpp(g)\ast}) \xi_j \rangle \\
  &= \sum_j \ol{ \langle \xi_j, V_g^* D V_g(\iota\circ \pi_\omega)(A^{\mpp(g)\ast}) \xi_j \rangle}^{\mpp(g)} 
  = \Tr_{\calK\otimes \cct}\big( V_g^* D V_g (\iota\circ \pi_\omega)(A) \big),
\end{align*}
where for the second equality we used that 
$\{V_g \xi_j\}_j$ is an orthonomal basis of $\caK\otimes\cct$.
Therefore, $V_g^* D V_g = D$ and so $P_{\mathrm{Supp}(D)}V_g = V_g P_{\mathrm{Supp}(D)}$.
\end{proof}
\begin{proofof}[Theorem \ref{thmevenmps}]
We use the notation of Theorem \ref{cuntzthm0} and Lemma \ref{lem:even_fmps_from_split_main}.
Let $m\in\bbN$ be the rank of $D$ from  Lemma \ref{lem:even_fmps_from_split_main}.
We naturally identify $P_{\mathrm{Supp}(D)} \caB(\caK\otimes\cct)P_{\mathrm{Supp}(D)} $
and $\Mat_m$.
Then we may regard $D$ as a faithful density matrix in $\Mat_m$, and $\{v_\mu\}_{\mu\in\caP}$ matrices
in $\Mat_m$.
Because $\Gamma_{\caK}$ commutes with
$P_{\mathrm{Supp}(D)} $,  $\mathfrak\Theta:=\Gamma_{\caK}P_{\mathrm{Supp}(D)} $
defines a self-adjoint unitary in $\Mat_m$.
Similarly, because of (iii) of Lemma \ref{lem:even_fmps_from_split_main},
$W_g:=V_gP_{\mathrm{Supp}(D)}$ defines a projective unitary/anti-unitary representation
of $G$ on $P_{\mathrm{Supp}(D)} $ relative to $\mpp$.
Clearly, the second cohomology class associated to $W$ is the same of that of $V$, i.e., $[\upsilon]$.
From $\Ad_{V_g}(\Gamma_{\caK})=(-1)^{\mqq(g) }\Gamma_{\caK}$,
we have that $\Ad_{W_g}\lmk {\mathfrak \Theta}\rmk=(-1)^{\mqq(g)} \mathfrak \Theta$.

Now we check the properties (i)-(v).
\\Parts (ii) and (v) are immediate from the definition of
$v_\mu$, $\mathfrak \Theta$, $W_g$, and the corresponding properties of $B_\mu$, $\Gamma_\caK$,
$V_g$.
Part (iii) follows from Lemma \ref{lem:even_fmps_from_split_main} (i), (iii).
For part (i), using \eqref{sat}, \eqref{vbv} and that $P_{\mathrm{Supp}(D)}$ is of 
finite rank, we have
\begin{align}
T_{\bf v}^N(x)=
P_{\mathrm{Supp}(D)} \, T_{\bf B}^N(x) \, P_{\mathrm{Supp}(D)}
 \xrightarrow[N\to \infty]{}    \langle \Oo,\iota^{-1}(x)\Oo \rangle \,  P_{\mathrm{Supp}(D)}
 =\Tr_{\caK\otimes\cct}\big( D x \big) P_{\mathrm{Supp}(D)}
\end{align}
for $x\in P_{\mathrm{Supp}(D)}\caR_{0,\caK} P_{\mathrm{Supp}(D)}=\Mat_m$ 
and convergence in the norm topology.
For part (iv), \eqref{bbbo} and \eqref{vbv} imply that
\begin{align}
\omega\lmk E_{\mu_0,\nu_0}^{(0)} E_{\mu_1,\nu_1}^{(1)} 
\cdots E_{\mu_N,\mu_N}^{(N)}\rmk
&=\Tr_{\caK\otimes\cct} \lmk D
\lmk \iota\circ\pi_\omega\lmk E_{\mu_0,\nu_0}^{(0)} E_{\mu_1,\nu_1}^{(1)} 
\cdots E_{\mu_N,\mu_N}^{(N)}\rmk\rmk
\rmk\nonumber\\
& = 
  (-1)^{\sum\limits_{k=1}^N (|\mu_k|+|\nu_k|) \sum\limits_{j=0}^{k-1} |\nu_j|} 
  \Tr_{\caK\otimes\cct} \lmk D
  B_{\mu_0} \cdots B_{\mu_N} B_{\nu_N}^* \cdots B_{\nu_0}^*\rmk
  \nonumber\\
& = 
  (-1)^{\sum\limits_{k=1}^N (|\mu_k|+|\nu_k|) \sum\limits_{j=0}^{k-1} |\nu_j|} 
  \Tr_{\Mat_m} \lmk D
  v_{\mu_0} \cdots v_{\mu_N} v_{\nu_N}^* \cdots v_{\nu_0}^*\rmk
\end{align}
for all $N\in\bbN\cup\{0\}$ and $\mu_0,\ldots\mu_N,\nu_0,\ldots,\nu_N\in\caP$.
This proves (iv).
\end{proofof}
\subsection{Case: $\kappa_\omega=1$ (odd fermionic MPS)}

\begin{thm}\label{thmoddmps}
Let $\omega$ be a pure, split, translation invariant and 
$\alpha$-invariant state on $\al$  with
index $\Ind(\omega)=(1,\mqq,[\upsilon])$.
For each $N\in\bbN$, let $Q_N$
be the support projection of the density matrix of $\omega\vert_{\caA_{[0,N-1]}}$ 
and assume 
$\sup_{N\in\bbN} \rank (Q_N)<\infty$.
Then there is some $m\in\bbN$, a faithful density matrix $D\in \Mat_m$, a set of matrices 
$\{v_\mu\}_{\mu\in\caP}$ in $\Mat_m$ and $\sigma_0\in\{0,1\}$
satisfying the following.
\begin{enumerate}
\item[(i)] 
Set $\hat v_\mu:=v_\mu\otimes \sigma_z^{\sigma_0+|\mu|}$
on $\bbC^m\otimes \cct$. 
Then $\lim_{N\to\infty}T_{\bf \hat v}^N(b)=\Tr \lmk \lmk D\otimes \frac 12\unit_{\cct} \rmk b\rmk\unit_{\Mat_m}\otimes
\unit_{\cct}$ in norm 
for all $b\in\Mat_m\otimes \Clf$.

\item[(ii)]  For any $l\in\bbN\cup\{0\}$, and $\mu_0,\ldots\mu_l, \nu_0,\ldots\nu_l\in \caP$,
  \begin{align}
&  \omega\big( E_{\mu_0,\nu_0}^{(0)} E_{\mu_1,\nu_1}^{(1)}\cdots E_{\mu_l,\nu_l}^{(l)} \big) \nonumber\\
&\hspace{1.5cm}  =(-1)^{\sum\limits_{k=1}^l (|\mu_k|+|\nu_k|)
\sum\limits_{j=0}^{k-1} \lmk \sigma_0+|\nu_j|\rmk } \delta_{\sum_{i=0}^l (|\mu_i|+|\nu_i|),\,0} 
   \Tr\lmk D
   \lmk v_{\mu_0}\cdots v_{\mu_l} v_{\nu_l}^* \cdots v_{\nu_0}^* \rmk
   \rmk.
  \end{align}
\item[(iii)] There a projective unitary/anti-unitary representation $W$ of $G$ on $\bbC^m$ 
relative to $\mpp$ and 
  $c_g \in \mathbb{T}$ such that for all $g\in G$ and $\nu\in\caP$
   \begin{align} \label{eq:odd_fMPS_group_action}
  (-1)^{\mqq(g)|\nu|} \sum_{\mu\in\calP} \langle \psi_\mu, \, \mathfrak{\Gamma}(U_g) \psi_\nu \rangle v_\mu = 
     c_g W_g v_\nu W_g^*, 
   \qquad \Ad_{W_g}(D)=D.
   \end{align}
   The second cohomology class associated to $W$ is $[\upsilon]$.
\end{enumerate}
\end{thm}

\begin{rem}[Comparison with index for fermionic MPS]
Like Remark \ref{rem:evenfMPS_comparison}, we briefly compare  our results with 
the $H^1(G, \Z_2)\times H^2(G, U(1)_\mpp)$-valued indices for fermionic MPS in~\cite{BWHV, KapustinfMPS, TurzilloYou}. 
An irreducible odd fermionic MPS is specified by matrices spanning a 
simple $\Z_2$-graded algebra with an odd central element. Like the even case, the group action is implemented by the 
adjoint action on generators by a projective unitary/anti-unitary representation, giving a second cohomology class.
The representation will commute or anti-commute with the odd central element, giving a homomorphism $G\to \Z_2$. 
Considering $\omega$ as an fermionic MPS, part (iii) of Theorem \ref{thmoddmps} shows that the second cohomology 
classes coincide and \eqref{eq:odd_fMPS_group_action} shows that the commutation of the 
projective unitary/anti-unitary representation with the odd central element is specified by $\mqq$. Hence 
in this setting the indices for fermionic MPS agree  with the indices defined in Section \ref{indsec}.
\end{rem}

\begin{lem} \label{lem:odd_fmps_from_split_main}
Consider the setting of Theorem \ref{thmoddmps}.
Suppose that the graded $W^*$-$(G,\mpp)$-dynamical system
$(\pi_\omega(\al_R)'', \Ad_{\Gamma_\omega}, \hat\alpha_\omega)$ associated to $\omega$
is equivalent to $(\caR_{1,\caK},\Ad_{\Gamma_\caK},\Ad_{V_g})\in\caS_{1}$, via
a $*$-isomorphism $\iota: \pi_\omega(\al_R)''\to \caR_{1,\caK}$.
Then the following holds.
\begin{enumerate}
\item[(i)]
There is a finite rank density operator $D$ on $\calK$ such that for all $A \in {\al_R}$,
\begin{align}\label{irg}
\Tr_{\caK\otimes\cct}\lmk
 \big( D\otimes \tfrac{1}{2} \unit_{\cct} \big) \big(\iota\circ \pi_\omega(A) \big) \rmk = \omega(A).
\end{align}
  \item[(ii)] Let $\{S_\mu\}_{\mu \in \calP}$ be the set of isometries given in Theorem \ref{cuntzthm1}.
  Then we have
  \begin{align}\label{vbvn}
  v_\mu:=P_{\mathrm{Supp}(D)} S_\mu=P_{\mathrm{Supp}(D)} S_\mu P_{\mathrm{Supp}(D)},\quad
  \mu\in\caP.
  \end{align}
\item[(iii)]$P_{\mathrm{Supp}(D)}V_g^{(0)} = V_g^{(0)} P_{\mathrm{Supp}(D)}$ 
and $\Ad_{V_g^{(0)}}(D)=D$ for any $g\in G$.
  \end{enumerate}
\end{lem}

\begin{proof}
(i)
Given the cyclic vector $\Omega_\omega$, $\langle \Omega_\omega, \iota^{-1}(x)\Omega_\omega\rangle$,
$x\in \caR_{1,\caK}$,
defines a normal state on $\caR_{1,\caK}$. 
Let
$\tilde  D$ be a density operator on $\caK\otimes\cct$
such that $\Tr_{{\caK\otimes\cct}}(\tilde D x) = \langle \Omega_\omega, \iota^{-1}(x)\Omega_\omega \rangle$ for $x\in \caR_{1,\caK}$. 
Because $\caR_{1,\caK}=\bk\otimes \Clf$ and recalling the notation 
$\bbP_\varepsilon$ from \eqref{eq:pm_espace_proj}, we may assume that 
$\tilde D$ is of the form $\tilde D=D_0\otimes \bbP_0(\sigma_x)+D_1\otimes \bbP_1(\sigma_x)$.
Because $\omega\circ \Theta = \omega$ and 
$\iota\circ \pi_\omega \circ \Theta\vert_{\al_R} = \mathrm{Ad}_{\Gamma_\caK} \circ \iota \circ \pi_\omega\vert_{\al_R}$, it follows 
that $\Tr_{{\caK\otimes\cct}}( \mathrm{Ad}_{\Gamma_\caK} (\tilde D) (\iota\circ \pi_\omega)(A)) 
  = \Tr_{{\caK\otimes\cct}} ( \tilde  D(\iota\circ \pi_\omega)(A))$ 
for all $A\in \al_R$. 
Therefore, we have $\mathrm{Ad}_{\Gamma_\caK} (\tilde D)= \tilde  D$, which implies
$D_0=D_1$. We set $D:=2 D_0$, and see that  $D$ is a density operator
on $\caK$ satisfying \eqref{irg}.

Let $\pi_0$ be the irreducible representation of $\al_R^{(0)}$ on $\caK$
given by
\begin{align}
\iota\circ\pi_\omega(a)=\pi_0(a)\otimes\unit_{\cct},\quad a\in \al_R^{(0)}.
\end{align}
 Let $\calH_0 = \bigcap\limits_{N=1}^\infty  \lmk \pi_0(Q_N)
 \caK\rmk$. Because $\pi_0$ is an irreducible representation of $\al_R^{(0)}$, 
 $\caH_0$ is finite-dimensional by Lemma \ref{lem:intersection_of_kernel_fin_dim}.
Because $\omega(\unit-Q_N)=0$, we have
$$
\Tr_{{\caK\otimes\cct}}\big( 
 (D\otimes \tfrac{1}{2} \unit_{\cct} ) ( \pi_0(\unit-Q_N) \otimes \unit_{\cct}) \big) 
= \omega(\unit-Q_N)=0.
$$
This means $P_{\mathrm{Supp}(D)}$ satisfies
  $P_{\mathrm{Supp}(D)}\le \pi_0(Q_N)$
for all $N\in\nan$.
Hence we have $P_{\mathrm{Supp}(D)}\caK\subset\caH_0$ and $D$ is finite rank.

\vspace{0.1cm}

\noindent
(ii) Recall the endomorphism $\rho$ satisfying \eqref{ipr} from Lemma \ref{iorho}.
Because $\omega(A) = \omega(\beta_{S_1}(A))$ for all $A\in\al_R$, 
the set of isometries $\{S_\mu\}_{\mu \in \calP}$ given in Theorem \ref{cuntzthm1} 
and $\sigma_0$,  \eqref{bimpn} gives that 
\begin{align}
 \Tr_{\caK\otimes\bbC^2}\lmk \lmk
 D\otimes \tfrac 12\unit_{\cct}\rmk(\iota\circ\pi_\omega)(A)\rmk
 &=\Tr_{\caK\otimes\bbC^2}\lmk \lmk
 D\otimes\tfrac 12\unit_{\cct}\rmk(\rho\circ \iota\circ\pi_\omega)(A)\rmk\nonumber\\
  & = \sum_{\mu} \Tr_{\caK\otimes\bbC^2}\lmk 
  \Ad_{( S_\mu^* \otimes \sigma_z^{\sigma_0+|\mu|} )}
\lmk
 D\otimes\tfrac 12\unit_{\cct}\rmk (\iota\circ \pi_\omega)(A) \rmk,
\end{align}
which implies that  $D=\sum_{\mu}\Ad_{S_\mu^*}(D)$.
We then obtain \eqref{vbvn} by the same 
proof as in Lemma \ref{lem:even_fmps_from_split_main}.

\vspace{0.1cm}

\noindent (iii)
By the same argument as in the proof of  Lemma \ref{lem:even_fmps_from_split_main}, 
we obtain
 $(V_g^{(0)})^* D V_g^{(0)} = D$ and so $P_{\mathrm{Supp}(D)}V_g^{(0)} = V_g^{(0)} P_{\mathrm{Supp}(D)}$.
\end{proof}
\begin{proofof}[Theorem \ref{thmoddmps}]
We use the notation of Theorem \ref{cuntzthm1}, and Lemma \ref{lem:odd_fmps_from_split_main}.
Let $m\in\bbN$ be the rank of $D$ from  Lemma \ref{lem:odd_fmps_from_split_main}.
We naturally identify $P_{\mathrm{Supp}(D)} \caB(\caK)P_{\mathrm{Supp}(D)} $
and $\Mat_m$.
Then we may regard $D$ as a faithful density matrix in $\Mat_m$, and $\{v_\mu\}_{\mu\in\caP}$ matrices
in $\Mat_m$.
Because of part (iii) of Lemma \ref{lem:odd_fmps_from_split_main},
$W_g:=V_g^{(0)}P_{\mathrm{Supp}(D)}$ defines a projective unitary/anti-unitary 
representation of
 $G$ on $P_{\mathrm{Supp}(D)} \caK$ relative to $\mpp$ 
whose cohomology class is the same as $V^{(0)}$, i.e. $[\upsilon]$.
Now we check the properties (i)-(iii) of Theorem \ref{thmoddmps}.

Part (iii) is immediate from the definition of
$v_\mu$, $W_g$, and the corresponding properties of $S_\mu$ and 
$V_g^{(0)}$.

For part (i), using \eqref{satn}, \eqref{vbvn} and that $P_{\mathrm{Supp}(D)}$ is finite rank, we have
\begin{align*}
T_{\bf \hat v}^N(x)&=
P_{\mathrm{Supp}(D)}T_{\bf \hat S}^N(x) P_{\mathrm{Supp}(D)}
\xrightarrow[N\to\infty]{}  \braket{\Oo}{ \iota^{-1}(x)\Oo} P_{\mathrm{Supp}(D)}=
\Tr \lmk \lmk D\otimes \tfrac 12\unit_{\cct} \rmk x\rmk P_{\mathrm{Supp}(D)}
\end{align*}
for 
$x\in \lmk P_{\mathrm{Supp}(D)}\otimes\unit\rmk
\caR_{1, \caK}\lmk P_{\mathrm{Supp}(D)}\otimes\unit\rmk=\Mat_m\otimes\Clf$
and convergence in the norm topology.

Part (ii) follows from \eqref{bbbon} and \eqref{vbvn},
as in the proof of Theorem  \ref{thmevenmps}.
\end{proofof}
\section*{Acknowledgements}
Y.O. would like to thank Y. Kubota and T. Matsui for discussion.
The present work was supported by JSPS Grants-in-Aid for Scientific Research no.~16K05171 and 19K03534 (Y.O.) 
and 19K14548 (C.B.). 
It was also supported by  JST CREST Grant Number JPMJCR19T2 (Y.O.).
%
%
%
\appendix
\section{Graded von Neumann algebras} \label{sec:GradedvN}

For convenience, we collect some facts about graded von Neumann algebras 
and linear/anti-linear group actions. See Section \ref{Subsec:graded_prelims} 
and \ref{subsec:graded_product_def}
for basic definitions.

\begin{lem}\label{gradei}
Let $(\caM, \theta)$ be a balanced graded von Neumann algebra. 
Assume that $\caM$ is of type $\mu$ and $\caM^{(0)}$ is of type $\lambda$,
with some $\mu,\lambda=\mathrm{I,II,III}$, and that
both of $\caM$ and $\caM^{(0)}$ have finite-dimensional centers.
Then $\lambda=\mu$.
\end{lem}
\begin{proof}
Let $U\in \caM^{(1)}$ be a self-adjoint unitary.
Let $\bbE:\caM\to\caM^{(0)}$ be the conditional expectation
\begin{align}\label{defe}
\bbE(x):=\frac 12(x+\theta(x)),\quad x\in \caM.
\end{align}
If $\caM^{(0)}$ has a faithful normal semifinite trace $\tau_0$ (i.e., $\caM^{(0)}$ is semifinite),
then $\tau:=(\tau_0+\tau_0\circ\Ad_{U})\circ\bbE$ defines a faithful normal semifinite trace on $\caM$.
Hence if $\caM^{(0)}$ is semifinite, then $\caM$ is semifinite.

Let us denote by $\caP(\caM), \caP(\caM^{(0)})$, the set of all orthogonal projections in
$\caM, \caM^{(0)}$.
As $\tau\vert_{\caM^{(0)}}$ is a faithful normal semifinite trace on $\caM^{(0)}$,
if $\lambda=\mathrm{II}$, then we have $\tau \lmk \caP(\caM^{(0)})\rmk=[0,\tau(1)]$.
As $\tau(\caP(\caM))$ contains $\tau\lmk \caP(\caM^{(0)})\rmk$ and 
$\caM$ is a finite direct sum of type $\mu$-factors, this means that 
$\mu=\mathrm{II}$.

If $\lambda=\mathrm{I}$, then there is a non-zero abelian projection $p$ of $\caM^{(0)}$.
We claim there is a non-zero abelian projection $r$ in $\caM$ such that
 $r\le p$. If $p\caM^{(1)}p=\{0\}$,
 then $p\caM p=\bbC p$ and $p$ itself is abelian in $\caM$.
 If  $p\caM^{(1)}p\neq\{0\}$, then there is a self-adjoint odd element
 $b\in \caM^{(1)}$ such that $pbp\neq 0$.
 Because
 $(pbp)^2=pbpbp\in p \caM^{(0)}p=\bbC p$,
 we may assume that $pbp$ is a non-zero self-adjoint unitary in  $p\caM p$.
 For any $x\in \caM^{(1)}$, we also have 
 $pxppbp\in p \caM^{(0)}p=\bbC p$. By the unitarity of $pbp$, we have
 $pxp\in \bbC pbp$, and 
 $p\caM^{(1)}p=\bbC pbp$.
 As $pbp$ is self-adjoint unitary,
 we have a spectral decomposition $pbp=r_+-r_-$, with mutually orthogonal projections
 $r_\pm$ in $\caM$  and at least one of $r_\pm$ is non-zero.
 From $p\caM^{(1)}p=\bbC pbp=\bbC(r_+-r_-)$ and $p\caM^{(0)}p=\bbC p$,
 $r_\pm$ are abelian in $\caM$ and $r_\pm\le p$, proving the claim.
 Hence $\caM$ is type I as well, $\mu=\mathrm{I}$.
 
Conversely,  if $\caM$ has a faithful normal semifinite trace $\tau$ (i.e., if $\caM$ is semifinite),
then $\tau\vert_{\caM^{(0)}}$ is a faithful normal semfinite trace on $\caM^{(0)}$.
Therefore, $\mu=\mathrm{III}$ if and only if $\lambda=\mathrm{III}$.

If $\mu=\mathrm{I}$, then $\lambda$ cannot be II or III and so is type I.
If $\mu=\mathrm{II}$, then $\lambda$ cannot be I or III and so is type II. 
\end{proof}
\begin{lem}\label{jh}
Let 
$(\caM,\theta)$ be a central graded  von Neumann algebra.
Then either $Z(\caM)=\bbC\unit$ 
or $Z(\caM)$
has a self-adjoint unitary $b\in Z(\caM)\cap \caM^{(1)}$ such that
\begin{align}\label{bba}
Z(\caM)\cap \caM^{(1)}=\bbC b.
\end{align}
\end{lem}
\begin{proof}
Let us assume that $\caM$ is not a factor.
By the condition of centrality, $Z(\caM)\cap \caM^{(0)}=\bbC \unit$,
there is a non-zero self-adjoint element $b\in  Z(\caM)\cap \caM^{(1)}$.
Because $b^2\in Z(\caM)\cap \caM^{(0)}=\bbC\unit$, we may assume that 
$b$ is unitary.
For any $x\in Z(\caM)\cap \caM^{(1)}$, $xb$ also belongs to $Z(\caM)\cap \caM^{(0)}=\bbC\unit$,
and by the unitarity of $b$, we obtain \eqref{bba}. 
\end{proof}
When $(\caM,\theta)$ is spatially graded, an analogous result holds for $\caM\cap \caM'\Gamma$.
\begin{lem}\label{jh2}
Let 
$(\caM,\Ad_{\Gamma})$ be a central graded  von Neumann algebra on $\caH$, spatially graded by
a self-adjoint untiary $\Gamma$.
Then the following holds.
\begin{enumerate}
\item[(i)]
If $\caM$ is not a factor, $\caM\cap \caM'\Gamma=\{0\}$.
\item[(ii)]
If $\caM\cap \caM'\Gamma\neq \{0\}$, then
there is a self-adjoint unitary $b\in \caM\cap \caM'\Gamma$ such that
$\caM\cap \caM'\Gamma=\bbC b$.
In particular, if $\Gamma\in \caM$, then
$\caM\cap \caM'\Gamma=\bbC \Gamma$.
\end{enumerate}
\end{lem}
\begin{proof}
(i)
If $\caM$ is not a factor,
from Lemma \ref{jh},
$Z(\caM)$
has a self-adjoint unitary $b\in Z(\caM)\cap \caM^{(1)}$ such that
$
Z(\caM)\cap \caM^{(1)}=\bbC b
$.
For any $a\in \caM\cap \caM'\Gamma$, we have
\begin{align}
ba=ab=a\Gamma\Gamma b\Gamma\Gamma
=a\Gamma\lmk -b\rmk\Gamma
=-(a\Gamma) b\Gamma
=-b(a\Gamma) \Gamma
=-ba.
\end{align}
The first equality is because $b\in Z(\caM)$, and the fifth equality is because
$a\Gamma\in\caM'$.
As $b$ is unitary, this means $a=0$.

\noindent
(ii) Note that for any $a,b\in \caM\cap \caM'\Gamma$,
$ab\in Z(\caM)$. From this observation and (i), the same proof 
as Lemma \ref{jh} gives the claim.
If $\Gamma\in \caM$, as $\Gamma=\unit\Gamma$,  we have
$\Gamma\in \caM\cap \caM'\Gamma$.
\end{proof}
Recall the graded tensor product product defined in Section \ref{subsec:graded_product_def}.
\begin{lem}\label{comgra}
For $i=1,2$, let $(\caM_i,\Ad_{\Gamma_i})$ be a graded von Neumann algebra
on $\caH_i$ spatially graded by a self-adjoint unitary $\Gamma_i$ on $\caH_i$.
Let $\caM_1{\hox}\caM_2$ be the graded tensor product
 of
 $(\caM_1,\caH_1,\Gamma_1)$ and $(\caM_2,\caH_2,\Gamma_2)$.
Then commutant of the graded tensor product $(\caM_1{\hox} \caM_2)'$ 
is generated by
\begin{align}\label{mmfc}
(\caM_1')^{(0)}\odot \caM_2',\qquad
(\caM_1')^{(1)}\odot \caM_2'\Gamma_2.
\end{align}
\end{lem}
\begin{proof}
The proof is given by a modification of corresponding result for ungraded tensor products.
Let $\caM:=\caM_1{\hox}\caM_2$ and
$\caN$ be a von Neumann algebra generated by \eqref{mmfc}.
We would like to show $\caN=\caM'$. 
A brief computation gives the inclusion $\caM\subset \caN'$.

We let $\sigma \in \{0,1\}$ and denote by $\caR^{h,(\sigma)}$ the set of all self-adjoint
elements with grading $\sigma$ in a graded von Neumann algebra $\caR$.
For a complex Hilbert space $\caK$ and
its real  subspace $\caV$,  $\caV_\bbR^\perp$ is the orthogonal complement
of $\caV$ in $\caK$ regarding $\caK$ as a real Hilbert space,
with respect to the inner product 
$\langle \cdot ,\cdot \rangle _{\bbR}:=\Re \langle \cdot ,\cdot \rangle$.

First we assume that $\caM_j$, $j=1,2$, has a cyclic vector $\Omega_j$ which is homogeneous in the sense that
$\Gamma_j\Omega_j=(-1)^{\epsilon_j}\Omega_j$
for some $\epsilon_j\in \{0,1\}$.

As $\Omega:=\Omega_1\otimes \Omega_2$
is cyclic for $\caM$ in $\caH_1\otimes \caH_2$,
to show $\caM'=\caN$, it suffices to show that
$\caM^h\Omega+i\caN^h\Omega$ is dense in $\caH_1\otimes\caH_2$ by~\cite[Chapter IV, Lemma 5.7]{Takesaki1}.
For $\sigma_j=0,1$, $j=1,2$, set
$\caL_{\sigma_j}^{(j)}:= (\unit+(-1)^{\sigma_j}\Gamma_j)\caH_j$, $j=1,2$.
Then by the cyclicity of $\Omega_j$, and $\Gamma_j\Omega_j=(-1)^{\epsilon_j}\Omega_j$,
$\caM_j^{(\sigma_j)}\Omega_j$ is a dense subspace of
$\caL_{\sigma_j+\epsilon_j}^{(j)}$.
We also note $(\caM_j')^{(\sigma_j)}\Omega_j\subset \caL_{\sigma_j+\epsilon_j}^{(j)}$.
By~\cite[Chapter IV, Lemma 5.7]{Takesaki1}, $i (\caM_j')^h\Omega_j$ is dense in 
$(\caM_j^h\Omega_j)_\bbR^{\perp}$.
Therefore, $i(\caM_j')^{h, (\sigma_j+\epsilon_j)}\Omega_j$ is dense in
$((\caM_j)^{h, (\sigma_j+\epsilon_j)}\Omega_j)_\bbR^{\perp}\cap \caL_{\sigma_j}^{(j)}$.
Set $Y_{\sigma_1}:=(\caM_1)^{h,(\sigma_1+\epsilon_1)}\Omega_1$
and $Z_{\sigma_2}:=(\caM_2)^{h, (\sigma_2+\epsilon_2)}\Omega_2$.
By the above observation, $i(\caM_1')^{h, (\sigma_1+\epsilon_1)}\Omega_1$ is dense in
$(Y_{\sigma_1})_\bbR^\perp\cap \caL_{\sigma_1}^{(1)}$ and
$i(\caM_2')^{h, (\sigma_2+\epsilon_2)}\Omega_2$ is dense in
$(Z_{\sigma_2})_\bbR^\perp\cap \caL_{\sigma_2}^{(2)}$.
Because
$Y_{\sigma_1}+iY_{\sigma_1}$ and $Z_{\sigma_2}+iZ_{\sigma_2}$
are dense in $\caL_{\sigma_1}^{(1)}$ and $\caL_{\sigma_2}^{(2)}$ 
respectively by~\cite[Chapter IV, Lemma 5.8]{Takesaki1}, 
$Y_{\sigma_1}\odot Z_{\sigma_2}+i((Y_{\sigma_1})_\bbR^{\perp}\cap \caL_{\sigma_1}^{(1)})\odot ((Z_{\sigma_2})_\bbR^{\perp}\cap \caL_{\sigma_2}^{(2)})$
is dense in $\caL_{\sigma_1}^{(1)}\otimes \caL_{\sigma_2}^{(2)}$.
Hence we conclude that 
\begin{align}\label{vden}
(\caM_1)^{h, (\sigma_1+\epsilon_1)}\Omega_1\odot
(\caM_2)^{h, (\sigma_2+\epsilon_2)}\Omega_2
+i (\caM_1')^{h, (\sigma_1+\epsilon_1)}\Omega_1\odot
(\caM_2')^{h, (\sigma_2+\epsilon_2)}\Omega_2=:\caV_{\sigma_1,\sigma_2}
\end{align}
is dense in $\caL_{\sigma_1}^{(1)}\otimes \caL_{\sigma_2}^{(2)}$.
Using the homogeneity of $\Omega_j$, $\Gamma_j\Omega_j=(-1)^{\epsilon_j}\Omega_j$
we can prove that
$
\caM^h\Omega+i\caN^h\Omega
$ includes
\begin{align}
\sum_{\sigma_1,\sigma_2=0,1} i^{(\sigma_1+\epsilon_1)(\sigma_2+\epsilon_2)}
\caV_{\sigma_1,\sigma_2}.
\end{align}
By the density of $\caV_{\sigma_1,\sigma_2}$ in $\caL_{\sigma_1}^{(1)}\otimes \caL_{\sigma_2}^{(2)}$,
$\caM^h\Omega+i\caN^h\Omega$ is dense in $\caH_1\otimes\caH_2$ and
this completes the proof for the case with cyclic vectors.

Now we drop the assumption of the existence of the cyclic vectors.
Let $\{E_a'\}_a$ be a family of mutually orthogonal projections
in $\caM_1'$
such that each
$E_a'$ is an orthogonal projection onto $\overline{\caM_1\xi_a}$,
with a homogeneous $\xi_a\in\caH_1$, and $\sum_a E_a'=\unit_{\caH_1}$.
Let $\{F_b'\}_b$ be a family of mutually orthogonal projections
in $\caM_2'$
such that each
$F_b'$ is an orthogonal projection onto $\overline{\caM_2\eta_b}$,
with a homogeneous $\eta_b\in\caH_2$, and $\sum_b F_b'=\unit_{\caH_2}$.
Note that because $\xi_a$, $\eta_b$ are homogeneous,
$E_a'$ and $F_b'$ are even with respect to $\Ad_{\Gamma_1}$, $\Ad_{\Gamma_2}$
respectively.
Because $E_a'$ and $F_b'$ are even, the argument in~\cite[Lemma 11.2.14]{KR}
shows that 
the central support of $E_a'\otimes F_b'\in\caN\subset \caM'$
with respect to $\caN$ and the central support of $E_a'\otimes F_b'\in\caN\subset \caM'$
with respect to $\caM'$ coincides.
We denote the common central support by $P_{a,b}$.
By the first part of the proof, we know that
$\lmk E_a'\otimes F_b'\rmk\caN\lmk E_a'\otimes F_b'\rmk
=\lmk E_a'\otimes F_b'\rmk\caM'\lmk E_a'\otimes F_b'\rmk$.
We also have $\sum_{a,b} E_a'\otimes F_b'=\unit_{\caH_1\otimes\caH_2}$.
Therefore, applying~\cite[Lemma 11.2.15]{KR}, we get
$\caN=\caM'$.
\end{proof}
\begin{lem}\label{oisoiso}
Let $(\caM_i,\Ad_{\Gamma_i})$ , $(\caN_i,\Ad_{W_i})$, $i=1,2$,
be spatially graded von Neumann algebras on $\caH_i$ and $\caK_i$ respectively,
with grading operators $\Gamma_i$ and $W_i$.
Let $\alpha_i: \caM_i\to\caN_i$, $i=1,2$ be graded $*$-isomorphisms.
Suppose that $\caM_2$ (hence $\caN_2$ as well)
is either balanced or trivially graded.
Let $\caM_1{\hox}\caM_2$ be the graded tensor product 
of $(\caM_1,\caH_1,\Gamma_1)$ and $(\caM_2,\caH_2,\Gamma_2)$.
Let $\caN_1{\hox}\caN_2$ be the graded tensor product 
of $(\caN_1,\caK_1,W_1)$ and $(\caN_2,\caK_2,W_2)$.
Then there exists a unique $*$-isomorphism
$\alpha_1{\hox}\alpha_2:\caM_1{\hox}\caM_2\to
\caN_1{\hox}\caN_2$ such that
\begin{align}\label{abab}
\lmk
\alpha_1{\hox}\alpha_2
\rmk
(a{\hox} b)
=\alpha_1(a) \hox \alpha_2(b),
\end{align}
for all $a\in \caM_1$ and homogeneous $b\in\caM_2$.
\end{lem}
\begin{proof}
As $\alpha_2^{(0)}:=\alpha_2\vert_{\caM_2^{(0)}}$ is a normal $*$-isomorphism
from $\caM_2^{(0)}$ onto $\caN_2^{(0)}$, by~\cite[Chapter IV, Corollary 5.3]{Takesaki1}
there is a unique $*$-isomorphism $\alpha^{(0)}$ from  
$\caM_1\otimes\caM_2^{(0)}$ onto $\caN_1\otimes\caN_2^{(0)}$
such that
\begin{align}
\alpha^{(0)}(a\otimes b)
=\alpha_1(a)\otimes\alpha_2(b),\quad a\in \caM_1,\quad b\in\caM_2^{(0)}.
\end{align}
If $\caM_2$ is trivially graded then we set $\alpha_1{\hox}\alpha_2:=\alpha^{(0)}$.
If $\caM_2$ is balanced,
let $U$ be a self-adjoint unitary element in $\caM_2^{(1)}$.
As we have $\caM
=\big( \caM_1\otimes \caM_2^{(0)}\big) \oplus \big(\caM_1\otimes \caM_2^{(0)} \big)
\lmk
\Gamma_1\otimes U
\rmk
$, we may define a linear map $\alpha_1{\hox}\alpha_2:\caM\to\caN$ by
\begin{align}
(\alpha_1{\hox}\alpha_2) (x+y(\Gamma_1\otimes U))
=\alpha^{(0)}(x)+\alpha^{(0)}(y)\lmk W_1\otimes \alpha_2(U)\rmk,\quad
x, y\in \caM_1\otimes \caM_2^{(0)}.
\end{align}
It is straightforward to check that $\alpha_1{\hox}\alpha_2$ is a normal $*$-homomorphism.
Similarly, we may define  a normal $*$-homomorphism 
$\lmk \alpha_1\rmk^{-1}{\hox}\lmk \alpha_2\rmk^{-1}:\caN\to\caM$, which turns out to be the inverse of  $\alpha_1{\hox}\alpha_2$.
Hence $\alpha_1{\hox}\alpha_2$ is a $*$-isomorphism
satisfying \eqref{abab}.
The uniqueness is trivial from \eqref{abab}.
\end{proof}
\begin{lem}\label{isoiso}
Let  $(\caM_i,\Ad_{\Gamma_i})$, $i=1,2$, 
be balanced and spatially graded von Neumann algebras on $\caH_i$
with a grading operator $\Gamma_i$.
Let $\caM_1{\hox}\caM_2$ be the graded tensor product
 of
 $(\caM_1,\caH_1,\Gamma_1)$ and $(\caM_2,\caH_2,\Gamma_2)$.
For any graded $*$-automorphism
$\beta_i$ on $\caM_i$ implemented by a unitary $V_i$ on $\caH_i$
satisfying
$V_i\Gamma_i=(-1)^{\nu_i}\Gamma_i V_i$, $\nu_i\in\{0,1\}$ for each $i=1,2$,  
the automorphism $\beta_1{\hox}\beta_2$ on $\caM_1{\hox}\caM_2$ defined in Lemma \ref{oisoiso}
satisfies
\begin{align}\label{aiai}
\lmk
\beta_1{\hox}\beta_2
\rmk
(a{\hox} b)
=\Ad_{(V_1\otimes V_2\Gamma_2^{\nu_1})}\lmk a{\hox} b\rmk,
\end{align}
for all $a\in \caM_1$ and homogeneous $b\in\caM_2$.
\end{lem}
\begin{proof}
We compute that 
\begin{align*}
\lmk
\beta_1{\hox}\beta_2
\rmk
(a{\hox} b)
&=\beta_1(a)\Gamma_1^{\partial b}\otimes  \beta_2(b)
=\Ad_{(V_1\otimes V_2)} \lmk
a\Gamma_1^{\partial b}(-1)^{\partial b\cdot \nu_1} \otimes b
\rmk \\
&=\Ad_{(V_1\otimes V_2)} \Ad_{(\unit\otimes \Gamma_2^{\nu_1})} 
\big( a\Gamma_1^{\partial b}\otimes b \big)
\end{align*}
from which \eqref{aiai} follows.
\end{proof}
We also consider anti-linear $\ast$-automorphisms.
\begin{lem}\label{aisoaiso}
Let  $(\caM_i,\Ad_{\Gamma_i})$, $i=1,2$, 
be balanced and spatially graded von Neumann algebras on $\caH_i$
with a grading operator $\Gamma_i$. Let
 $\caM_1{\hox}\caM_2$ be the graded tensor product of
 $(\caM_1,\caH_1,\Gamma_1)$ and $(\caM_2,\caH_2,\Gamma_2)$
Suppose that $\caM_i$
has a faithful normal representation $(\caK_i,\pi_i)$
with a self-adjoint unitary $W_i$ on $\caK_i$ satisfying 
$\Ad_{W_i}\circ\pi_i(x)=\pi_i\circ \Ad_{\Gamma_i}(x)$, $x\in\caM_i$
and a complex conjugation $\caC_i$ on $\caK_i$ satisfying 
$\Ad_{\caC_i}\lmk\pi_i(\caM_i)\rmk=\pi_i(\caM_i)$ 
and $\caC_iW_i=W_i\caC_i$, for $i=1,2$.
Then for any graded anti-linear $*$-automorphism
$\beta_i$ on $\caM_i$, $i=1,2$,
there exists a unique anti-linear $*$-automorphism
$\beta_1{\hox}\beta_2$ on $\caM_1{\hox} \caM_2$
such that
\begin{align}\label{abab3}
\lmk
\beta_1{\hox}\beta_2
\rmk
\lmk a{\hox} b\rmk
=\beta_1(a){\hox} \beta_2(b),
\end{align}
for all $a\in \caM_1$ and homogeneous $b\in\caM_2$.

If $\beta_i$ is  implemented by an anti-unitary $V_i$  on
$\caH_i$ satisfying
$V_i\Gamma_i=(-1)^{\nu_i}\Gamma_i V_i$, $\nu_i\in\{0,1\}$ for each $i=1,2$,  then
\begin{align}
\lmk
\beta_1{\hox}\beta_2
\rmk
\lmk a \hox b\rmk
=\Ad_{(V_1\otimes V_2\Gamma_2^{\nu_1})}\lmk a{\hox} b\rmk.
\end{align}
\end{lem}
\begin{proof}
Let $\pi_1(\caM_1){\hox}\pi_2(\caM_2)$ be the graded tensor product of the 
$(\pi_1(\caM_1),\caK_1,{W_1})$ and  $(\pi_2(\caM_2),\caK_2,{W_2})$.
By Lemma \ref{oisoiso}, there is a $*$-isomorphism 
$\pi:=\pi_1{\hox} \pi_2$ from $\caM_1{\hox}\caM_2$
onto $\pi_1(\caM_1){\hox}\pi_2(\caM_2)$
satisfying $\lmk \pi_1{\hox}\pi_2\rmk (a{\hox} b)=\pi_1(a){\hox} \pi_2(b)$
for $a\in\caM_1$ and homogeneous $b\in \caM_2$.
Because $\beta_i$, $\Ad_{\caC_i}$ and $\pi_i$ preserve the grading,
$\alpha_i:=\Ad_{\caC_i}\circ\pi_i\circ\beta_i\circ\pi_i^{-1}$ is a graded (linear) $*$-automorphism
on $\pi_i(\caM_i)$.
By Lemma \ref{oisoiso}, there is a $*$-automorphism 
$\alpha:=\alpha_1{\hox}\alpha_2$ on $\pi_1(\caM_1){\hox}\pi_2(\caM_2)$
such that $\lmk \alpha_1{\hox}\alpha_2\rmk (a{\hox} b)=\alpha_1(a){\hox} \alpha_2(b)$
for $a\in\pi_1(\caM_1)$ and homogeneous $b\in \pi_2(\caM_2)$.
Furthermore, for $\caC:=\caC_1\otimes\caC_2$,
$\Ad_{\caC}$ preserves $\pi_1(\caM_1){\hox}\pi_2(\caM_2)$.
Therefore, $\beta_1{\hox}\beta_1:=\pi^{-1}\circ\Ad_{\caC}\circ\alpha\circ\pi$
defines an anti-linear $*$-automorphism on $\caM_1{\hox} \caM_2$
and it satisfies \eqref{abab3}.

The proof for the second half of the lemma is the same as in 
Lemma \ref{isoiso}.
\end{proof}
\begin{lem}\label{ga8}
Let $G$ be a finite group and $\mpp: G\to\bbZ_2$ be a group homomorphism.
Let $(\caM_1,\Ad_{\Gamma_1},\alpha_1)$,  $(\caM_2,\Ad_{\Gamma_2},\alpha_2)$ 
be graded $W^*$-$(G,\mpp)$-dynamical systems such that, for $i=1,2$, $\calM_i$ is a balanced, central, 
spatially graded and type I von Neumann algebra with grading operator $\Gamma_i$.
Let $\caM_1{\hox}\caM_2$ be
the graded tensor product of $(\caM_1,\caH_1,\Gamma_1)$ and $(\caM_2,\caH_2,\Gamma_2)$.
Then for every $g\in G$, there exists a
linear $*$-automorphism (${\mpp(g)}=0$) or anti-linear automorphism (${\mpp(g)}=1$),
 $(\alpha_{1}{\hox}\alpha_{2})_g$ on $\caM_1{\hox}\caM_2$ such that
\begin{align}
\lmk \alpha_{1}{\hox}\alpha_{2}\rmk_g
(a{\hox} b)
=\alpha_{1,g}(a){\hox} \alpha_{2,g}(b),
\end{align}
for all  homogeneous $a\in \caM_1$ and $b\in\caM_2$.
\end{lem}
\begin{proof}
By Lemma \ref{casebycase},
there are graded $*$-isomorphisms $\iota_i: \caM_i\to \caR_{\kappa_i, \caK_i}$ 
with some\\
 $(\caR_{\kappa_i, \caK_i},\Ad_{\Gamma_{\caK_i}},\Ad_{V_{i,g}})\in\caS_{\kappa_i}$ for each
$i=1,2$.
Hence,  $(\caK_i\otimes \cct, \iota_i)$ is a faithful normal 
representation 
with a self-adjoint unitary $\Gamma_{\caK_i}$ implementing $\Ad_{\Gamma_i}$ on $\caK_i\otimes\cct$.
Let $C$ be a complex conjugation with respect to the standard basis of $\cct$
and $\caC_i$ be any complex conjugation on $\caK_i$.
Then $\caC_i\otimes C$ is a complex conjugation
on $\caK_i\otimes \cct$ commuting with
$\Gamma_{\caK_i}=\unit_{\caK_i}\otimes\sigma_z$,
preserving
$\caR_{\kappa_i, \caK_i} = \iota_i\lmk\caM_i\rmk$.
Hence we may apply Lemma \ref{oisoiso}  and Lemma \ref{aisoaiso}, which gives the result.
\end{proof}
%
%
%
%
\section{Lieb-Robinson bound for lattice fermion systems}\label{flr}
In this section, prove the Lieb-Robinson bound for one-dimensional lattice fermion systems. 
While this result is not new, see~\cite{BSP, NSY17}, our method of 
using an odd self-adjoint unitary to derive the Lieb-Robinson bound for 
odd elements from even elements is new.


The result holds for more general metric graphs, but to avoid the  
introduction of further notation, we restrict ourselves to one-dimensional case.
Let us recall the basic setting for Lieb-Robinson bound, 
see~\cite{BMNS,NSY17,NSY18} for details.
\begin{defn}\label{deff}
An $F$-function $F$ on $\bbZ$
is a non-increasing function $F:[0,\infty)\to (0,\infty)$
such that
\begin{enumerate}
\item[(i)]
$\lV F\rV:=\sup_{x\in\bbZ}\lmk \sum_{y\in\bbZ}F\lmk d(x,y)\rmk\rmk<\infty$,
and
\item[(ii)]
$C_{F}:=\sup_{x,y\in\bbZ}\lmk \sum_{z\in\bbZ}
\frac{F\lmk d(x,z)\rmk F\lmk d(z,y)\rmk}{F\lmk d(x,y)\rmk}\rmk<\infty$.
\end{enumerate}
\end{defn}

\begin{defn}
Let $F$ be an $F$-function on $\bbZ$, and $I$ an interval in $\bbR$.
We denote by $\caB_{F}^e(I)$ the set of all norm continuous paths of
even interactions on $\caA$ defined on an interval $I$
such that the function $\lV
\Phi
\rV_F: I\to \bbR$ defined by
\begin{align}\label{pnb}
\lV
\Phi
\rV_F(t):=
\sup_{x,y\in\bbZ}\frac{1}{F\lmk d(x,y)\rmk}\sum_{Z\in{{\mathfrak S}_{\bbZ}}, Z\ni x,y}
\lV\Phi(Z;t)\rV,\quad t\in I,
\end{align}
is uniformly bounded, i.e., 
$\sup_{t\in I}\lV \Phi \rV(t)<\infty$.
\end{defn}
For the rest of this Appendix, we fix some $\Phi\in \caB_{F}^e(I)$.
For each $s\in I$, we define a local Hamiltonian by \eqref{GenHamiltonian}.
We denote by $U_{\Lambda,\Phi}(t;s)$ the solution of 
\begin{align}
\frac{d}{dt} U_{\Lambda,\Phi}(t;s)=-iH_{\Lambda,\Phi}(t) U_{\Lambda,\Phi}(t;s),\quad t,s\in I,\quad
U_{\Lambda,\Phi}(s;s)=\unit.
\end{align}
We define the corresponding automorphisms $\tau_{t,s}^{(\Lambda),\Phi}$ on $\caA_{\bbZ}$ by
\begin{align}
\tau_{t,s}^{(\Lambda), \Phi}(A):=U_{\Lambda,\Phi}(t;s)^{*}AU_{\Lambda,\Phi}(t;s)
\end{align}
with $A \in \caA_\bbZ$. Note that ${\tau}_{s,t}^{(\Lambda), \Phi}$ is the inverse of
$\tau_{t,s}^{(\Lambda),\Phi}$.
Because $\Phi(s)$ is even, 
the proof of~\cite[Theorem 3.1]{NSY18}
gives the following.
\begin{lem}\label{lre}
Let $X,Y\in{\mathfrak S}_\bbZ$ with $X\cap Y=\emptyset$.
If either $A\in\al_{X}$ or $B\in\al_Y$ is even, then 
\begin{align}
\lV\left[
\tau_{t,s}^{(\Lambda), \Phi}(A), B
\right]
\rV
\le
\frac{2\lV A\rV\lV B\rV}{C_F} \lmk
e^{v|t-s|}-1
\rmk
D_0(X,Y),
\end{align}
where $v>0$ is some constant and 
\begin{align}
D_0(X,Y):=\sum_{x\in X}\sum_{y\in Y} F(|x-y|).
\end{align}
\end{lem}
Using this lemma and because $\Phi$ is even, the proof of~\cite[Theorem 3.4]{NSY18} 
guarantees the existence of the limit
\begin{align}
\tau_{t,s}^{\Phi}(A):=\lim_{\Lambda \nearrow\bbZ}\tau_{t,s}^{(\Lambda), \Phi}(A),\quad
A\in\caA, \quad t,s\in[0,1].
\end{align}
Clearly the limit dynamics $\tau_{t,s}^{\Phi}$ satisfy the 
same Lieb-Robinson bound as in Lemma \ref{lre}.
We would like to have an analogous bound as Lemma \ref{lre} 
for odd $A,B$.
To do this, fix an odd self-adjoint unitary $U_0\in\caA_{\{0\}}$.
For each $m\in\bbZ$, $\beta_{S_{m}}(U_0)$ is a self-adjoint unitary in $\caA_{\{m\}}$.
Define an interaction $\tilde\Phi_m(s)$ by
\begin{align}
\tilde\Phi_m(Z; s):=\Ad_{\beta_{S_{m}}(U_0)}\lmk
\Phi(Z; s)
\rmk,\quad Z\in{\mathfrak S}_{\bbZ},\quad s\in I,\quad m\in\bbN.
\end{align}
Note that $\tilde\Phi_m(Z; s)=\Phi(Z; s)$ if $Z$ does not include $m$.
Because $\tilde\Phi_m$ and $\Phi$ are even, Lemma \ref{lre} and
the proof of~\cite[Theorem 3.4]{NSY18}
implies the bound 
\begin{align}
\lV
\tau_{t,s}^{\Phi}(A)-\tau_{t,s}^{\tilde\Phi_m}(A)
\rV
&\le
\frac{4\lV A\rV}{C_F} \sum_{Z\ni m} 
\int_{[s,t]} dr \lV \Phi(Z; r)\rV D_0(X,Z)\lmk e^{v|t-r|}-1\rmk
\nonumber\\
&\le
{4\lV A\rV}
\int_{[s,t]} \lmk e^{v|t-r|}-1\rmk
\lV\Phi\rV_F(r)\sum_{x\in X} F(|x-m|)=:g(m), 
\end{align}
for any $A\in \al_X^{(1)}$, where the last inequality 
uses (i) and (ii) of Definition \ref{deff} as well as Equation \eqref{pnb}.
Note that $\lim_{m\to\infty} g(m)=0$.
Therefore, we have
\begin{align}\label{ulrb}
\lV
\left\{
\tau_{t,s}^{\Phi}(A), \beta_{S_{m}}(U_0)
\right\}
\rV
=\lV
\tau_{t,s}^{\Phi}(A)-\tau_{t,s}^{\tilde\Phi_m}(A)
\rV
\le g(m),
\end{align}
for any $A\in \al_X^{(1)}$ and $X\in{\mathfrak S}_{\bbZ}$ with $m \notin X$.
Let $X,Y\in{\mathfrak S}_\bbZ$ with $X\cap Y=\emptyset$, 
$A\in\al_X^{(1)}$, $B\in \al_Y^{(1)}$ and $m\notin X$.
Because $B\beta_{S_m}(U_0)\in \al_{Y\cup \{m\}}^{(0)}$,
Lemma \ref{lre} and \eqref{ulrb} implies
\begin{align}
\lV
\left\{
\tau_{t,s}^{\Phi}(A), B
\right\}
\rV
&=\lV
\left[
\tau_{t,s}^{\Phi}(A), B\beta_{S_m}(U_0)
\right] \beta_{S_m}(U_0)
+B\beta_{S_m}(U_0)
\left\{
\tau_{t,s}^{\Phi}(A), \beta_{S_{m}}(U_0)
\right\}
\rV\nonumber\\
&\le
\frac{2\lV A\rV\lV B\rV}{C_F} \lmk
e^{v|t-s|}-1
\rmk
D_0(X,Y\cup\{m\})
+g(m)\lV B\rV.
\end{align}
Taking the limit $m\to\infty$ and using 
Lemma \ref{lre}, we obtain the following.
\begin{lem}\label{lruni}
Let $X,Y\in{\mathfrak S}_\bbZ$ with $X\cap Y=\emptyset$.
For homogeneous $A\in\al_{X}$ and $B\in\al_Y$ we have
\begin{align}
\lV
\tau_{t,s}^{\Phi}(A) B-(-1)^{\partial A\partial B} B\tau_{t,s}^{\Phi}(A)
\rV
\le
\frac{2\lV A\rV\lV B\rV}{C_F} \lmk
e^{v|t-s|}-1
\rmk
D_0(X,Y).
\end{align}
\end{lem}
As in quantum spin systems, we can estimate the locality of the time evolved
observables from Lieb-Robinson bounds.
To do this, let $\{ \mathbb{E}_N : \mathcal{A} \to \mathcal{A}_{ \Lambda_N} ~|~ N\in \mathbb{N} \}$ 
be the  family of conditional expectations with respect to the trace on $\al$, 
see \cite{am}.
By the same argument as~\cite[Corollary 4.4]{NSY18},
if
 $A \in \mathcal{A}^{(0)}$ is such that 	
 \begin{equation}
		\begin{split}
		\big\| [A, B] \big\| \leq C \lVert B \rVert.
		\end{split}, 
	\end{equation}
	for all $B \in \bigcup_{\substack{X \in \mathfrak{S}_{\mathbb{Z}^\nu}\\ X \cap [-N,N]= \emptyset}} \mathcal{A}_X$,
then $\lVert A - \mathbb{E}_N (A) \rVert \leq C$. 
We extend this bound to odd elements.

Suppose that 
$A \in \mathcal{A}^{(1)}$ is such that 	\begin{equation}
		\begin{split}
		\big\| AB-(-1)^{\partial B} BA \big\| \leq C \lVert B \rVert.
		\end{split}
	\end{equation}
	for all homogeneous $B \in \bigcup_{\substack{X \in \mathfrak{S}_{\mathbb{Z}^\nu}\\ X \cap [-N,N]= \emptyset}} \mathcal{A}_X$.
Let $U_0\in \al_{\{0\}}^{(1)}$ be a self-adjoint unitary.
Then we have $AU_0\in \al^{(0)}$ and
\begin{align}
\big\| [AU_0, B] \big\|
=
\big\| \big( AB-(-1)^{\partial B} B A \big) U_0 \big\|
\le C \lV B\rV
\end{align}
for all homogeneous 
$B \in \bigcup_{\substack{X \in \mathfrak{S}_{\mathbb{Z}^\nu}\\ X \cap [-N,N]= \emptyset}} \mathcal{A}_X$.
Hence we have that
$\lV
[AU_0, B]
\rV\le 2C\lV B\rV$ for any 
 $B \in \bigcup_{\substack{X \in \mathfrak{S}_{\mathbb{Z}^\nu}\\ X \cap [-N,N]= \emptyset}} \mathcal{A}_X$.
Therefore, by the even case, 
we obtain that 
\begin{align}
\lV
A- \mathbb{E}_N (A)
\rV
=\lV\lmk
A- \mathbb{E}_N (A)
\rmk U_0
\rV
=
\lVert AU_0 - \mathbb{E}_N (AU_0) \rVert \leq 2C, 
\end{align}
where we used the fact that $U_0\in \al_{\Lambda_N}$.
From this and Lemma \ref{lruni},
we have shown the following.	
\begin{lem}\label{lruni2}
For any  $N\in\bbN$, $X\in{\mathfrak S}_\bbZ$ with $X\subset [-N,N]$ and
$A\in\al_{X}$, 
we have
\begin{align}
\big\|
\bbE_N\lmk
\tau_{t,s}^{\Phi}(A)
\rmk
-\tau_{t,s}^{\Phi}(A)
\big\|
\le
\frac{8\lV A\rV}{C_F} \lmk
e^{v|t-s|}-1
\rmk
D_0(X,[-N,N]^c).
\end{align}
\end{lem}
Having Lemma \ref{lruni} and Lemma \ref{lruni2} as input,
we can carry out all the arguments in~\cite[Theorem 1.3]{mo} 
and~\cite[Proposition 3.5]{OgataTRI}.

\end{document}